\documentclass[10pt,journal,compsoc]{IEEEtran}
\ifCLASSOPTIONcompsoc
  \usepackage[nocompress]{cite}
\else
  \usepackage{cite}
\fi
\ifCLASSINFOpdf
  \usepackage[pdftex]{graphicx}
  \graphicspath{{./}}
  \DeclareGraphicsExtensions{.pdf,.jpeg,.png}
\else
\fi
\usepackage{url}
\usepackage{amsmath}
\usepackage{amssymb}
\usepackage{booktabs}
\usepackage{enumitem}
\usepackage{bm}
\usepackage{epsfig}
\usepackage{float} 
\usepackage{makecell}
\usepackage{colortbl}
\usepackage{mathtools}
\usepackage{extarrows}
\usepackage{algorithmic}
\usepackage{balance}
\usepackage{multirow}
\usepackage[percent]{overpic}

\usepackage{multirow}
\usepackage{comment}

\usepackage[utf8]{inputenc} % allow utf-8 input
\usepackage[T1]{fontenc}    % use 8-bit T1 fonts
\usepackage{booktabs}       % professional-quality tables
\usepackage{amsfonts}       % blackboard math symbols
\usepackage{nicefrac}       % compact symbols for 1/2, etc.
\usepackage{microtype}      % microtypography
\usepackage{pifont}
\usepackage{wrapfig}
\usepackage{graphicx}
\usepackage{subfig}
\usepackage{adjustbox}
\usepackage[pagebackref=true,breaklinks=true,bookmarks=false,citecolor=green,linkcolor=blue,colorlinks]{hyperref}

\definecolor{cGrey}{RGB}{212,214,212}

\usepackage[dvipsnames]{xcolor}

\usepackage{amsthm}

\newtheorem{corollary}{Corollary}

\newtheorem{assumption}{Assumption}

\begin{document}
\title{$S^2$-Transformer for Mask-Aware Hyperspectral\\ Image Reconstruction}

\author{
        Jiamian~Wang,
        Kunpeng~Li,
        Yulun~Zhang,
        Xin~Yuan,~\IEEEmembership{Senior~Member,~IEEE},
        and~Zhiqiang~Tao
        % <-this % stops a space
\IEEEcompsocitemizethanks{
\IEEEcompsocthanksitem Jiamian Wang and Zhiqiang Tao are with the School of Information, Rochester Institute of Technology, Rochester,
NY 14623 USA. E-mail: \{jw4905, zxtics\}@g.rit.edu.
\IEEEcompsocthanksitem Kunpeng Li is with Meta GenAI, CA, USA. E-mail: kunpengli@meta.com
\IEEEcompsocthanksitem Yulun Zhang is with the Shanghai Jiao Tong University, Shanghai, China. Email: yulun100@gmail.com 
\IEEEcompsocthanksitem Xin Yuan is with Westlake University, Hangzhou, China. E-mail: xyuan@westlake.edu.cn 
}% <-this % stops an unwanted space
\thanks{(Corresponding authors: Kunpeng Li \& Zhiqiang Tao)}
\thanks{Manuscript updated March 12, 2025.}
}

\IEEEtitleabstractindextext{%
{
\begin{abstract}
Snapshot compressive imaging (SCI) surges as a novel way of capturing hyperspectral images. It operates an optical encoder to compress the 3D data into a 2D measurement and adopts a software decoder for the signal reconstruction. Recently, a representative SCI set-up of coded aperture snapshot compressive imager (CASSI) with Transformer reconstruction backend remarks high-fidelity sensing performance. However, dominant spatial and spectral attention designs show limitations in hyperspectral modeling. The spatial attention values  describe the inter-pixel correlation but overlook the across-spectra variation within each pixel. The spectral attention size is unscalable to the token spatial size and thus bottlenecks information allocation.  Besides, CASSI entangles the spatial and spectral information into a 2D measurement, placing a barrier for information disentanglement and modeling.  In addition, CASSI blocks the light with a physical binary mask, yielding the masked data loss. To tackle above challenges, we propose a spatial-spectral ($S^2$-) Transformer implemented by a paralleled attention design and a mask-aware learning strategy. {Firstly, we systematically explore pros and cons of different spatial (-spectral) attention designs}, based on which we find performing both attentions in parallel well disentangles and models the blended information.   Secondly, the masked pixels induce higher prediction difficulty and should be treated differently from unmasked ones. We adaptively prioritize the loss penalty attributing to the mask structure by referring to the mask-encoded prediction as an uncertainty estimator. We theoretically discuss the distinct convergence tendencies between masked/unmasked regions of the proposed learning strategy. {Extensive experiments demonstrate that on average, the results of the proposed method are superior over the state-of-the-art methods.}
We empirically visualize and reason the behaviour of spatial and spectral attentions, and comprehensively examine the impact of the mask-aware learning, both of which advances the physics-driven deep network design for the reconstruction with CASSI.
Code is available at \url{https://github.com/Jiamian-Wang/S2-transformer-HSI}.
\end{abstract}}

\begin{IEEEkeywords}
Snapshot Compressive Imaging (SCI), Hyperspectral Image Reconstruction, Coded Aperture Snapshot Spectral Imaging (CASSI), Transformer, Interpretability.
\end{IEEEkeywords}}

\maketitle

\IEEEdisplaynontitleabstractindextext

% \fi
\IEEEpeerreviewmaketitle

\IEEEraisesectionheading{\section{Introduction}\label{sec:introduction}}
\IEEEPARstart{H}yperspectral {images record pixels of the scene across a wide range of spectrum, which enable not only high spatial granularity~\cite{suo2023recent,wen2021change,zheng2020foreground}, but also fine wavelength resolutions~\cite{jia2021survey,wang2022comprehensive,herold2003spectral,xu2021multiple,xu2023luojia}}. Due to their expressiveness, hyperspectral images are widely used in applications of biomedicine, remote sensing, astronomy~\cite{fu2018hyperspectral,fu2019hyperspectral,lu2014medical,suo2021computational}, \emph{etc}. For hyperspectral image capturing, the technology of snapshot compressive imaging (SCI) has come to the forefront in recent years, among which the  coded aperture snapshot compressing imager (CASSI)~\cite{gehm2007single,wagadarikar2008single} serves as a prevailing setup. It operates as an optical encoder, which compresses the hyperspectral signals into \textbf{2D} measurements and then requires a software decoder for the signal retrieval~\cite{yuan2021snapshot}. Recently, how to perform high-fidelity reconstruction draws lots of research attention~\cite{bioucas2007new,liu2018rank,ma2019deep,Bioucas-Dias2007TwIST,figueiredo2007gradient,wang2016adaptive,Meng20ECCV_TSAnet,Wang_2020_CVPR,zheng2021deep}, among which, Transformer-based models achieve the best performance~\cite{cai2022mask,cai2022coarse}.

{However, dominant multi-head self-attention designs (spatial-, spectral-) in Transformers show limitations in CASSI reconstruction, respectively. Firstly, spatial attention assigns each pixel with a \textbf{3D} token ($1\times1\times N$), which collects the information from all spectra. The \textbf{2D} attention matrix can only describe the inter-pixel (token) correlation, but fails to reflect the spectral variations within tokens. Despite the fact that the multi-head mechanism splits channels into groups and empowers the representation diversity, it is inflexible and fails to adapt to the complex spectral variation in practice. As shown in Fig.~\ref{fig: coverfig}, content from two different wavelengths ($498.0$nm and $614.4$nm) is similar in \texttt{region 1}, but drastically varies in \texttt{region 2}. Secondly, spectral attention takes each channel as a token and computes a \textbf{2D} attention matrix. Notably, no matter the spatial size of the tokens, the pairwise token correlation remains a scalar value (as highlighted in the right of Fig.~\ref{fig: coverfig}), leading to the dimensionality of the attention matrix remaining invariant to the token resolution. This property implies a bottleneck in allocating visual details in attention calculation. Finally, different from the traditional hyperspectral imaging tasks that directly handle the \textbf{3D} hyperspectral data, CASSI imposes a novel challenge by entangling the information from different wavelengths into a \textbf{2D} measurement (implemented by a disperser) as shown in Fig.~\ref{fig: cassi}. How to arrange the spatial and spectral attention for proper signal disentangling remains under-explored. }

{Besides the information entangling, we observe that CASSI blocks the light with a physical binary coded aperture (mask) for the signal encoding (as shown in Fig.~\ref{fig: cassi}), yielding a \textit{masked data loss}, which to the best knowledge, is generally overlooked by the literature -- the commonly-used MSE (RMSE) loss treats all of the pixels equally, hardly accounting for the different behaviours underlying masked and unmasked regions. To this end, we explicitly consider the \textit{masked data loss} from the data uncertainty perspective~\cite{kendall2017uncertainties}. Since the signals in masked areas are physically blocked, the data clues underlying the masked pixels are missed, which indicates a higher uncertainty of masked pixels than the unmasked ones. However, there lacks a way to distinguish this novel optical-induced uncertainty from the inherent texture-induced aleatoric uncertainty~\cite{kendall2017uncertainties}. This can be solved by approximating an encoded signal ($\mathbf{F}'$ in Fig.~\ref{fig: cassi}) as an  uncertainty estimator during the optimization on-the-fly.  }

\begin{figure}[tp]
\includegraphics[width=0.47\textwidth]{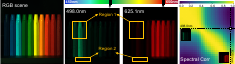}
\caption{ Hyperspectral image characteristics. Spatial correlations vary from adjacent wavelengths (\emph{e.g.}, \texttt{neighbor A}, \texttt{B}) to distant bands (\texttt{region 1}). The bottom-left correlation matrix upon multiple sampled wavelengths (\emph{i.e.}, 28$\times$28) defines a unique spectral distribution of a hyperspectral image. Spatial and spectral properties are inter-dependent.
}
\vspace{-5mm}
\label{fig: coverfig}
\end{figure}

{In this study, we propose a spatial-spectral ($S^2$-) Transformer with a mask-aware learning strategy. 
Firstly, we systematically investigate the standard spatial (spectral) multi-head self-attention designs along with their combinations under a unified and scalable network structure. We introduce a \textbf{parallel spatial-spectral attention} structure as our final proposal, based on which we empirically reveal the distinct functionalities of both attentions in the signal modeling procedure.
Secondly, we propose a mask-aware learning strategy to explicitly consider the optics-induced pixel-wise reconstruction difficulties, which assumes that pixels with higher ``uncertainty'' (regression difficulty) should be prioritized for the loss penalty~\cite{kendall2017uncertainties,ning2021uncertainty,sapkota2024distributionally}. Unlike the existing aleatoric uncertainty~\cite{lee2019gram,gilton2021deep}, we approximate the encoded signal on-the-fly and take it as an uncertainty estimator.
We then adaptively penalize the pixel-wise reconstruction weighted by the uncertainty for a fine-grained reconstruction. 
{We summarize the contributions as follows:}
{\begin{itemize}
    \item By observing the optical encoding procedure of SCI, this work identifies two challenges of \textit{information entanglement} and \textit{masked data loss} that impede a high-fidelity reconstruction, bringing new insights into designing SCI reconstruction networks from a physical perspective.
\item  This work systematically studies a group of spatial/spectral self-attention architectures to address the information disentanglement pertaining to SCI, based on which we propose an $S^2$-Transformer with parallel spatial-spectral attention. The proposed method, to our best knowledge, is the first to thoroughly investigate different attention behaviours across both spatial and spectral dimensions in reconstructing SCI data.
\item The proposed $S^2$-Transformer method introduces a mask-aware learning strategy to mitigate the \textit{masked data loss} rooted in the physical mask encoding process of SCI. Both theoretical analysis and empirical evidence demonstrate better reconstruction quality in masked regions, boosting the overall performance perceptually and quantitatively. Besides, experiments on diverse competitive reconstruction backbones suggest the benefit of mask-aware learning for SCI systems.  
\item Extensive experiments demonstrate that the proposed method not only outperforms the state-of-the-art SCI methods on simulated and real HSI data, but also shows potential 
in handling various noise cases, as well as geological remotely sensed data. Our observations and induction provide insights for the future Transformer architecture design in snapshot
compressive imaging.
\end{itemize}}}

\begin{figure}[t]
\includegraphics[width=0.47\textwidth]{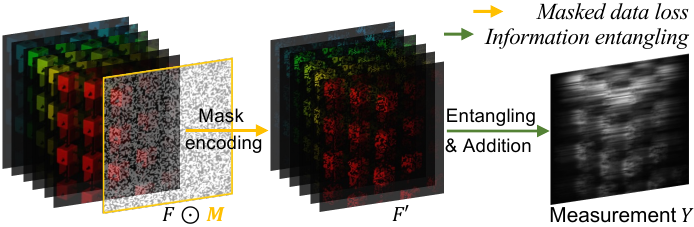}
\caption{CASSI pipeline (best viewed in color).}
\vspace{-3mm}
\label{fig: cassi}
\end{figure}

The rest of this paper is organized as follows. Section~\ref{sec: related works} introduces the related works in detail. Section~\ref{sec: method} develops the proposed $S^2$-Transformer and theoretically evaluates the mask-aware learning strategy. Extensive experiments of performance under different scenarios, ablation studies, and visual grounded evidence are demonstrated in Section~\ref{sec: experiment}. Section~\ref{sec: conclusion} concludes this paper. 

\section{Related Work}\label{sec: related works}
Inspired by the compressed sensing (CS) theories~\cite{candes2006robust,donoho2006compressed}, a series of HSI technologies such as occlusion mask~\cite{Cao11PAMI}, spatial light modulator~\cite{Yuan15JSTSP} and digital-micromirror-device~\cite{Wu11OL_DMD} have also been used for hyperspectral imaging. Some other systems use multiple-shots~\cite{kittle2010multiframe}, dual-channel~\cite{huang2021deep,wang2015dual,Wang15CVPR,Wang18TPAMI} and high-order information~\cite{arguello2013higher} to improve the performance.
In comparison, coded aperture snapshot
spectral imager (CASSI) uses a coded aperture and a prism to implement the spectral modulation, which serves as one of the most popular optical designs due to its concise setup, short acquisition time, low-cost consumption, and low-power usage~\cite{yuan2021snapshot}.

CASSI plays a role of optical encoder by compressively tweaking and storing the hyperspectral signal into 2D measurements. Correspondingly, it requires software-based algorithms for the signal retrieval. Existing works tackle this ill-posed reconstruction problem from different perspectives.  Regularization-based optimization algorithms~\cite{Bioucas-Dias2007TwIST,figueiredo2007gradient,wang2016adaptive} adopt various priors to confine the signal space, such as sparsity~\cite{wang2016adaptive}, and total variation (TV)~\cite{yuan2016generalized}, etc. Among them, DeSCI~\cite{liu2018rank} estimates the low-rank matrix structure of similar patch groups and leads to excellent performance. However, these methods generally can provide limited reconstruction performance and sometimes suffer from the unstable convergence and long running time.

Different from the regularization-based algorithms, deep unfolding methods ~\cite{Wang_2020_CVPR,wang2019hyperspectral,Ma19ICCV,fu2021bidirectional,zhang2022herosnet} employ the deep neural networks as denoisers in the widely-used optimization solver such as alternating direction method of multipliers (ADMM)~\cite{boyd2011distributed}.
Among them, the~\cite{wang2019hyperspectral} adopts a prior network to jointly learn the local coherence and dynamic characteristics of hyperspectral data, while~\cite{wang2020dnu} integrates the local and non-local correlation exploitation in the learnable module. 
Despite the strong interpretability and end-to-end learning property of the unfolding methods, the computational complexity and efficiency degrades with more stages of CNN models introduced. Another way of combining the optimization algorithm with deep networks is the Plug-and-Play (PnP) frameworks~\cite{chan2016plug,qiao2020snapshot,zheng2021deep}, where pretrained deep networks are incorporated as denoisers for a better efficiency, robustness, and performance trade-off. However, pretrained deep denoisers without re-training limits the performance of the framework. {Also, the proper determination of pretrained models remains underexplored.}

Inspired by strong representation ability of convolutional neural networks, 
deep CNN model-based methods~\cite{meng2020snapshot,Miao19ICCV,wang2019hyperspectral,Wang19TIP} have been introduced for a high-fidelity reconstruction. 
Among them, $\lambda$-net~\cite{Miao19ICCV} incorporates a generative adversarial learning framework. TSA-Net~\cite{Meng20ECCV_TSAnet} firstly introduces the self-attention in both spatial and spectral domain. However, it heavily relays on the spatial locality of convolution for the down-scaled attention map calculation, leading to limited long-range relationship modeling capacity. DGSMP~\cite{huang2021deep} formulates the reconstruction as a MAP problem and effectively learns the prior, achieving better performance.  
BIRNAT~\cite{cheng2022recurrent} adopts bidirectional RNNs for a follow-up channel retrieval, exploiting the sequential property of the hyperspectral data for the first time.   Besides, HDNet~\cite{hu2022hdnet} mainly regularizes the reconstruction in the frequency domain upon discrete Fourier transform (DFT). It performs the spatial-spectral learning with convolutional structure for a better feature extraction, but barely discloses the spatial and spectral characteristics of the hyperspectral images underlying the model design principles.
{More recently, Transformer architectures~\cite{vaswani2017attention,dosovitskiy2021an} becomes an emerging option for the hyperspectral image reconstruction. Specifically, MST~\cite{cai2022mask} models the long-range dependencies across different spectral channels. CST~\cite{cai2022coarse} selectively computes the self-attention within informative spatial windows and groups semantically-similar tokens into buckets regardless of the pixel spatial location, which sets the state-of-the-art performance. The proposed method differs from the above methods in certain aspects. Firstly, this work observes an information disentanglement and masked data loss from a physical scope. There lacks a systematic overview of the prevailing attention mechanisms for the task of SCI despite the novel Transformer design. The proposed method bridges this gap.}

There has been abundant work exploiting the power of spatial-spectral attention in hyperspectral imaging. For hyperspectral image classification, 
representative works~\cite{tang2020hyperspectral,li2020joint,peng2022spatial} abstract the determinative spatial-spectral features for a precise classification. In hyperspectral denoising, SQAD~\cite{pan2022sqad} computes the spatial and spectral correlations among adjacent regions (spectra), benefiting the intrinsic noise distribution modeling. In hyperspectral super-resolution, SSPSR~\cite{jiang2020learning} develops a spatial residual module and spectral attention residual module upon the CNN network for better representation learning. For snapshot compressive imaging (SCI), MST~\cite{cai2022mask} devises a spectral attention Transformer for the reconstruction. GAP-CCoT~\cite{wang2022snapshot} introduces spatial and channel attention in the unfolding framework. To our best knowledge, few works observe and analyze the different behaviour of the spatial and spectral attentions in processing the hyperspectral data. By comparison,  the proposed method not only exploits informative spatial-spectral clues in hyperspectral images, but also pursues better information disentangling upon the 2D measurement by systematically discussing different spatial-spectral attention designs. We provide extensive evidences to better explain the  different behaviours of spatial and spectral attention, encouraging the future endeavours.

Different from the prior efforts, in this work, our proposed method is inspired by observing the physical encoding process of the CASSI. We attribute the reconstruction difficulty into two types of data loss and present the solutions respectively. For the \textit{entangled data loss}, we systematically explore both spatial and spectral attentions under a unified structure. We uncover the distinct functionalities of spatial and spectral attention under the best-performed attention arrangement. {For the masked data loss, different from existing  works~\cite{cai2022mask,hu2022hdnet,wang2021new,wang2022modeling,wang2023cooperative} that mainly adopt $\ell_1$ or $\ell_2$ loss for reconstruction, we identify and take advantage of the optical-induced data uncertainty in a pixel-wise manner by adaptively emphasizing the masked pixels during training. We theoretically and empirically present the convergence tendencies of the proposed learning strategy upon both mask and unmasked regions. Notably, mask-aware learning brings consistent performance boost when applied on different reconstruction backbones. It is worth noting that adopting physical mask to encode signal widely exists in snapshot compressive imaging systems beyond CASSI, such as Hyperspectral compressive sensing imager~\cite{yang2015compressive}, random modulation imaging~\cite{gehm2007single},  CACTI~\cite{wang2022spatial,wu2023adaptive}. Besides, (2) mask encoding has been shown theoretically associated with the performance of SCI syetems~\cite{zhao2023theoretical}. Orthogonal to previous mask-oriented works, we hope this work can inspire future works from the perspective of the loss function design.}

Notably, the most recent hyperspectral reconstruction techniques are focusing on visible wavelengths such as 400 - 700 nm. This is mainly due to the limited available public data. For example, the real data that can be used to evaluate different algorithms for CASSI reconstruction was captured and released to the public in 2020~\cite{Meng20ECCV_TSAnet}. At this time, we believe that finding an efficient algorithm is important for the visible wavelengths.
Plus, from the optics perspective, it is practical and important to first design SCI systems in the visible bandwidth and then further explore SCI in more challenging light conditions.

\begin{figure*}[htp] 
\centering 
\includegraphics[width=0.96\textwidth]{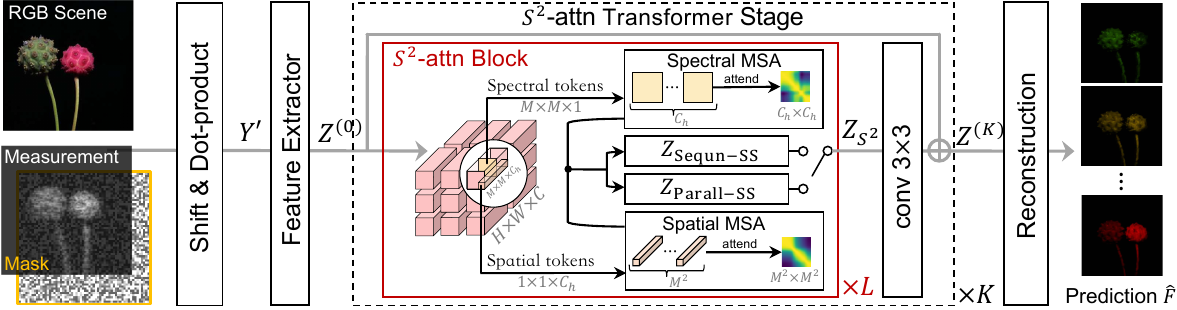} %\vspace{-2mm}
\caption{ Overview of the spatial-spectral ($S^2$-) Transformer. The network takes the 2D measurement $\mathbf{Y}$ with the mask $\mathbf{M}$ as input and retrieves the hyperspectral image $\widehat{\mathbf{F}}$. It mainly contains $K$ stages, where each consists of $L$ $S^2$-attn blocks. Both sequential ($\mathbf{Z}_{\texttt{Sequn-SS}}$) and paralleled ($\mathbf{Z}_{\texttt{Parall-SS}}$)  blocks employ spatial and spectral multi-head self-attention (MSA).
We perform window partition to the embedding before MSA. Visualization takes one window in a head as an example.}
\vspace{-3mm}
\label{fig: framework} 
\end{figure*}

\section{Method}\label{sec: method}
We firstly give a preliminary knowledge of hyperspectral imaging. We also uncover the potential challenges that may set the bottleneck for the reconstruction performance. Followed by, we propose an $S^2$-Transformer architecture with mask-aware learning strategy as a corresponding solution. 

\subsection{Preliminary Knowledge}\label{subsec: CASSI System} 
The CASSI-based hyperspectral imaging process consists of an optical encoder and a software decoder. Let $\mathbf{F}\in\mathbb{R}^{H \times W \times N_{\lambda}}$ represent the hyperspectral cube with $N_{\lambda}$ discrate wavelengths (spectral channels), $\mathbf{M}\in\mathbb{R}^{H \times W}$ denotes the physical coded aperture (mask). As shown in Fig.~\ref{fig: cassi}, the whole compression procedure could be simplified as two steps. CASSI firstly encodes the signal by
\begin{equation}\label{eq: CASSI mask modulate}
    \mathbf{F}' = \mathbf{F} \odot \mathbf{M},
\end{equation}
where $\odot$ denotes a pixel-wise multiplication 
with  broadcasting, 
and $\mathbf{F}'$ denotes a ``sponge'' cube encoded by the mask. Note that the visual information of certain pixels will be erased according to the mask pattern, \emph{i.e.}, where $\mathbf{M}_{ij}$$=$$0$, or partially disrupted by the noisy mask values, \emph{i.e.}, $\mathbf{M}_{ij}$$\in$$(0,1)$, both of which lead to the reconstruction challenge of \textit{masked data loss}.  {After that,  CASSI tiles $\mathbf{F}'$ by $y$-axis-shearing}, specifically, $\mathbf{F}'(u,v,n_{\lambda}) \rightarrow \mathbf{F}'(u,v+d(\lambda-\lambda^*), n_{\lambda})$, where $\lambda$ indicates the wavelength of the $n_{\lambda}$-th spectral channel. The $\lambda^*$ denotes the pre-defined anchor wavelength, and $d(\cdot)$ defines the shifting principle. We conduct a two-pixel shift for each spectral channel following~\cite{Meng20ECCV_TSAnet}.  CASSI finally produces the measurement $\mathbf{Y} \in \mathbb{R}^{H \times (W+d(N_{\lambda}-1))}$ by 
\begin{equation}\label{eq: 3d_to_2d}
    \mathbf{Y} = \textstyle \sum_{n_{\lambda}=1}^{N_{\lambda}}\mathbf{F}'(:,:,n_{\lambda})+\mathbf{\Omega},
\end{equation}
where $\mathbf{\Omega}$ denotes the measurement noise. Notably, the pixel-wise addition by Eq.~\eqref{eq: 3d_to_2d} imposes the second reconstruction challenge of \textit{entangled data loss} as one needs to precisely distinguish the spatial details of a specific wavelength from another one given the single 2D measurement.

From a physical encoding perspective, a high-fidelity reconstruction is largely about properly tackling the two-fold data loss.  
For the \textit{entangled data loss}, we trace back to the nature of the hyperspectral cube, and design a spatial-spectral attention ($S^2$-attn) mechanism accordingly in expectation to disassemble the spectral signals out of the measurement. We systematically discuss this part in Section~\ref{subsec: Spatial-Spectral Attention}. In Section~\ref{subsec: mask-driven loss}, we explicitly measure the pixel-wise difficulty-level owning to the \textit{masked data loss}, and propose a curriculum training strategy~\cite{bengio2009curriculum,wang2021survey}, mask-aware learning, without additional training cost (\emph{i.e.}, time, parameters) introduced.

\begin{figure*}[ht] 
\centering 
\includegraphics[width=.9\textwidth]{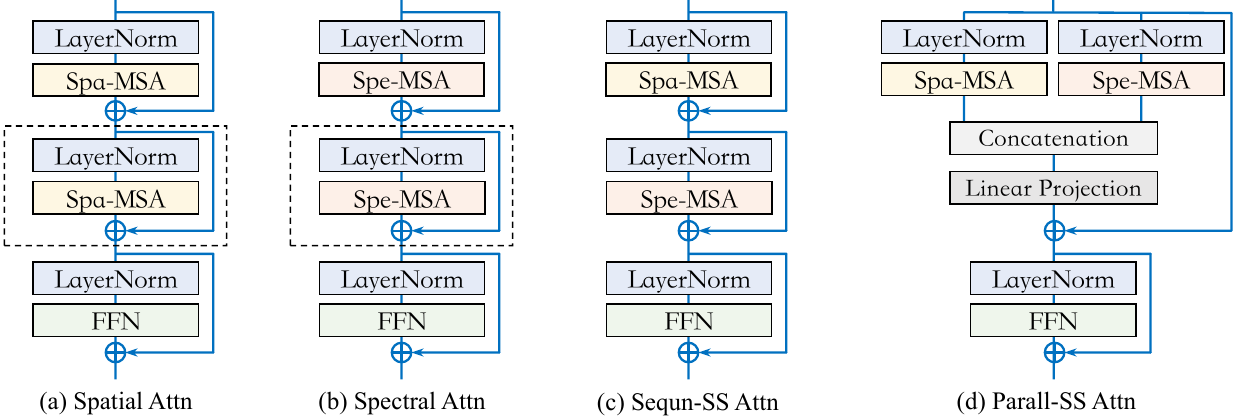} %\vspace{-2mm}
\caption{Different attention block designs, which are readily wrapped by off-the-shelf components, \emph{i.e.}, layer normalization (\texttt{LN}), and multi-head self-attention (\texttt{MSA}). Single-typed attentions (a, b) are  enhanced by additional \texttt{LN}-\texttt{MSA} (in dashed boxes) for a comparable complexity. All the attention blocks takes the feature embedding upon the window partition. The proposed $S^2$-Transformer adopts the \texttt{Parall-SS Attn}. }
\vspace{-3mm}
\label{fig: ss_attn} 
\end{figure*}

\subsection{Overall Architecture}\label{subsec: Overall Architecture}
{The proposed spatial-spectral ($S^2$-) Transformer takes the \textbf{3D} data as input and produces the hyperspectral data cube. Following the merit of \cite{meng2020gap,Meng20ECCV_TSAnet,wang2021new}, we initialize the network input $\mathbf{Y}'$$\in$$\mathbb{R}^{H \times W \times N_{\lambda}}$ by measurement $\mathbf{Y}$ via a channel-wise shift operation and a pixel-wise production with mask, denoted by the \texttt{Shift \& Dot-product} module in Fig.~\ref{fig: framework}.}
\begin{equation}\label{eq: shift}
    \mathbf{Y}'(:,:,n_{\lambda})=\mathbf{Y}(:, d(\lambda-\lambda^*):d(\lambda-\lambda^*)+W) \odot \textbf{M},
\end{equation}
{where $\odot$ denotes a pixel-wise multiplication. The whole network is given by $\widehat{\mathbf{F}} = f(\theta; \mathbf{Y}')$, 
where $\theta$ denotes the learnable parameters and $\widehat{\mathbf{F}}$ represents the reconstruction result. As shown in Fig.~\ref{fig: framework}, the proposed $S^2$-Transformer is composed of three parts. (1) A feature extractor by a \texttt{CONV3$\times$3} layer, producing the embedding representation $\mathbf{Z^{(0)}} \in \mathbb{R}^{H \times W \times C}$}
\begin{equation}\label{eq: extractor}
    \mathbf{Z^{(0)}} = \texttt{CONV}(\mathbf{Y}').
\end{equation}
{Let $C$ denote the number of embedding channels, which should be large enough to provide redundancy for the spectrum correlation exploitation. (2) A reconstruction head (\texttt{Reconstruction} module in Fig.~\ref{fig: framework}) at the end of the network employs a \texttt{CONV3$\times$3} layer, which maps the embedded space to the hyperspectral domain. (3) $K$ consecutive $S^2$-attn Transformer stages. Each stage is composed of $L$ concatenated $S^2$-attn blocks governed by a residual connection. Let $\mathbf{Z}^{(k)} \in \mathbb{R}^{H \times W \times C}$ denote the output of the $k$-th stage, $1 \leq k \leq K$. The feed-forward function $f_{\texttt{S}}(\cdot)$ of $k$-th stage is }
\begin{equation}\label{eq: stage func}
    \mathbf{Z}^{(k)} = f_{\texttt{S}}(\mathbf{Z}^{(k-1)}) = \mathbf{Z}^{(k-1)} + \texttt{CONV}(f_{\texttt{L}}(\mathbf{Z}^{(k-1)})),
\end{equation}
{where $\mathbf{Z}^{(k-1)}$ denotes the output embedding of $(k-1)$-th stage. $f_{\texttt{L}}(\cdot)$ denotes the mapping of $L$ concatenated $S^2$-attn blocks, which can be further expanded as   $f_{\texttt{L}}(\cdot)$$=$$f_{\texttt{B}}(...f_{\texttt{B}}(\cdot)...)$. We use $f_{\texttt{B}}(\cdot)$ to denote the mapping of a single $S^2$-attn block. Notably, we expend the  $\mathbf{Z}^{(k)}$ as $\mathbf{Z}^{(k,l)}$ to represent the $l$-th $S^2$-attn block within $k$-th stage, we have }
\begin{equation}\label{eq: blocks func}
\begin{aligned}
    \mathbf{Z}^{(k,L)} &= f_{\texttt{L}}(\mathbf{Z}^{(k,0)}), \\ 
    \mathbf{Z}^{(k,l)} &= f_{\texttt{B}}(\mathbf{Z}^{(k, l-1)}),
\end{aligned}
\end{equation}
{where $l=\{1,...,L\}$. Notably, we systematically devise four types of attention blocks of $f(\cdot)$ in the following. To avoid the quadratic complexity in the traditional self-attention calculation~\cite{vaswani2017attention}, we perform the window partition following~\cite{liang2021swinir,liu2021swin} toward the feature embedding for all types of attention mechanisms as shown in Fig.~\ref{fig: framework}.} 

{Hyperspectral images are characterized by  \textit{spatially and spectrally informative} clues, which can be leveraged for fine-grained signal retrieval upon the \textbf{2D} measurement (see Fig.~\ref{fig: coverfig}). Besides, both the spatial and spectral data clues are entangled and fused within the \textbf{2D} measurement. Despite the previous efforts of spatial and spectral modeling for hyperspectral data reconstruction, there still lacks a systematic exploitation to (1) uncover the behaviour of both attention types and (2) determine an appropriate information disentanglement upon the measurement.  In the following, we first discuss the single modalities (spatial, spectral attention designs). Then we propose joint modeling prototypes (sequential and paralleled attentions). We adopt the paralleled version for the proposed method.}

\subsection{Spatial \& Spectral Attention Blocks}\label{subsec: Spatial-Spectral Attention}
In this section, we give a brief introduction toward both spatial and spectral attentions. We  uncover their advantages and potential limitations underlying hyperspectral characteristics, respectively. Furthermore, we propose hybrid spatial-spectral attention structures and systematically discuss their behaviours toward the hyperspectral data modeling.

\textbf{Spatial Attention.} Since proposed by~\cite{vaswani2017attention} and developed in~\cite{dosovitskiy2021an}, the frequently-used spatial attention (\texttt{Spa}) has been evolved with key components of multi-head self-attention (MSA), layer normalization (LN)~\cite{ba2016layer}, and feed-forward network (FFN). Given the input of the size $H$$\times W$$\times C$,
the attention module partitions it into $\frac{HW}{M^2}$ windows, where $M$ is the window size. We adopt the multi-head self attention (MSA) module  for the attention computation. 
Let $T$ be the number of heads. Each head in MSA is allocated with $C_h$$=$$\lfloor\frac{C}{T}\rfloor$ channels.
For the annotation simplicity, we demonstrate the computation within a window for each head by feeding the attention modules with the input feature embedding $\mathbf{Z}_{\texttt{in}} \in \mathbb{R}^{M^2 \times C_h}$ accordingly. All the other windows among different heads follow the same computational procedure.
As shown in Fig.~\ref{fig: ss_attn} (a), 
the mapping is conducted by
\begin{equation}\label{eq: spa-msa}
\begin{aligned}
        \mathbf{Z}_{\texttt{Spa}} &= f_{\texttt{Spa-MSA}}(\texttt{LN}(\mathbf{Z}_{\texttt{in}})) + \mathbf{Z}_{\texttt{in}}, \\
        \mathbf{Z}_{\texttt{out}} &= f_{\texttt{FFN}}(\texttt{LN}(\mathbf{Z}_{\texttt{Spa}})) + \mathbf{Z}_{\texttt{Spa}},
        \end{aligned}
\end{equation}
where $f_{\texttt{FFN}}(\cdot)$ is instantiated by a \texttt{Linear}-\texttt{GELU}-\texttt{Linear}-\texttt{GELU} structure, and $f_{\texttt{Spa-MSA}}(\cdot)$ denotes a MSA module, yielding $\mathbf{Z}_{\texttt{Spa}}$. The $\mathbf{Z}_{\texttt{out}}$ denotes the final output of the spatial attention block.  We adopt $\mathbf{Q}$, $\mathbf{K}$, and $\mathbf{V}\in\mathbb{R}^{
M^2 \times C_h}$ to represent the query, key, and value, and compute them with the feature embedding $\mathbf{Z}_{\texttt{in}}$ and a linear projection by
\begin{equation}\label{eq: kqv}
    \mathbf{K}=\mathbf{W}^{\mathbf{K}}\mathbf{Z}_{\texttt{in}}, ~
    \mathbf{Q}=\mathbf{W}^{\mathbf{Q}}\mathbf{Z}_{\texttt{in}}, ~
    \mathbf{V}=\mathbf{W}^{\mathbf{V}}\mathbf{Z}_{\texttt{in}}, 
\end{equation}
where $\mathbf{W}^{\mathbf{K}}$, $\mathbf{W}^{\mathbf{Q}}$, and $\mathbf{W}^{\mathbf{V}}$ are learnable parameters. The output of spatial attention $\mathbf{Z}_{\texttt{Spa}}$ is given as
\begin{equation}\label{eq: spa attn}
\begin{aligned}
    \mathbf{A}_{\texttt{Spa}} &= \texttt{softmax}(\mathbf{K}\mathbf{Q}^{T}/\beta+\mathbf{B}), \\
    \mathbf{Z}_{\texttt{Spa}} &= \mathbf{A}_{\texttt{Spa}}\mathbf{V},
    \end{aligned}
\end{equation}
where $\beta$ is a learnable scaling factor for the overlarge values, initialized as $\sqrt{C_h}$, $\mathbf{B}$ represents a learnable position bias matrix following~\cite{bao2020unilmv2,hu2018relation}, and $\mathbf{A}_\texttt{Spa}$$\in$$\mathbb{R}^{
M^2 \times M^2}$ stands for self-attention matrix in a partitioned window. The outputs of $T$ heads will be concatenated afterward.

Given the spatial tokens of the size $1$$\times$$1$$\times C_h$, the spatial attention benefits the reconstruction by 1) taking well advantage of semantic clues of each pixel, 2) enabling interactions among neighbored channels within each head, which meets with the nature that spectral-adjacent contents are highly correlated (\texttt{neighbor A}, \texttt{B} in Fig.~\ref{fig: coverfig}). However, it shows limitations in describing the spectral variations of each token.

\textbf{Spectral Attention.} 
We leverage the self-attention in spectral domain (\texttt{Spe}) after the window partition. The attention module shares a similar architecture to Eq.~\eqref{eq: spa-msa}, but differs in multi-head self-attention (MSA) computation, \emph{i.e.}, $f_{\texttt{Spe-MSA}}(\cdot)$, for the spectral dependency modeling.
We compute query, key, and value with Eq.~\eqref{eq: kqv} and the spectral attention in each window is computed by
\begin{equation}\label{eq: spe attn}
\begin{aligned}
    \mathbf{A}_{\texttt{Spe}} &=\texttt{softmax}(\mathbf{K}^T\mathbf{Q}/\beta + \mathbf{B}), \\
    \mathbf{Z}_{\texttt{Spe}} &= \mathbf{V}\mathbf{A}_{\texttt{Spe}},
    \end{aligned}
\end{equation}
where $\mathbf{Z}_{\texttt{Spe}} \in \mathbb{R}^{M^2 \times C_h}$ refers to the output feature embedding of $f_{\texttt{Spe-MSA}}(\cdot)$ and $\mathbf{A}_{\texttt{Spe}}$ $\in \mathbb{R}^{C_h \times C_h}$ stands for 
the self-attention matrix of the  partitioned cube in each head. We scale the matrix multiplication by a learnable scalar $\beta$ and add $\mathbf{B}$ as a relative position bias matrix~\cite{shaw2018self}.

Computed with the spectral token of the size $M$$\times M$$\times 1$, 
the dimensionality of the spectral attention matrix $\mathbf{A}_{\texttt{Spe}}$ differs from the $\mathbf{A}_{\texttt{Spa}}$ of the size $M^2\times M^2$. By integrally considering the pixels in a channel,
the spectral attention can better abstract the inherent spectrum principle underlying discrete wavelengths. However, $\mathbf{A}_{\texttt{Spe}}$ is of $C_h \times C_h$, regardless of the token spatial size, \emph{i.e.}, $M\times M$. Such a resolution-agnostic property impedes the  abstraction ability of the attention when visual detail increases, but instead only allows limited visual clues abstraction for the channel-wise modeling.

Considering the limitations of both attention types, it would be inadequate to solely employ either one for a high-fidelity reconstruction. Thereby, we provide two types of hybrid structure, which not only exploit two-fold advantages, but also explore the spatial-spectral inter-dependencies.

\textbf{Sequential Spa-Spe Attention.} 
We perform sequential attention (\texttt{Sequn-SS}) by cascading the $f_{\texttt{Spa-MSA}}(\cdot)$, $f_{\texttt{Spe-MSA}}(\cdot)$, and the feed-forword network $f_{\texttt{FFN}}(\cdot)$. As shown in Fig.~\ref{fig: ss_attn} (c), the feed-forward pass upon the feature embedding $\mathbf{Z}_{\texttt{in}}$ is
\begin{equation}\label{eq: sequn-ss attn}
    \begin{aligned}
     &\mathbf{Z}_{\texttt{Spe}} = f_{\texttt{Spe-MSA}}(\texttt{LN}(\mathbf{Z}_{\texttt{in}})) + \mathbf{Z}_{\texttt{in}}, \\
    &\mathbf{Z}_{\texttt{Spa}} = f_{\texttt{Spa-MSA}}(\texttt{LN}(\mathbf{Z}_{\texttt{Spe}})) + \mathbf{Z}_{\texttt{Spe}},\\
    &\mathbf{Z}_{\texttt{Sequn-SS}} = f_{\texttt{FFN}}(\texttt{LN}(\mathbf{Z}_{\texttt{Spa}})) + \mathbf{Z}_{\texttt{Spa}},
    \end{aligned}
\end{equation}
where $f_{\texttt{Spa-MSA}}(\cdot)$ and $f_{\texttt{Spe-MSA}}(\cdot)$ are given by Eq.~\eqref{eq: spa attn} and Eq.~\eqref{eq: spe attn}, respectively. The \texttt{Sequn-SS} explores the data correlations from spatial and spectral perspectives, and allows interaction between both attention mechanisms. However, the behaviours of these two attention types are highly interrelated and either one lacks independence. Besides, we notice the potential superiority of \texttt{Sequn-SS} might attribute to a large model size and higher complexity in comparison to previous attention block types under the same optimization setting. Thus, we empower the aforementioned spatial and spectral attention block designs with additional MSA-LN modules
for a comparable model size and complexity, as shown by the dashed boxes in Fig.~\ref{fig: ss_attn} (a, b).

\begin{figure}[tp]
\includegraphics[width=0.48\textwidth]{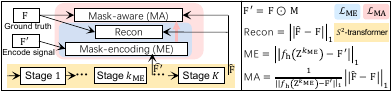}
\caption{Mask-aware learning pipeline. We start with the mask-encoding loss ($\mathcal{L}_{\texttt{ME}}$) to estimate the encoded/complete signals. Then the mask-aware loss ($\mathcal{L}_{\texttt{MA}}$) is adopted. The mask-encoding (ME) induces the mask-aware (MA) term.}
\vspace{-4mm}
\label{fig: mask_aware pipeline}
\end{figure} 

\textbf{Paralleled Spa-Spe Attention.}  
Besides the \texttt{Sequn-SS}, concurrently conducting spatial and spectral attention (\texttt{Parall-SS}) is another hybrid schema. The outputs of both attentions will be combined in a learnable mode. 
As shown in Fig.~\ref{fig: ss_attn} (d), given the feature embedding $\mathbf{Z}_{\texttt{in}}$, we have
\begin{equation}\label{eq: parall-ss attn}
    \begin{aligned}
    &\mathbf{Z}_{\texttt{cat}} = \mathbf{W}^{\texttt{cat}}[f_{\texttt{Spe-MSA}}(\texttt{LN}(\mathbf{Z}_{\texttt{in}})), f_{\texttt{Spe-MSA}}(\texttt{LN}(\mathbf{Z}_{\texttt{in}}))] + \mathbf{Z}_{\texttt{in}}, \\
    &\mathbf{Z}_{\texttt{Parall-SS}} = f_{\texttt{FFN}}(\texttt{LN}(\mathbf{Z}_{\texttt{cat}})) + \mathbf{Z}_{\texttt{cat}},
    \end{aligned}
\end{equation}
where $[\cdot]$ denotes the concatenation and  $\mathbf{W}^{\texttt{cat}}$ performs the linear projection. The underlying advantages of \texttt{Parall-SS} are both spatial and spectral attentions take effect upon the same input, during which their mutual interference existing in \texttt{Sequn-SS} is minimized. Plus, the learnable combination is more expressive in feature fusion. As a result, the proposed \texttt{Parall-SS} enjoys a structural superiority with negligible increase in parameters, in comparison to the other structures.

In summary, this section discusses the fors and againsts for spatial/spectral attentions adapting to the hyperspectral characterises, based on which we further propose two block designs with sequential and paralleled integration schema. By comparison, the paralleled fusion empirically allows a better modeling ability (see Section~\ref{subsec: spatial-spectral attn} for more analysis).

\subsection{Mask-aware Learning}\label{subsec: mask-driven loss}

Previous works have well explored different types of learning objectives for the HSI reconstruction, \emph{e.g.}, $\mathcal{L}_1$ loss~\cite{huang2021deep}, RMSE~\cite{Meng20ECCV_TSAnet}, Spectrum Constancy~\cite{cai2022mask}, and perceptual loss~\cite{meng2021perception}, \emph{etc}. 
By taking advantage of the semantic representations or spectral correlations, they enable promising texture reconstruction and perceptual quality. However, there still exist content-irrelevant artifacts in predictions and unexpected reconstruction difficulty in smooth areas (compared by (a) and (b) in Fig.~\ref{fig: mask_aware}), {since existing learning objectives treat all pixels equally and overlook the additional data loss, especially in masked areas. Masked pixels are different from the unmasked ones from a physical perspective of view -- the light intensity on the masked area is blocked by the coded aperture, leading to an attenuated visual clue and breaking the local texture smoothness. Correctly predicting the masked pixels plays an important role in the final performance qualitatively and perceptually. To this end, we conjecture a novel optics-induced uncertainty and recognize the masked pixels with higher uncertainty, while the unmasked regions potentially allow a more promising reconstruction by inferring from the well-preserved signal.}

Following this intuition, we propose a mask-aware learning strategy, as shown in Fig.~\ref{fig: mask_aware pipeline}. The main idea is to first distinguish between the masked and unmasked regions, implemented by fitting the encoded signal as $\mathbf{F}'=\mathbf{F}\odot\mathbf{M}$. Then we adaptively emphasize masked regions with a relatively higher loss penalty. The penalization degree is determined by referring to the latest reconstruction difficulty of $\mathbf{F}'$ throughout the training. To achieve this goal, we first pre-train the model with a mask-encoding loss ($\mathcal{L}_{\texttt{ME}}$) by
\begin{equation}\label{eq: loss_me}
    \mathcal{L}_{\texttt{ME}} = \alpha \underbrace{||f_{h}(\mathbf{Z}^{k_\texttt{ME}}) - \mathbf{F}'||_1}_{\rm \texttt{ME}} + \underbrace{||\widehat{\mathbf{F}} - \mathbf{F}||_1}_{\rm \texttt{Recon}},
\end{equation}
\noindent where the mask-encoding (\texttt{ME}) term explores the pixel-wise reconstruction difficulty relative to the encoded signal $\mathbf{F}'$ with early $S^2$-attn stages, \emph{i.e.}, the first $k_\texttt{ME}$ stages ($k_\texttt{ME}<K$). The $f_{h}(\cdot)$ is implemented by an \texttt{LN-CONV} structure. Meanwhile, we keep the original reconstruction term (\texttt{Recon}) to train the late blocks and stabilize the training. We use $\alpha$ to balance these two terms. In Fig.~\ref{fig: mask_aware} (c), we visualize the estimation of the encoded signal, which is corrupted from a mask pattern as expected. 
By inferring from this signal,
we propose a mask-aware (\texttt{MA}) term by taking advantage of the \texttt{ME} term, leading to the mask-aware loss ($\mathcal{L_{\texttt{MA}}}$) as
\begin{equation}\label{eq: loss_mu}
    \mathcal{L}_{\texttt{MA}} = \mathcal{L_{\texttt{ME}}} + \underbrace{\dfrac{\beta}{||f_{h}(\mathbf{Z}^{k_\texttt{ME}}) - \mathbf{F}'||_1}||\widehat{\mathbf{F}} - \mathbf{F}||_1}_{\texttt{MA}},
\end{equation}
where we attenuate the $\alpha$ inside  $\mathcal{L_{\texttt{ME}}}$ after pre-training and emphasize the masked areas by determining a $\beta$ value. Note that, the \texttt{ME} term serves as the weight of the \texttt{MA} term. In the following, we will theoretically discuss the impact of the proposed mask-aware loss to the unmasked pixels and reveal a prioritization effect on the masked pixels.

\begin{figure}[tp]
\includegraphics[width=0.48\textwidth]{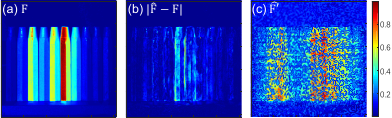}
\vspace{-1mm}
\caption{Visual evidence underlying mask-aware learning. (a) Reference. (b) Reconstruction difficulty by the recent method, MST~\cite{cai2022mask}. (c) Estimation of encoded signal with mask-encoding loss $\mathcal{L}_{\texttt{ME}}$. We select 2D spectral channels from $\mathbf{F}$,$|\widehat{\mathbf{F}}$$-$$\mathbf{F}|$, and $\widehat{\mathbf{F}}'$$\in$$\mathbb{R}^{H\times W \times N_\lambda}$ for visualization in (a$\sim$c). }
\vspace{-4mm}
\label{fig: mask_aware}
\end{figure}

On one hand, for the extreme case of unmasked areas, we consider a very small spatial region of the real mask so that we always have $\mathbf{M}_{ij}$$\rightarrow$$1$. Without loss of generality, we accordingly consider a feature map of $\mathbf{F}$ to simplify the notations, yet the result can be easily extended to \textbf{3D} feature maps. Correspondingly, we have  $\mathbf{F}_{ij} \in [0,1]$ and the encoded signal wherein $\mathbf{F}'_{ij}$$=$$\mathbf{M}_{ij}\odot \mathbf{F}_{ij}$. Since the $\widehat{\mathbf{F}}$ overrides the embedding representation underlying $f_h(\mathbf{Z}^{k_\texttt{ME}})$, we simplify their relationship as $\widehat{\mathbf{F}}$$=$$\tau(f_h(\mathbf{Z}^{k_\texttt{ME}}))$, where $\tau(\cdot)$ represents the cascaded $S^2$-attn stages. 
Considering their residual structures, we have $\tau(\mathbf{t})$$=$$g(\mathbf{t})+\mathbf{t}$ by simplifying $f_h(\mathbf{Z}^{k_\texttt{ME}})$ as $\mathbf{t}$. The mask-aware loss $\mathcal{L}_{\texttt{MA}}$ can be rewritten as
\begin{equation}\label{eq: loss_ma_res}
    \mathcal{L}_{\texttt{MA}} = \mathcal{L_{\texttt{ME}}} + \dfrac{\beta||g(\mathbf{t}) + \mathbf{t} - \mathbf{F}||_1}{||\mathbf{t} - \mathbf{M}\odot \mathbf{F}||_1},\quad   \mathbf{t}=f_h(\mathbf{Z}^{k_\texttt{ME}}),
\end{equation}
{where $g(\mathbf{t})$ represents the residual mapping and $\mathbf{t}$ is produced by the identity mapping. In the following, we finalize the induction of $\mathcal{L}_{\texttt{MA}}$ with Eq.~(14)  by providing an assumption and drawing a corollary.
\begin{assumption}\label{th: ap-mask}
Given a network $g_\phi(x)$ with a near-zero input, \emph{i.e.}, $x\rightarrow0$, a small value $\epsilon$ can upper-bound the network output, \emph{i.e.}, $|g_\phi(x)|\le \epsilon, \epsilon\rightarrow0$, in a small spatial region.
\end{assumption}}

\begin{table*}[t]
\footnotesize
\caption{PSNR (dB) values by different methods on the benchmark simulation dataset.}
\label{Tab: psnr_single_trn}
\centering
\resizebox{.95\textwidth}{!}{
\setlength{\tabcolsep}{0.7mm}
\centering
\begin{tabular}{l|cccccccccc|c} 
	\toprule
	Methods & Scene1 & Scene2 & Scene3 & Scene4 & Scene5 & Scene6 & Scene7 & Scene8 & Scene9 & Scene10 & \emph{Avg.}\\
	\midrule

	GAP-TV~\cite{yuan2016generalized} & 26.82 & 22.89 & 26.31 & 30.65 & 23.64 & 21.85 & 23.76 &21.98 & 22.63 & 23.10 &24.36 \\
	
	DeSCI~\cite{liu2018rank} &27.13 & 23.04 & 26.62 & 34.96 & 23.94 & 22.38 & 24.45 & 22.03 & 24.56 & 23.59 & 25.27 \\
	
	TSA-Net~\cite{Meng20ECCV_TSAnet} & 32.03 & 31.00 & 32.25 & 39.19 & 29.39 & 31.44 & 30.32 & 29.35 & 30.01 & 29.59 & 31.46 \\
	
	DGSMP~\cite{huang2021deep} & 33.26 & 32.09 & 33.06 & 40.54 & 28.86 & 33.08 & 30.74 & 31.55 & 31.66 & 31.44 & 32.63 \\
	
	SRN~\cite{wang2021new} &34.96 & 35.46 & 36.18 & 41.60 & 32.70 & 34.70 & 33.83 & 32.88 & 35.09 & 32.31 & 35.07 \\
	
	HDNet~\cite{hu2022hdnet} & 35.14 & 35.67 & 36.03 & 42.30 & 32.69 & 34.50 & 33.67 & 32.48 & 34.89 & 32.38 & 34.97 \\
	
	MST~\cite{cai2022mask} & 35.40 & 35.87 & 36.51 & 42.27 & 32.77 & 34.80 & 33.66 & 32.67 & 35.39 & 32.50 & 35.18 \\

        CST~\cite{cai2022coarse} & 35.89 & 36.82 & \textbf{38.20} & 42.38 & 33.16 & 35.71 & 34.87 & 34.30 & 36.44 & 33.03 & 36.08 \\
	
    \midrule
	$S^2$-Transformer & \textbf{36.17} & \textbf{37.57} & 37.29 & \textbf{42.96} & \textbf{34.40} & \textbf{36.44} & \textbf{35.41} & \textbf{34.50} & \textbf{36.54} & \textbf{33.57} & \textbf{36.48} \\
	\bottomrule
\end{tabular}}
\end{table*}

\begin{table*}[t]
\footnotesize
\caption{SSIM values by different methods on the benchmark simulation dataset.}
\label{Tab: ssim_single_trn}
\centering
\resizebox{.95\textwidth}{!}{
\setlength{\tabcolsep}{0.7mm}
\centering
\begin{tabular}{l|cccccccccc|c} 
	\toprule
	Methods & Scene1 & Scene2 & Scene3 & Scene4 & Scene5 & Scene6 & Scene7 & Scene8 & Scene9 & Scene10 & \emph{Avg.}\\
	\midrule

	GAP-TV~\cite{yuan2016generalized} & 0.7544 & 0.6103 & 0.8024 & 0.8522 & 0.7033 & 0.6625 & 0.6881 & 0.6547 & 0.6815 & 0.5839 & 0.6993 \\

	DeSCI~\cite{liu2018rank} & 0.7479 & 0.6198 & 0.8182 & 0.8966 & 0.7057 & 0.6834 & 0.7433 & 0.6725 & 0.7320 & 0.5874 & 0.7207 \\

	TSA-Net~\cite{Meng20ECCV_TSAnet} & 0.8920 & 0.8583 & 0.9145 & 0.9528 & 0.8835 & 0.9076 & 0.8782 & 0.8884 & 0.8901 & 0.8740 & 0.8939 \\

	DGSMP~\cite{huang2021deep} & 0.9152 & 0.8977 & 0.9251 & 0.9636 & 0.8820 & 0.9372 & 0.8860 & 0.9234 & 0.9110 & 0.9247 & 0.9166 \\

	SRN~\cite{wang2021new} & 0.9345 & 0.9373 & 0.9476 & 0.9703 & 0.9444 & 0.9512 & 0.9241 & 0.9443 & 0.9414 & 0.9348 & 0.9430 \\
	
	HDNet~\cite{hu2022hdnet} & 0.9352 & 0.9404 & 0.9434 & 0.9694 & 0.9460 & 0.9518 & 0.9263 & 0.9406 & 0.9415 & 0.9365 & 0.9431 \\
	
	MST~\cite{cai2022mask} & 0.9405 & 0.9440 & 0.9525 & 0.9734 & 0.9471 & 0.9553 & 0.9254 & 0.9479 & 0.9491 & 0.9408 & 0.9476 \\

        CST~\cite{cai2022coarse} & \textbf{0.9494} & 0.9546 & \textbf{0.9623} & 0.9752 & 0.9544 & 0.9631 & 0.9451 & 0.9608 & 0.9566 & 0.9450 & 0.9566 \\

    \midrule
	$S^2$-Transformer & 0.9490 & \textbf{0.9582} & 0.9567 & \textbf{0.9754} & \textbf{0.9596} & \textbf{0.9654} & \textbf{0.9461} & \textbf{0.9625} & \textbf{0.9592} & \textbf{0.9517} & \textbf{0.9584}\\
	\bottomrule
\end{tabular}}
\vspace{-2mm}
\end{table*}

\begin{figure*}[tp] 
\centering 
\includegraphics[width=0.95\textwidth]{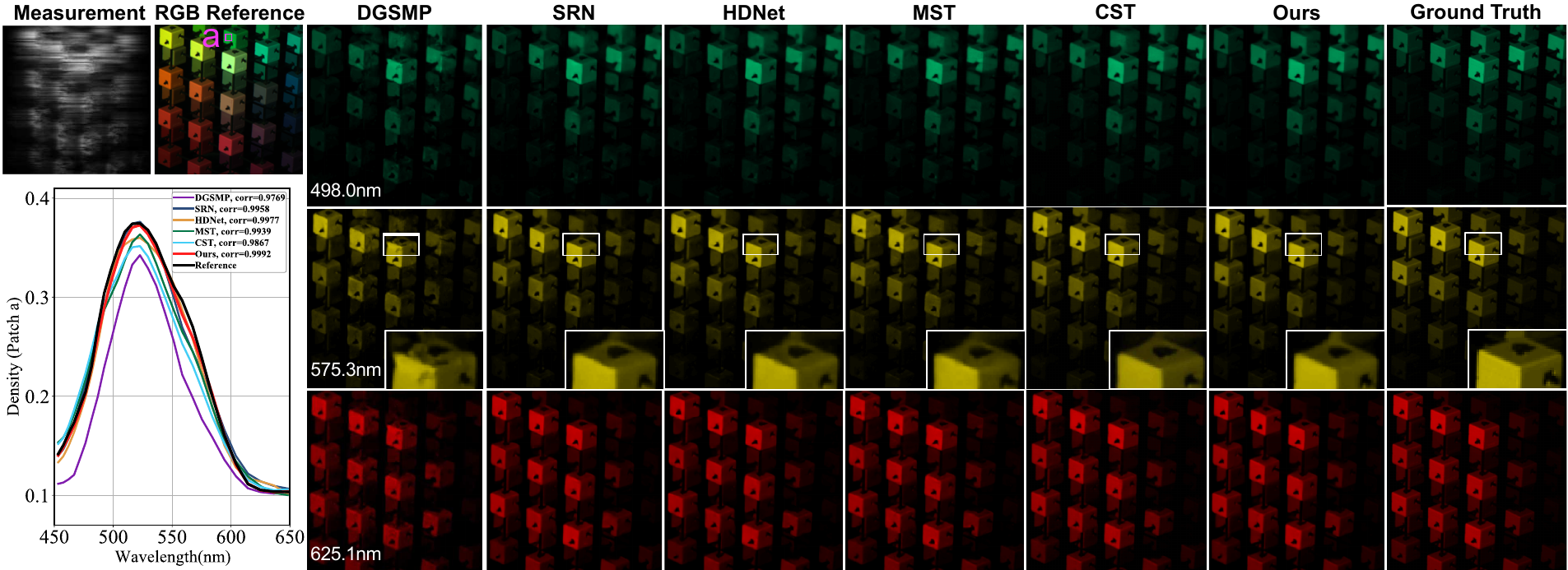}
\vspace{-2mm}
\caption{Reconstruction results for a simulation data. Five state-of-the-art methods and the proposed method (second to the right) are presented on 3 out of 28 wavelengths. The RGB reference is shown to demonstrate the color (top-left). The density-vs-wavelength curves (bottom-left) corresponding to the chosen patch (\texttt{patch a}) demonstrates the \textbf{spectral fidelity}.} 
\label{Fig: simu_result}
\vspace{-2mm}
\end{figure*}

{\begin{corollary}\label{th: mask}
With adequate iterations of stochastic gradient descent and a proper optimization schedule, the \texttt{MA} term could converge to $\sum_{ij}\pm\nabla_{\mathbf{t}_{ij}}|g(\mathbf{M}_{ij}\mathbf{F}_{ij})+\mathbf{t}_{ij}-\mathbf{F}_{ij}|$.  
\end{corollary}}
\begin{proof}
We consider the situation where mask-encoding loss ($\mathcal{L}_{\texttt{ME}}$) yields a good approximation of encoded signal $\widehat{\mathbf{F}}$, where $\mathbf{t}_{ij} \rightarrow \mathbf{M}_{ij}\mathbf{F}_{ij}$. According to Eq.~\eqref{eq: loss_ma_res},  the \texttt{MA} becomes
\begin{equation*}
    |g(\mathbf{t}_{ij})+\mathbf{t}_{ij}-\mathbf{F}_{ij}| \rightarrow |g(\mathbf{M}_{ij}\mathbf{F}_{ij})-(1-\mathbf{M}_{ij})\mathbf{F}_{ij}|.
\end{equation*}
{According to Assumption~\ref{th: ap-mask}, the residual $g(\mathbf{M}_{ij}\mathbf{F}_{ij})$ learns a small value bounded by $\epsilon$, \emph{i.e.}, $|g(\mathbf{M}_{ij}\mathbf{F}_{ij})|<\epsilon, \epsilon\rightarrow 0$ upon enough SGD steps. Thereby,
we can find }
\begin{equation*}
    |g(\mathbf{M}_{ij}\mathbf{F}_{ij})-(1-\mathbf{M}_{ij})\mathbf{F}_{ij}|\rightarrow0,
\end{equation*}
where the term $(1-\mathbf{M}_{ij})\mathbf{F}_{ij}$ is omitted with  $\mathbf{M}_{ij}\rightarrow1$ in the small spatial region. Accordingly, we have the numerator
\begin{equation*}
    \sum_{ij}|g(\mathbf{M}_{ij}\mathbf{F}_{ij})-(1-\mathbf{M}_{ij})\mathbf{F}_{ij}|\rightarrow0.
\end{equation*}
Similarly, let $\mathbf{t}_{ij} \rightarrow \mathbf{M}_{ij}\mathbf{F}_{ij}$, the denominator becomes
\begin{equation*}
    \sum_{ij}|\mathbf{t}_{ij}-\mathbf{M}_{ij}\mathbf{F}_{ij}| \rightarrow 0.
\end{equation*}
Based on the previous induction, we find the indeterminacy form of both numerator and the denominator under the condition of $\mathbf{t}_{ij} \rightarrow \mathbf{M}_{ij}\mathbf{F}_{ij}$. Notably, both numerator and the denominator are differentiable on an open interval of $\mathbf{t}_{ij}$ except at $\mathbf{M}_{ij}\mathbf{F}_{ij}$, in which the denominator enables a non-zero derivation, \emph{i.e.}, $\sum_{ij}\nabla_{\mathbf{t}_{ij}}|\mathbf{t}_{ij}-\mathbf{M}_{ij}\mathbf{F}_{ij}|\ne0$ for all $\mathbf{t}_{ij}$ in an open interval with $\mathbf{t}_{ij} \ne \mathbf{M}_{ij}\mathbf{F}_{ij}$. Based on these conditions, 
we then exploit the convergence of the \texttt{MA} term upon Bernoulli's rule
\begin{equation}\label{eq: proof}
    \begin{aligned}   &\lim\limits_{\mathbf{t}_{ij}\to\mathbf{M}_{ij}\mathbf{F}_{ij}}\dfrac{\sum_{ij}|g(\mathbf{t}_{ij})+\mathbf{t}_{ij}-\mathbf{F}_{ij}|}{\sum_{ij}|\mathbf{t}_{ij}-\mathbf{M}_{ij}\mathbf{F}_{ij}|} \\
    =& \lim\limits_{\mathbf{t}_{ij}\to\mathbf{M}_{ij}\mathbf{F}_{ij}}\dfrac{\sum_{ij}\nabla_{\mathbf{t}_{ij}}|g(\mathbf{t}_{ij})+\mathbf{t}_{ij}-\mathbf{F}_{ij}|}{\sum_{ij}\nabla_{\mathbf{t}_{ij}}|\mathbf{t}_{ij}-\mathbf{M}_{ij}\mathbf{F}_{ij}|}\\
    =&\sum_{ij}\pm\nabla_{\mathbf{t}_{ij}}|g(\mathbf{M}_{ij}\mathbf{F}_{ij})+\mathbf{t}_{ij}-\mathbf{F}_{ij}|,
    \end{aligned}
\end{equation}
where $\nabla_{\mathbf{t}_{ij}}|\mathbf{t}_{ij}-\mathbf{M}_{ij}\mathbf{F}_{ij}|=\pm1$ adapting to the value of $\mathbf{t}_{ij}$. Thus, we complete the proof.
\end{proof}

\begin{figure*}[t] 
\centering 
\includegraphics[width=0.97\textwidth]{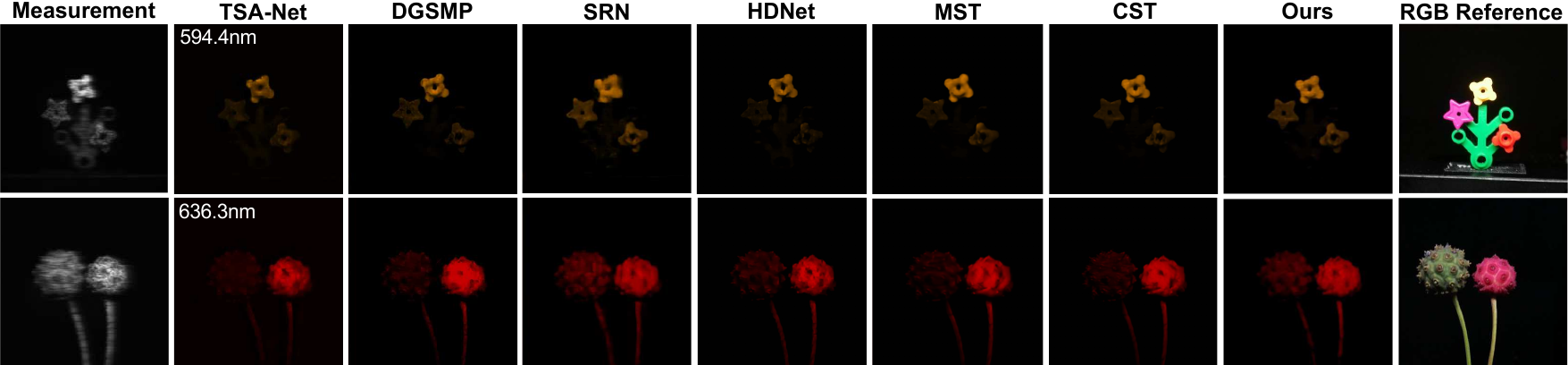}
\caption{Reconstruction results on real-captured measurements. The proposed method retrieves more contents (\emph{e.g.}, the smallest hole at center of the yellow flower in top row) \textbf{Zoom in for a better visualization}.} \label{Fig: real_result}
\vspace{-2mm}
\end{figure*}

\noindent \emph{Remark}. 
Corollary~\ref{th: mask} uncovers a constant convergence tendency of the proposed \texttt{MA} term within a small spatial regions by satisfying the Bernulli's rule.  
Therefore, the $\mathcal{L}_{\texttt{MA}}$ becomes
\begin{equation*}\label{eq: m=1_Lma}
\begin{aligned}
    \mathcal{L}_{\texttt{MA}}=&\alpha||f_{h}(\mathbf{Z}^{k_\texttt{ME}}) - \mathbf{F}'||_1+||\widehat{\mathbf{F}} - \mathbf{F}||_1 \\
    &+ \sum_{ij}\pm\nabla_{\mathbf{t}_{ij}}|g(\mathbf{M}_{ij}\mathbf{F}_{ij})+\mathbf{t}_{ij}-\mathbf{F}_{ij}|,
\end{aligned}
\end{equation*}
{which indicates that in the \textit{extreme case} of totally unmasked regions, \emph{i.e.}, $\mathbf{M}_{ij}\rightarrow1$,   the proposed \texttt{MA} term (1) hardly penalize the unmasked pixels and (2) is inoffensive to the reconstruction. Notably, the above induction may not applicable to the heavily noisy areas in practice. In Fig.~\ref{Fig: loss MA}, we empirically demonstrate that the \texttt{MA} term converges on unmasked regions (\emph{i.e.}, $\mathbf{M}_{ij}>0.5$), benefiting the reconstruction. 
On the other hand, given a real mask value $\mathbf{M}_{ij}\rightarrow0$ in a small spatial region and ground truth pixel value $\mathbf{F}_{ij}\in[0,1]$, we have encoded signal $\mathbf{F}'_{ij}\rightarrow0$. The mask-aware loss  could be simplified as }
\begin{equation*}\label{eq: m=0Lma}
    \mathcal{L}_{\texttt{MA}}=||\widehat{\mathbf{F}}-\mathbf{F}||_1+\alpha||f_{h}(\mathbf{Z}^{k_\texttt{ME}})||_1+\dfrac{\beta\sum_{ij}|\widehat{\mathbf{F}}_{ij} - \mathbf{F}_{ij}|}{\sum_{ij}|f_{h}(\mathbf{Z}^{k_\texttt{ME}}_{ij})|}, 
\end{equation*}
{where \texttt{ME}, \texttt{MA}, and \texttt{Recon} take effect for minimizing the overall learning objective of $\mathcal{L}_{\texttt{MA}}$. Notably, a smaller $||f_{h}(\mathbf{Z}^{k_\texttt{ME}}_{ij})||_1$ yields a larger weight of $\tfrac{\beta}{||f_{h}(\mathbf{Z}^{k_\texttt{ME}}_{ij})||_1}$. The $\sum_{ij}|\widehat{\mathbf{F}}_{ij} - \mathbf{F}_{ij}|$ needs to be further minimized to compensate for the \texttt{MA} value amplification.} Therefore, all of the underlying masked pixels are prioritized during the reconstruction.

In summary, we theoretically present the effectiveness of the proposed mask-aware learning strategy. It adaptively enhance the reconstruction penalty for the masked pixels, while equally retrieves the unmasked pixels without discrimination, under mild conditions. Note that the mask-encoding term facilitates the mask-aware term in $\mathcal{L}_{\texttt{MA}}$, thus the pre-training upon $\mathcal{L}_{\texttt{ME}}$ is necessary for the proposed strategy.

\section{Experiment}\label{sec: experiment}
We do extensive experiments for the proposed method. Specifically, Section~\ref{subsec: setting} introduces the experimental settings. Section~\ref{subsec: performance} quantitatively and perceptually present the reconstruction performance on both simulation and real data. In Section~\ref{subsec: spatial-spectral attn}, we discuss the proposed $S^2$-attn blocks and reason the superiority of the \texttt{Parall-SS} version by exploiting the underlying interpretability. In Section~\ref{subsec: mask-aware learning}, we empirically analysis the mask-aware learning strategy. { Besides, we analysis the performance of the proposed method under challenging scenarios in Section~\ref{subsec: challenge}. We include a comparison experiment using geological remotely sensed data of AVIRIS in Section~\ref{sec: AVIRIS}. }

\subsection{Experimental Settings}
\label{subsec: setting}

\begin{table}[tp] 
\footnotesize
\caption{Model size and complexity of different methods.}
\vspace{-1.58mm}
\label{Tab: complexity}
\centering
\resizebox{.48\textwidth}{!}{
\centering
\begin{tabular}{l|rrrrrr} 
            	\toprule
            	Methods & PSNR & SSIM & \#params (M) & FLOPs (G)   \\
            	\midrule
            	TSA-Net~\cite{Meng20ECCV_TSAnet} & 31.46 & 0.8939 & 44.25 & 110.06  \\
            	DGSMP~\cite{huang2021deep}       & 32.63 & 0.9166 & 3.76  & 646.65  \\
            	SRN~\cite{wang2021new}           & 35.07 & 0.9430 & \textbf{1.25}  & 81.84   \\
            	HDNet~\cite{hu2022hdnet}         & 34.97 & 0.9431 & 2.37  & 154.76  \\
            	MST~\cite{cai2022mask}           & 35.18 & 0.9476 & 2.03  & 28.15   \\
                    CST~\cite{cai2022coarse} & 36.08 & 0.9566 & 3.00 & 40.10 \\
            	\midrule
            	$S^2$-Transformer                & \textbf{36.48} & \textbf{0.9584} & 1.80 & \textbf{27.21}\\
            	\bottomrule
            \end{tabular}}
\end{table}

\begin{table}[t]
\caption{Naturalness Image Quality Evaluator (NIQE) evaluation on real hyperspectral dataset by pupolar methods.}
\vspace{-1mm}
\label{Tab: real metric}
\centering
\resizebox{.48\textwidth}{!}{
\centering
\begin{tabular}{c|ccccc} 
	\toprule
	\multirow{2}{*}{Methods} & \multirow{2}{*}{MST} & \multirow{2}{*}{HDNet} & \multirow{2}{*}{CST} & \multirow{2}{*}{$S^2$-Transformer} & $S^2$-Transformer \\
        & & & &  & w/ $\mathcal{L}_\texttt{ME}$,$\mathcal{L}_\texttt{MA}$\\
	\midrule
	NIQE ($\downarrow$) & 6.9219 & 5.9207 & 6.5755 & 6.0950 & \textbf{5.8833} \\
	\bottomrule
\end{tabular}}\vspace{-0.5cm}
\end{table}

\begin{table}[tp] 
\caption{{Performance and complexity analysis of different attention blocks. All the other settings are kept the same.  The proposed $S^2$-Transformer adopts \texttt{Parall-SS}.} }
\label{Tab: structure ablation}
\centering
\resizebox{.48\textwidth}{!}{
\centering
\begin{tabular}{l|rrr|r} 
\toprule
\multirow{2}{*}{Types} &  \multirow{2}{*}{\texttt{SpaSpa}} & \multirow{2}{*}{\texttt{SpeSpe}} & \multirow{2}{*}{\texttt{Sequn-SS}} & $S^2$-Transformer \\
&&&&\texttt{Parall-SS} \\
\midrule
PSNR    & 36.01  & 35.35  & 36.06  & \textbf{36.48}  \\
SSIM    & 0.9563 & 0.9540 & 0.9565 & \textbf{0.9584} \\
FLOPs (G)   & 61.78   & 27.49 & 44.64 & \textbf{27.21} \\
\#params (M)   & 2.86    & 2.93 & 2.90  & \textbf{1.80}  \\
\bottomrule
\end{tabular}}
\vspace{-3mm}
\end{table}

\textbf{Dataset.}
Following previous methods~\cite{Meng20ECCV_TSAnet,cai2022mask,cai2022coarse}, we use the same dataset setting for network training and evaluation. Specifically, we adopt the CAVE~\cite{yasuma2010generalized} database for training, which provides 32 scenes with 512$\times$512 resolution and 400nm$\sim$700nm range of wavelength. The training samples are randomly cropped into 256$\times$256 patches for data augmentation. The total 28 spectral channels are determined by spectral interpolation. The testing dataset is composed of ten 256$\times$256$\times$28 hyperspectral images from the KAIST~\cite{choi2017high} dataset. In simulation, we treat the 3D hyperspectral images as ground truths and compute the 2D measurements following Eq.~\eqref{eq: CASSI mask modulate} and Eq.~\eqref{eq: 3d_to_2d}.  
For the real hyperspectral image acquisition, we feed the 660$\times$714 real-captured measurements by~\cite{Meng20ECCV_TSAnet} into the pre-trained model, generating the 660$\times$660$\times$28 hyperspectral images accordingly. We inject the Gaussian noise to the measurements during training to simulate the real measurement noise $\mathbf{\Omega}$ in the collected data.

\begin{figure*}[ht]
\footnotesize
\centering
\resizebox{0.97\textwidth}{!}{
\begin{tabular}{cc}
\begin{adjustbox}{valign=t}
\begin{tabular}{c}
\begin{adjustbox}{valign=t}
\begin{tabular}{ccccccccc}
\Huge\rotatebox{90}{\color{white}\textbf{AA}}& 
\begin{overpic}[width=0.289\textwidth]{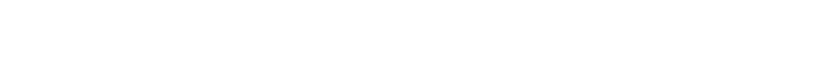}
\put(7,5){\Huge blk 25/chl 18}
\end{overpic}&
\begin{overpic}[width=0.289\textwidth]{white.pdf}
\put(7,5){\Huge blk 27/chl 18}
\end{overpic}&
\begin{overpic}[width=0.289\textwidth]{white.pdf}
\put(10,5){\Huge blk 19/chl 14}
\end{overpic}&
\begin{overpic}[width=0.289\textwidth]{white.pdf}
\put(7,5){\Huge blk 27/chl 29}
\end{overpic}&
\begin{overpic}[width=0.289\textwidth]{white.pdf}
\put(5,5){\Huge blk 25/chl 14}
\end{overpic}&
\begin{overpic}[width=0.289\textwidth]{white.pdf}
\put(8,5){\Huge blk 27/chl 16}
\end{overpic}&
\begin{overpic}[width=0.289\textwidth]{white.pdf}
\put(3,5){\Huge blk 28/chl 20}
\end{overpic}&
\begin{overpic}[width=0.289\textwidth]{white.pdf}
\put(3,5){\Huge blk 27/chl 20}
\end{overpic}
\end{tabular}
\end{adjustbox}
\vspace{-1mm}
\\
\begin{adjustbox}{valign=t}
\begin{tabular}{ccccccccc}
\Huge\rotatebox{90}{\textbf{\texttt{Spe-MSA}}}&
\begin{overpic}[width=0.289\textwidth]{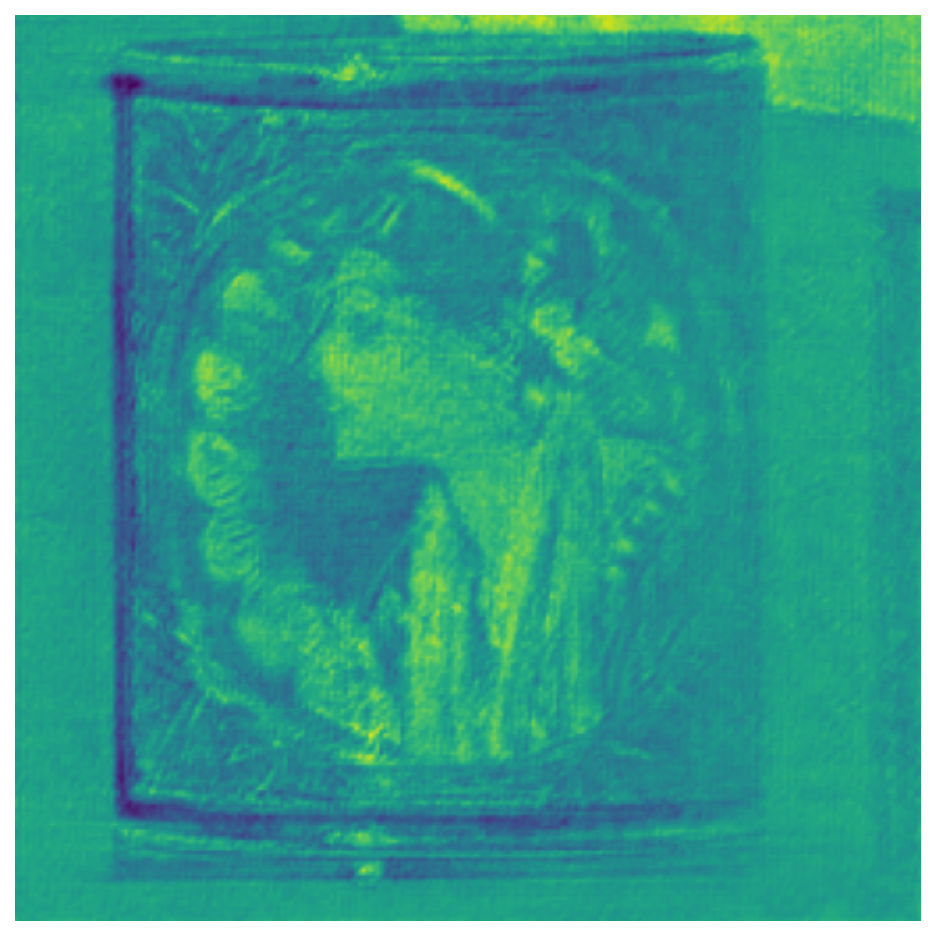}
\put(5,5){}
\end{overpic}&
\begin{overpic}[width=0.289\textwidth]{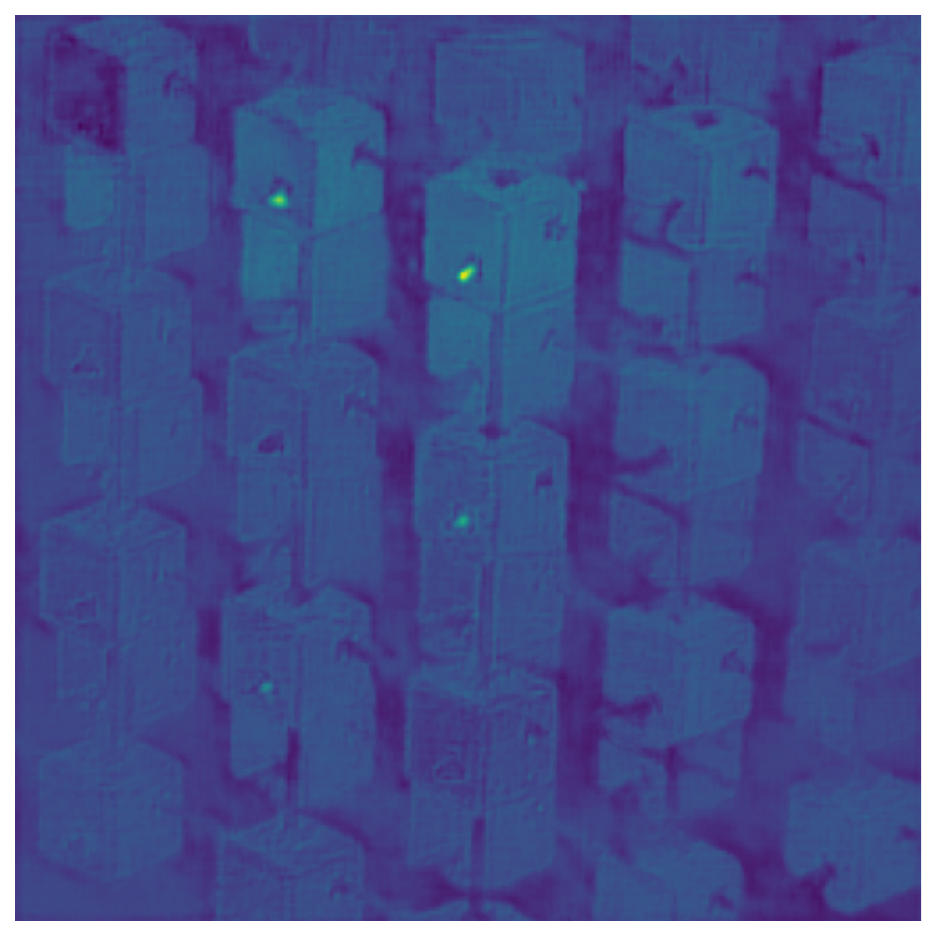}
\put(5,5){}
\end{overpic}&
\begin{overpic}[width=0.289\textwidth]{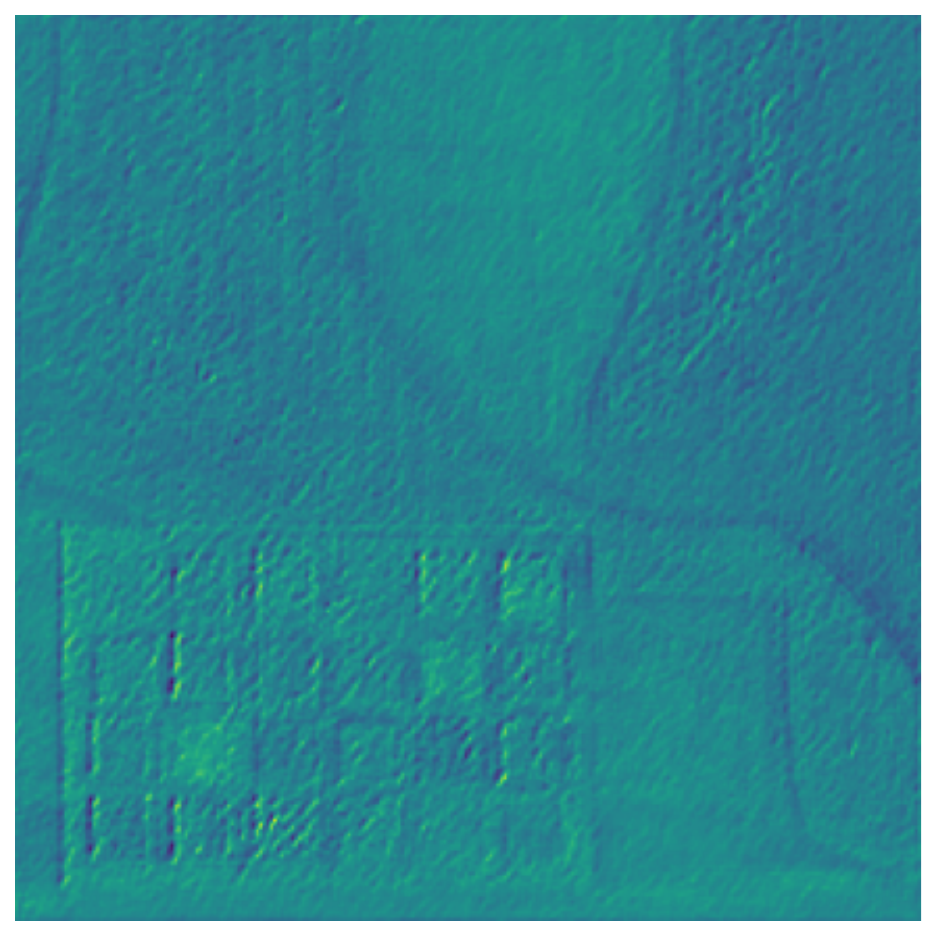}
\put(5,5){}
\end{overpic}&
\begin{overpic}[width=0.289\textwidth]{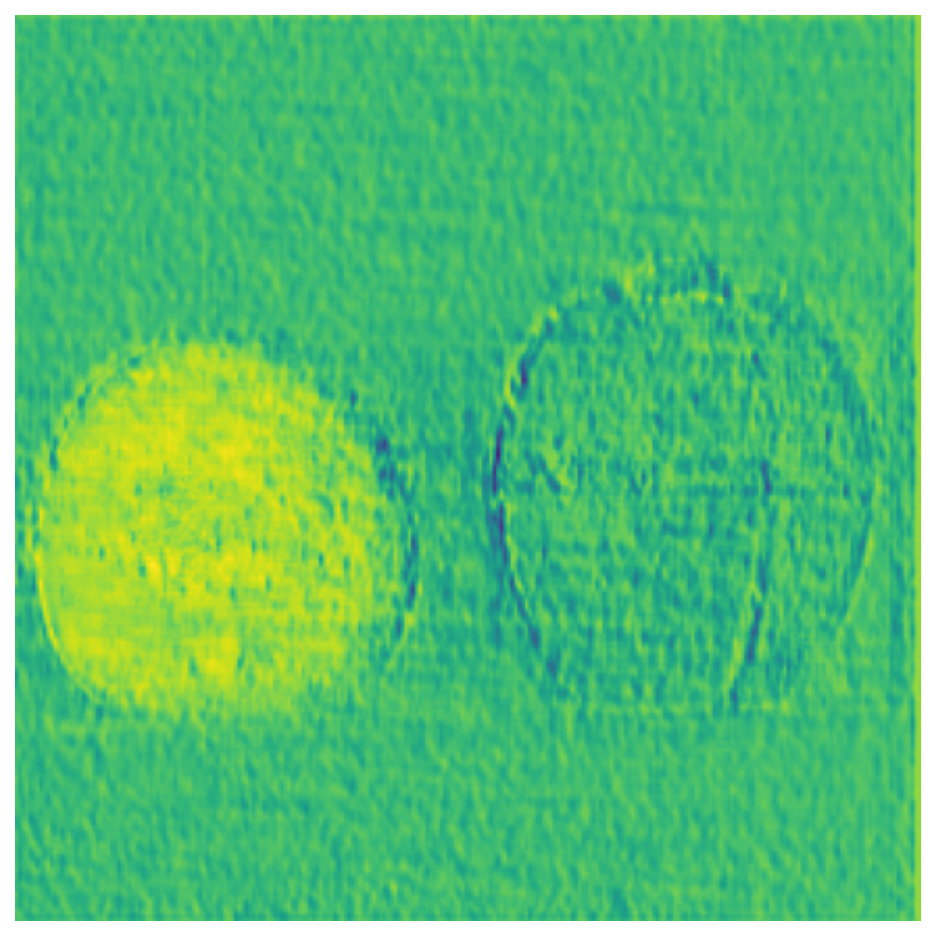}
\put(5,5){}
\end{overpic}&
\begin{overpic}[width=0.289\textwidth]{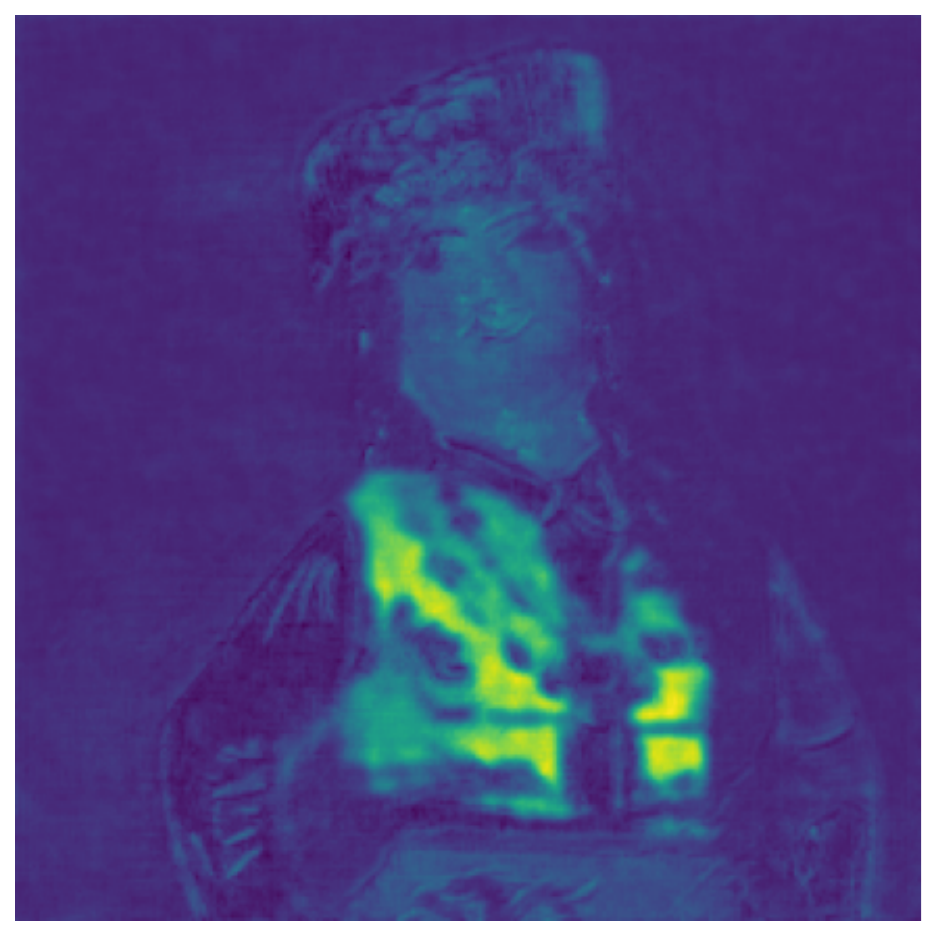}
\put(5,5){}
\end{overpic}&
\begin{overpic}[width=0.289\textwidth]{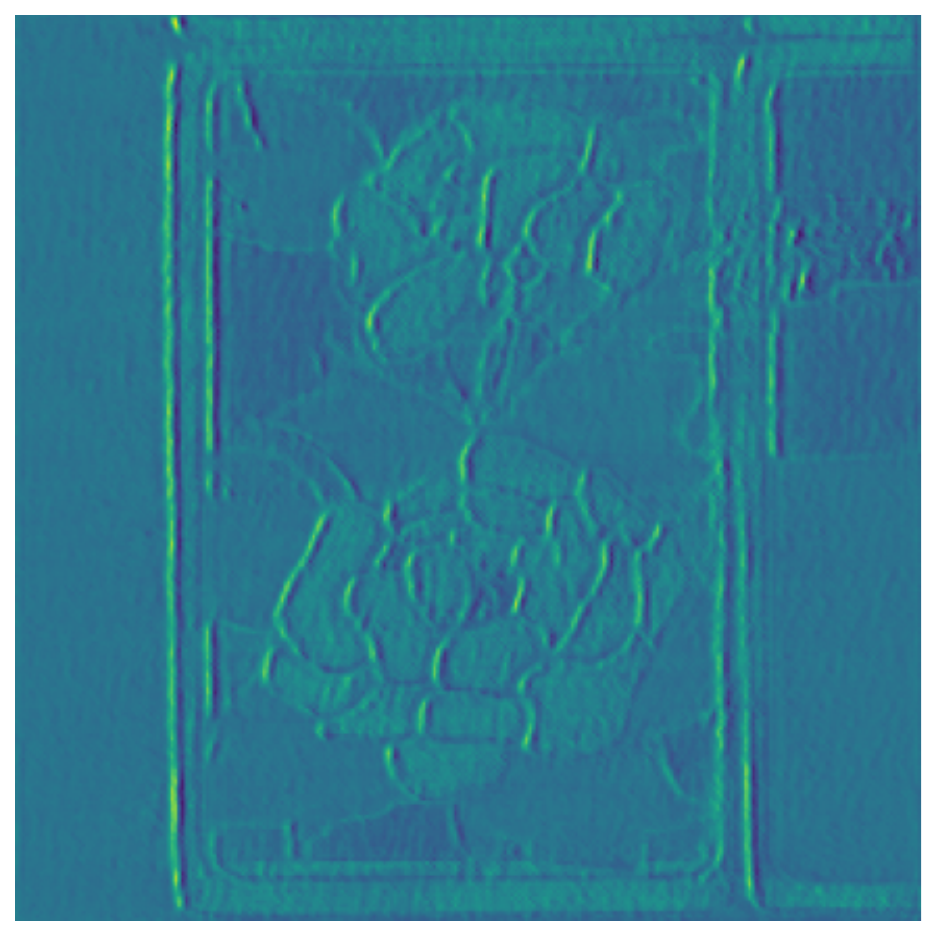}
\put(5,5){}
\end{overpic}&
\begin{overpic}[width=0.289\textwidth]{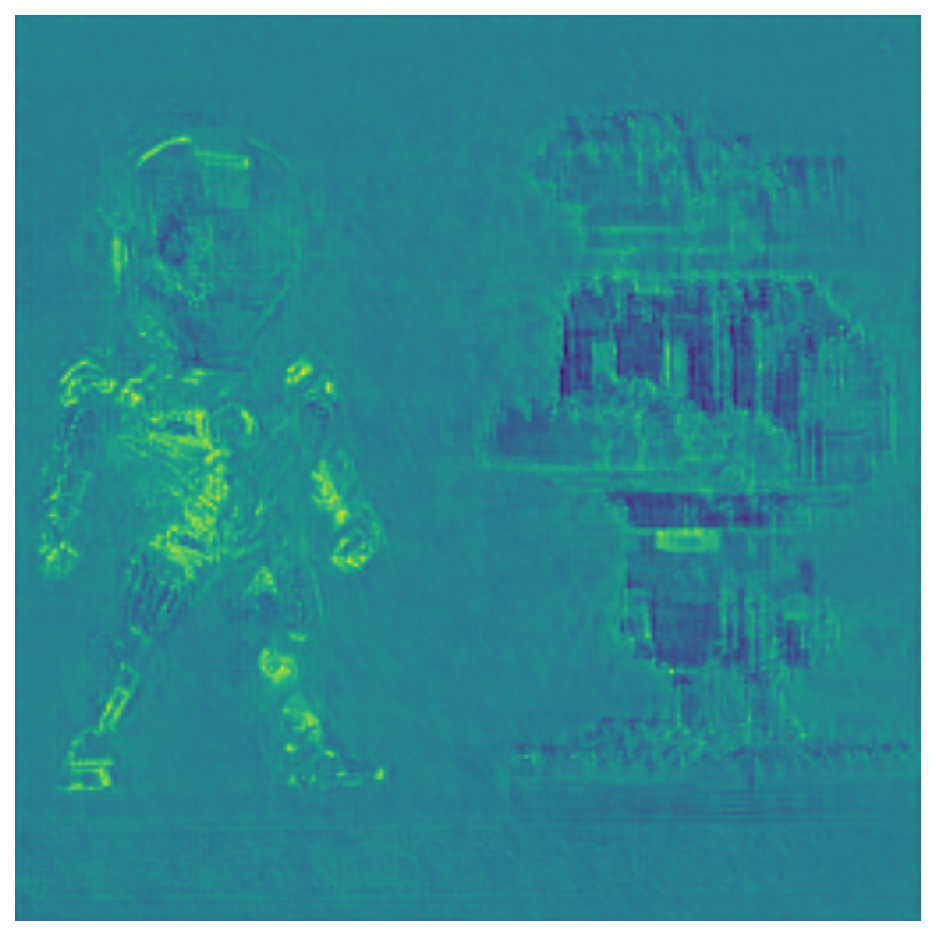}
\put(5,5){}
\end{overpic}&
\begin{overpic}[width=0.289\textwidth]{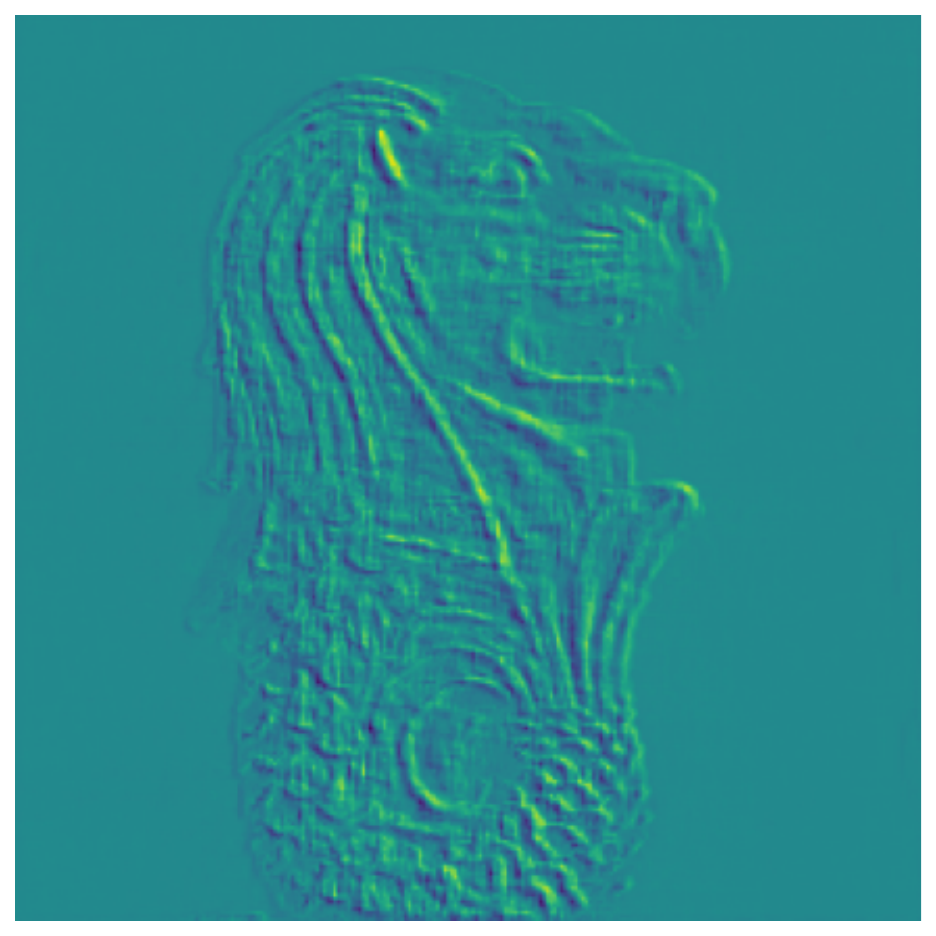}
\put(5,5){}
\end{overpic}
\end{tabular}
\end{adjustbox}
\\
\begin{adjustbox}{valign=t}
\begin{tabular}{ccccccccc}
\Huge\rotatebox{90}{\textbf{\texttt{Spa-MSA}}}&
\begin{overpic}[width=0.289\textwidth]{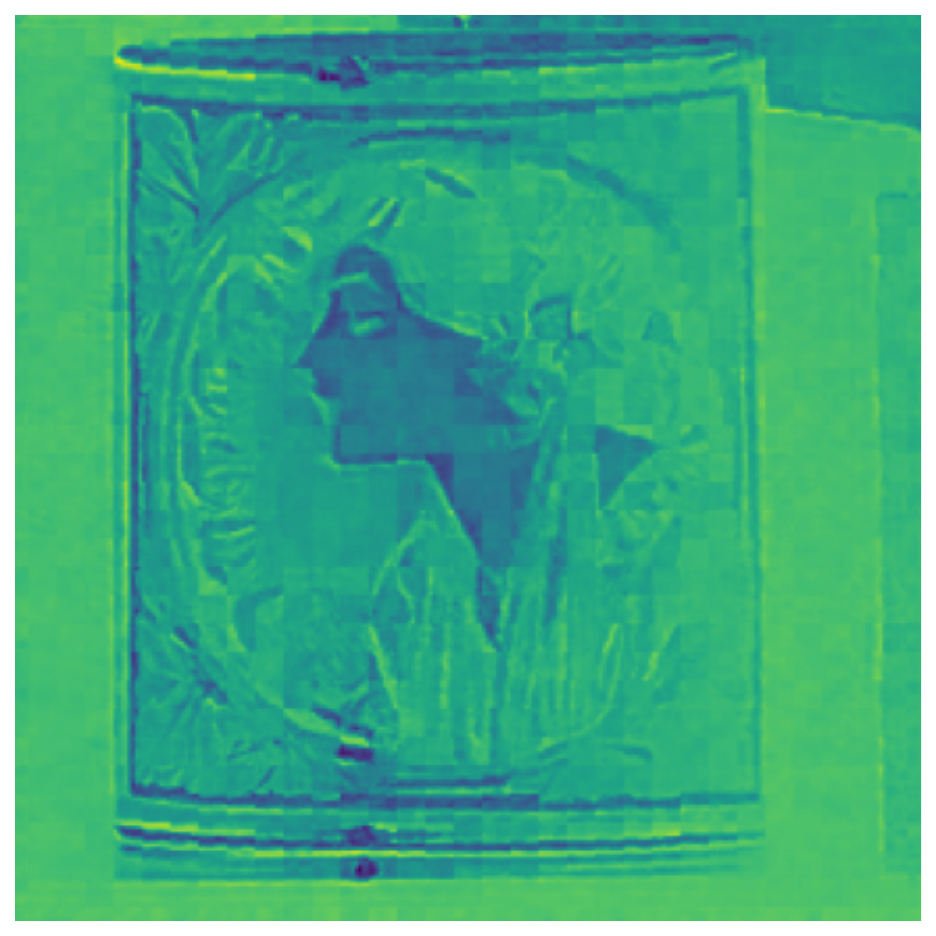}
\put(5,5){}
\end{overpic}&
\begin{overpic}[width=0.289\textwidth]{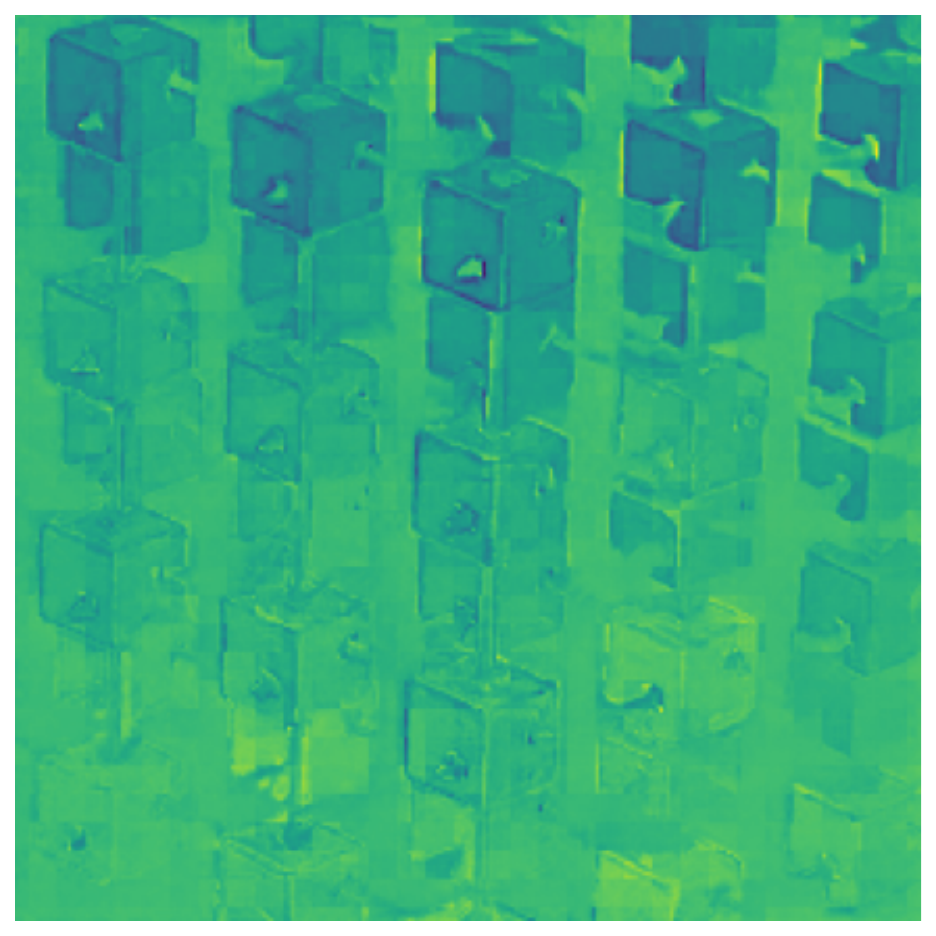}
\put(5,5){}
\end{overpic}&
\begin{overpic}[width=0.289\textwidth]{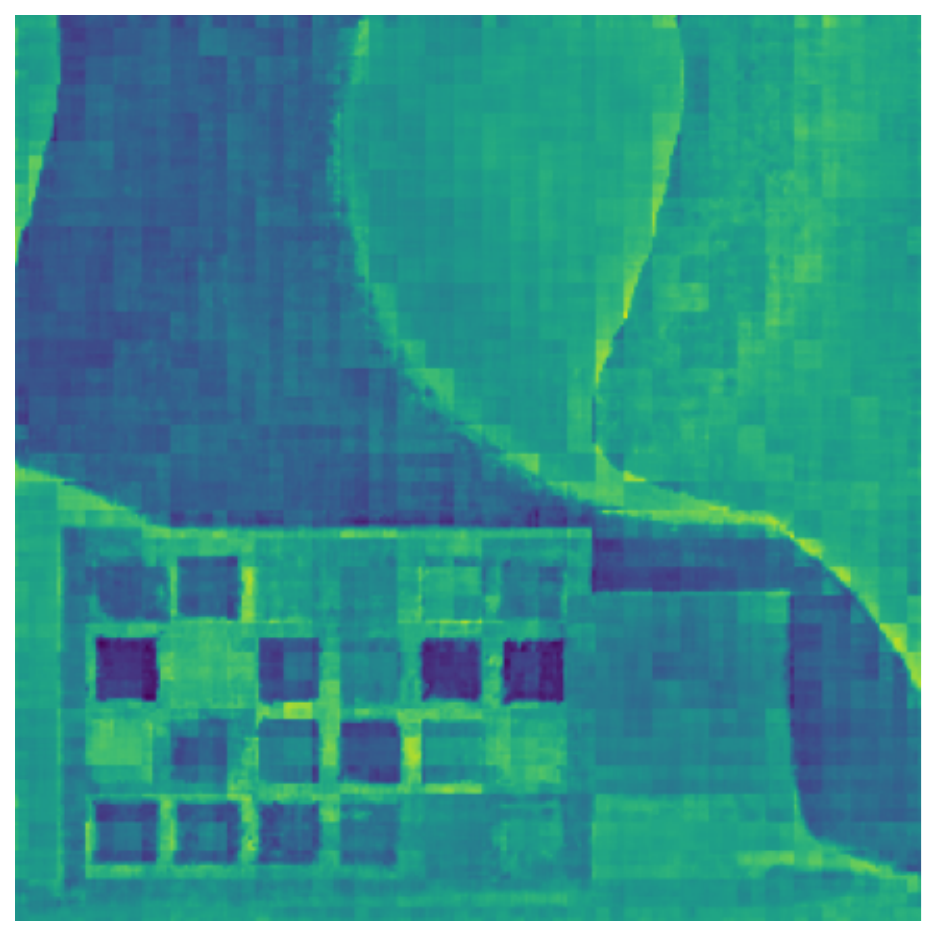}
\put(5,5){}
\end{overpic}&
\begin{overpic}[width=0.289\textwidth]{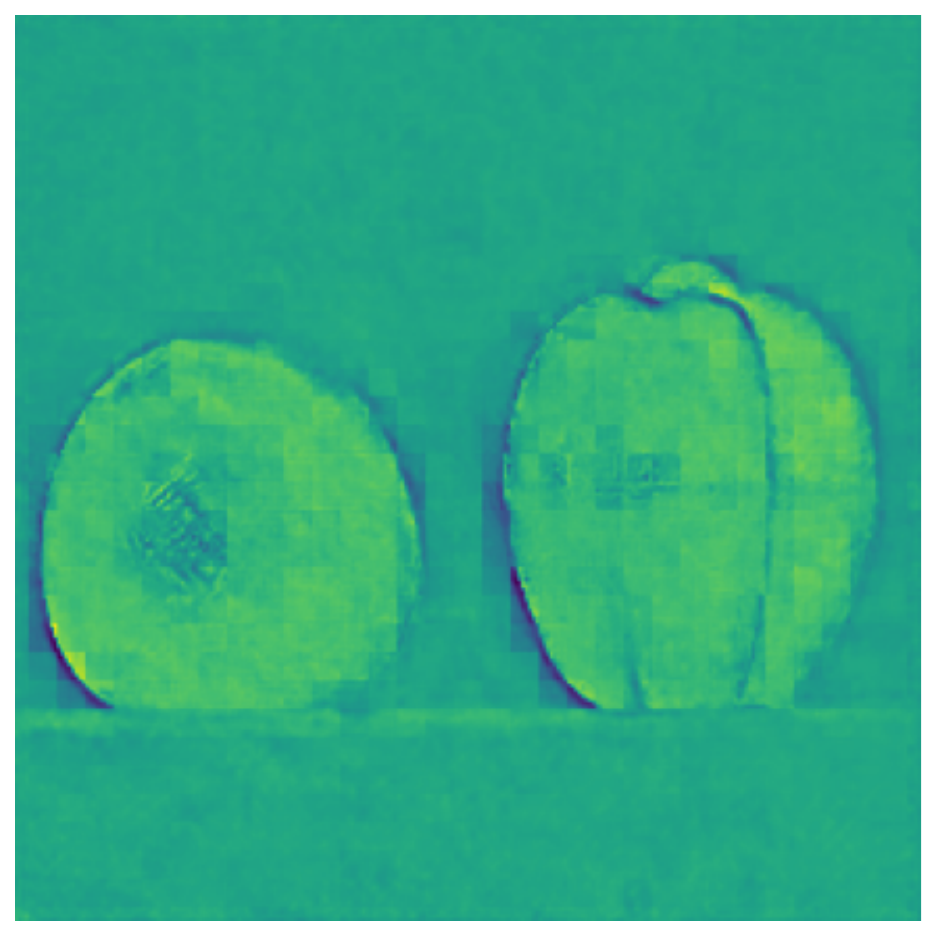}
\put(5,5){}
\end{overpic}&
\begin{overpic}[width=0.289\textwidth]{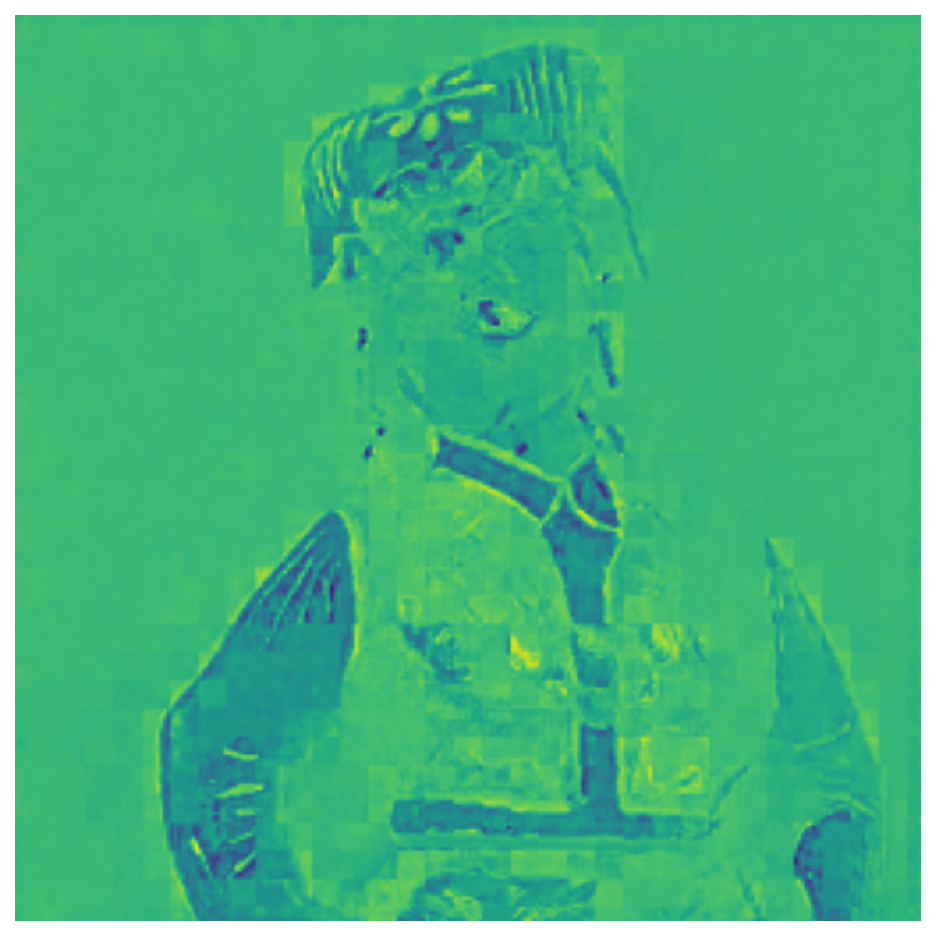}
\put(5,5){}
\end{overpic}&
\begin{overpic}[width=0.289\textwidth]{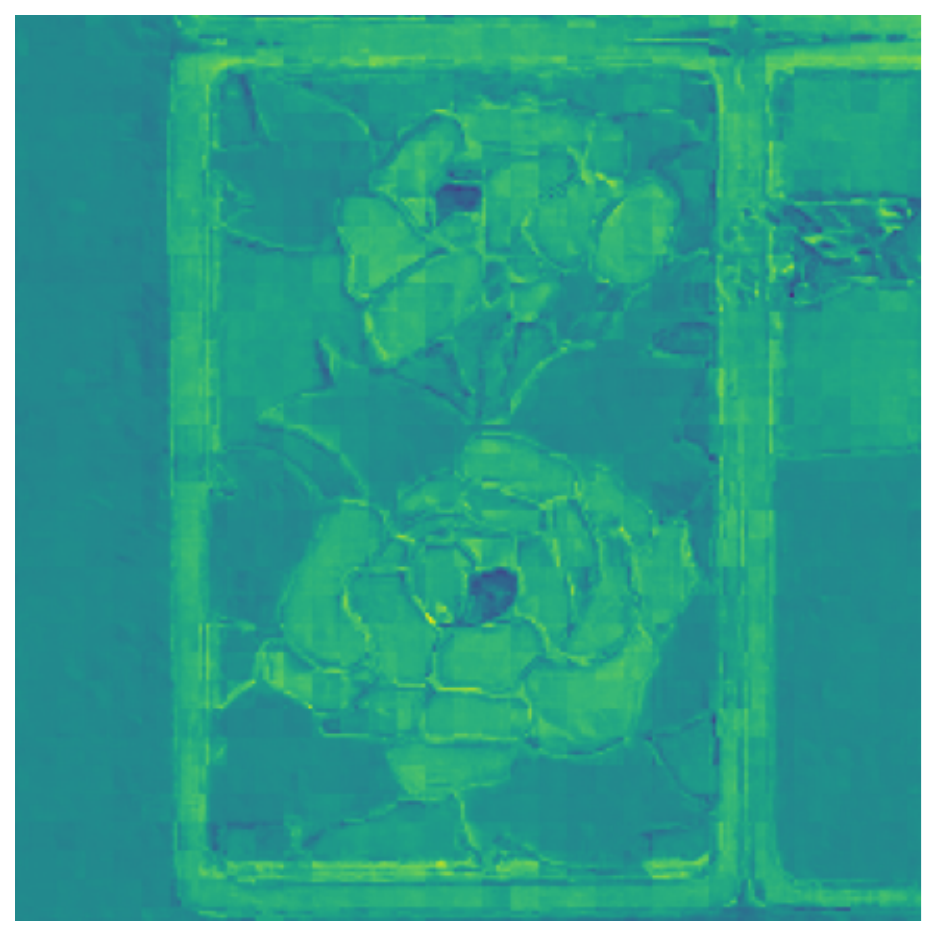}
\put(5,5){}
\end{overpic}&
\begin{overpic}[width=0.289\textwidth]{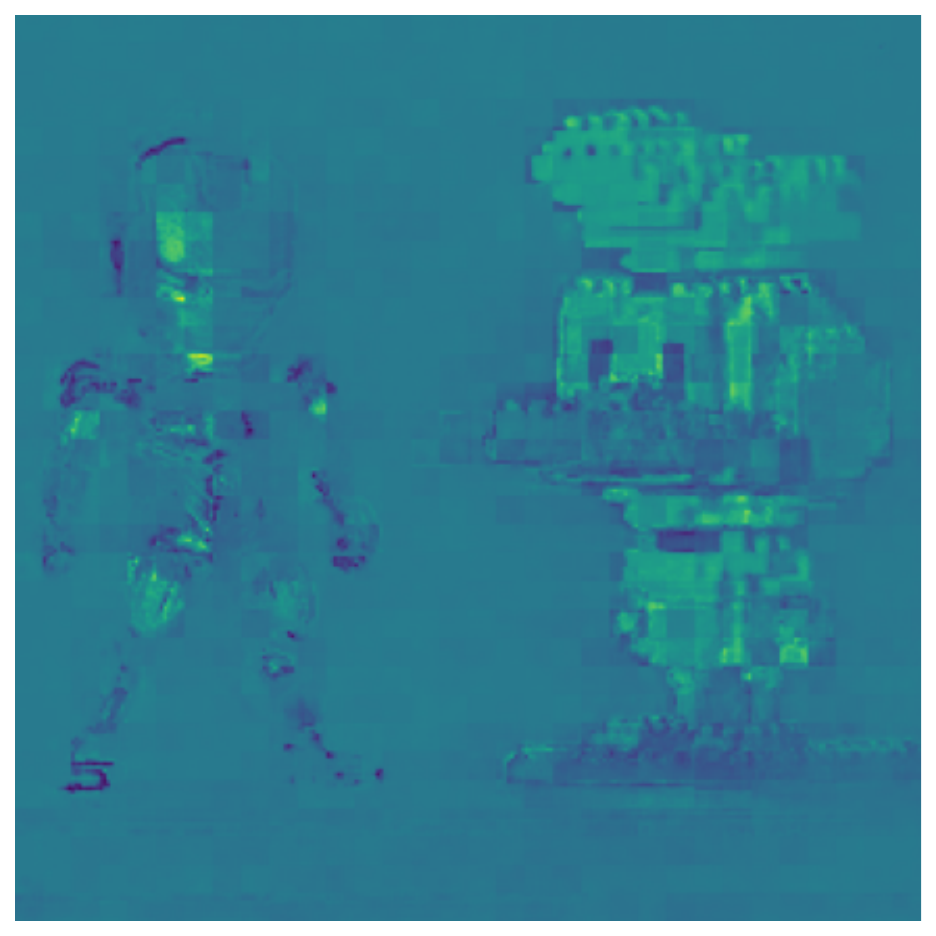}
\put(5,5){}
\end{overpic}&
\begin{overpic}[width=0.289\textwidth]{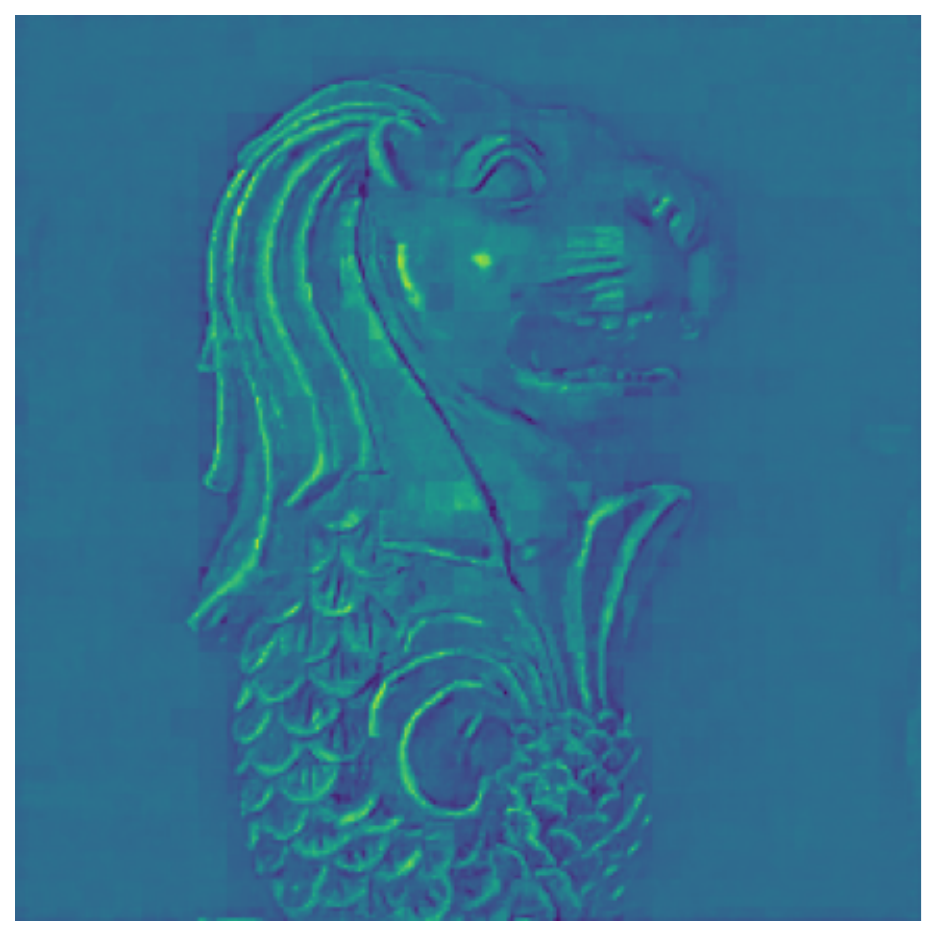}
\put(5,5){}
\end{overpic}
\end{tabular}
\end{adjustbox}
\\
\begin{adjustbox}{valign=t}
\begin{tabular}{ccccccccc}
\Huge\rotatebox{90}{\textbf{RGB}}&
\begin{overpic}[width=0.289\textwidth]{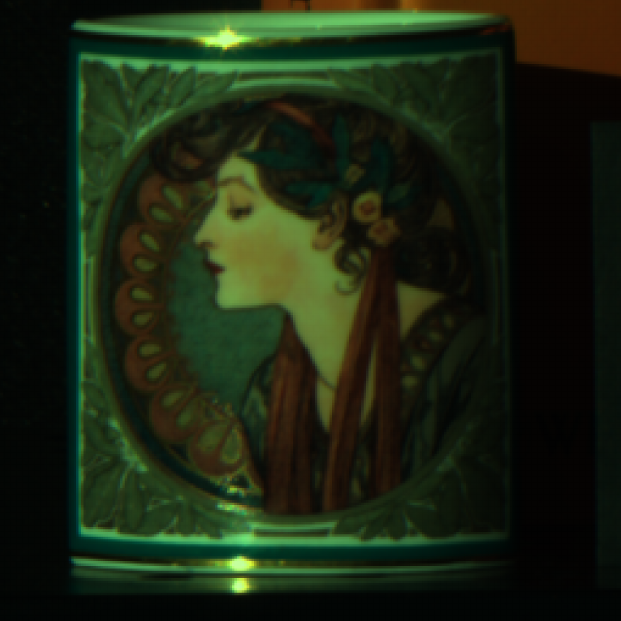}
\put(5,5){\color{white}\Huge Scene 1}
\end{overpic}&
\begin{overpic}[width=0.289\textwidth]{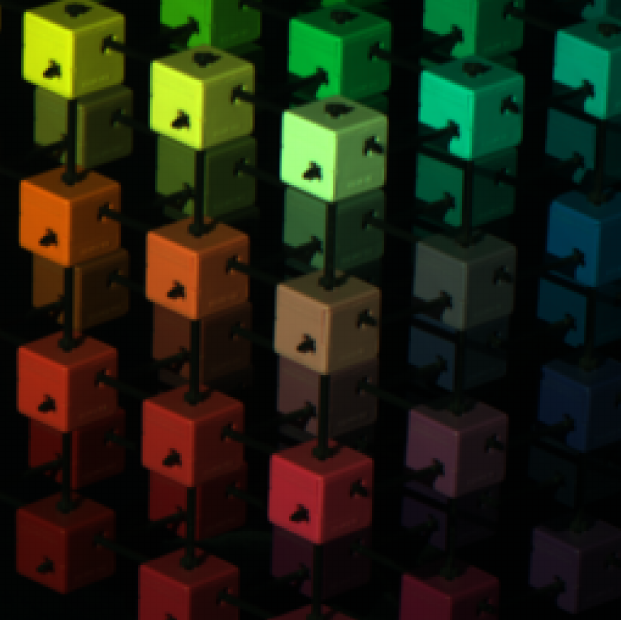}
\put(5,5){\color{white}\Huge Scene 2}
\end{overpic}&
\begin{overpic}[width=0.289\textwidth]{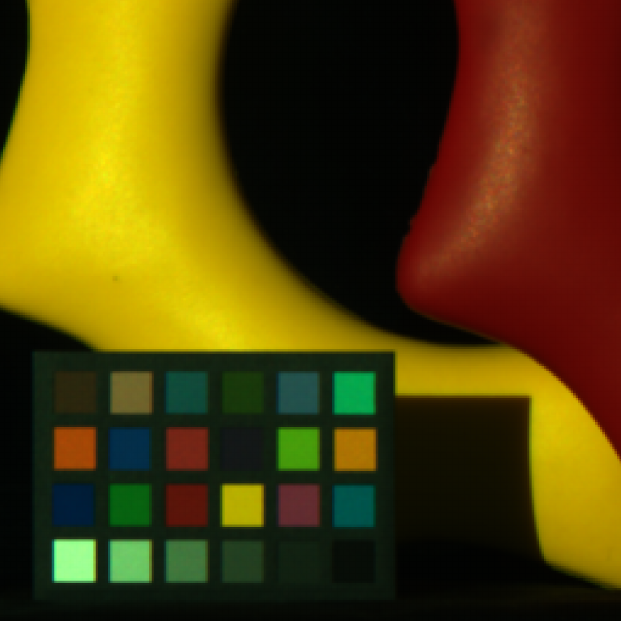}
\put(5,5){\color{white}\Huge Scene 3}
\end{overpic}&
\begin{overpic}[width=0.289\textwidth]{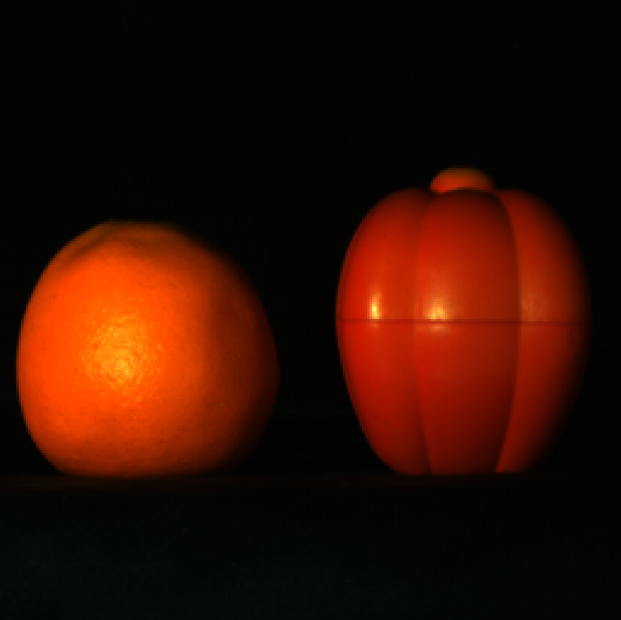}
\put(5,5){\color{white}\Huge Scene 5}
\end{overpic}
&
\begin{overpic}[width=0.289\textwidth]{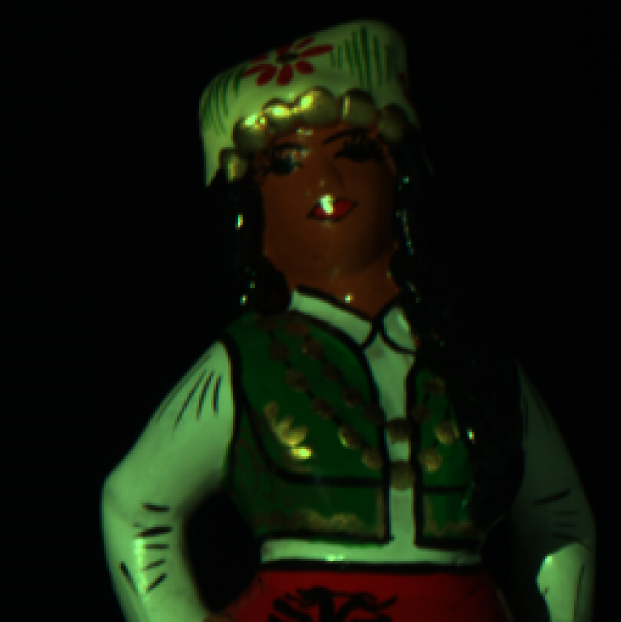}
\put(5,5){\color{white}\Huge Scene 6}
\end{overpic}&
\begin{overpic}[width=0.289\textwidth]{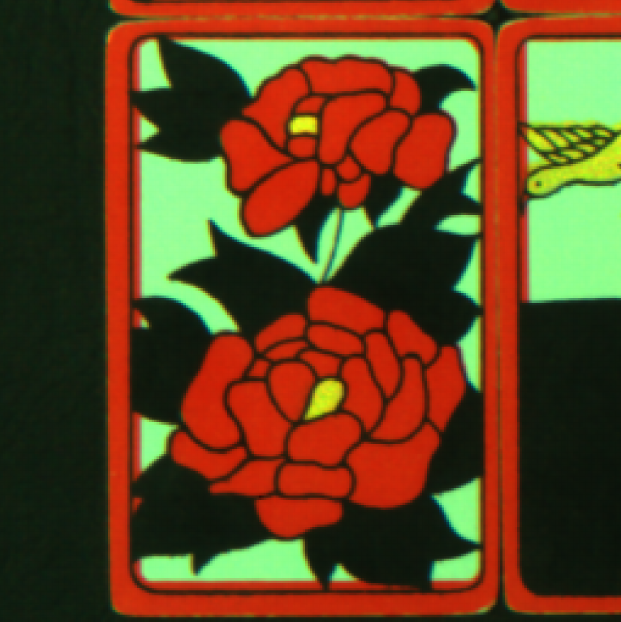}
\put(5,5){\color{white}\Huge Scene 7}
\end{overpic}&
\begin{overpic}[width=0.289\textwidth]{gt8}
\put(5,5){\color{white}\Huge Scene 8}
\end{overpic}&
\begin{overpic}[width=0.289\textwidth]{gt10}
\put(5,5){\color{white}\Huge Scene 10}
\end{overpic}
\end{tabular}
\end{adjustbox}
\end{tabular}
\end{adjustbox}
\hspace{-4mm}
\begin{adjustbox}{valign=t}
\begin{tabular}{c}
\includegraphics[width=0.076\textwidth]{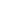}
\\
\includegraphics[width=0.1376\textwidth]{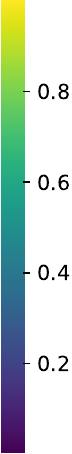}
\end{tabular}
\end{adjustbox}
\end{tabular}}
\vspace{-1mm}
\caption{{
Visualization of the feature embedding within {\textbf{\texttt{Sequn-SS}} blocks (\emph{i.e.}, \texttt{blk})}. Each column provides embedding representations from deep blocks (\emph{i.e.}, 19$\sim$28 out of 36) with selected feature channels (\emph{i.e.}, \texttt{chl} $i$, 1$\le$$i$$\le$$66$). \textit{Top} row: Embedding right after \texttt{Spe-MSA}, which works as an edge detector. \textit{Medium} row: Embedding right after \texttt{Spa-MSA}, which captures visual details of \texttt{Spa-MSA}. \textit{Bottom} row: RGB content.}
}
\label{fig: equn_embedding visal}
\vspace{-2mm}
\end{figure*}

\begin{figure*}[ht]
\footnotesize
\centering
\resizebox{0.97\textwidth}{!}{
\begin{tabular}{cc}
\begin{adjustbox}{valign=t}
\begin{tabular}{c}
\begin{adjustbox}{valign=t}
\begin{tabular}{ccccccccc}
\Huge\rotatebox{90}{\color{white}\textbf{AA}}& 
\begin{overpic}[width=0.289\textwidth]{white.pdf}
\put(7,5){\Huge blk 22/chl 31}
\end{overpic}&
\begin{overpic}[width=0.289\textwidth]{white.pdf}
\put(7,5){\Huge blk 22/chl 19}
\end{overpic}&
\begin{overpic}[width=0.289\textwidth]{white.pdf}
\put(10,5){\Huge blk 24/chl 4}
\end{overpic}&
\begin{overpic}[width=0.289\textwidth]{white}
\put(7,5){\Huge blk 23/chl 29}
\end{overpic}&
\begin{overpic}[width=0.289\textwidth]{white.pdf}
\put(5,5){\Huge blk 23/chl 14}
\end{overpic}&
\begin{overpic}[width=0.289\textwidth]{white.pdf}
\put(8,5){\Huge blk 21/chl 3}
\end{overpic}&
\begin{overpic}[width=0.289\textwidth]{white.pdf}
\put(3,5){\Huge blk 23/chl 20}
\end{overpic}&
\begin{overpic}[width=0.289\textwidth]{white.pdf}
\put(3,5){\Huge blk 23/chl 20}
\end{overpic}
\end{tabular}
\end{adjustbox}
\vspace{-1mm}
\\
\begin{adjustbox}{valign=t}
\begin{tabular}{ccccccccc}
\Huge\rotatebox{90}{\textbf{\texttt{Spe-MSA}}}&
\begin{overpic}[width=0.289\textwidth]{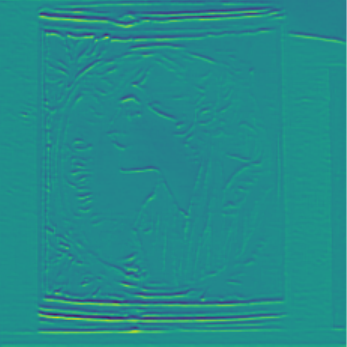}
\put(5,5){}
\end{overpic}&
\begin{overpic}[width=0.289\textwidth]{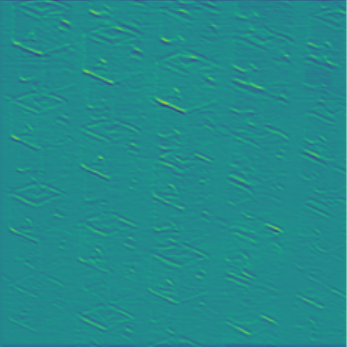}
\put(5,5){}
\end{overpic}&
\begin{overpic}[width=0.289\textwidth]{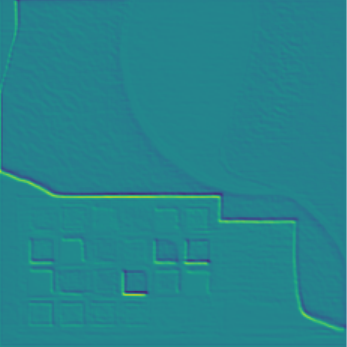}
\put(5,5){}
\end{overpic}&
\begin{overpic}[width=0.289\textwidth]{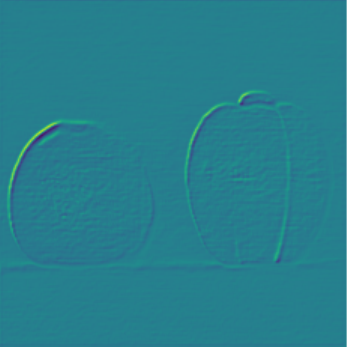}
\put(5,5){}
\end{overpic}&
\begin{overpic}[width=0.289\textwidth]{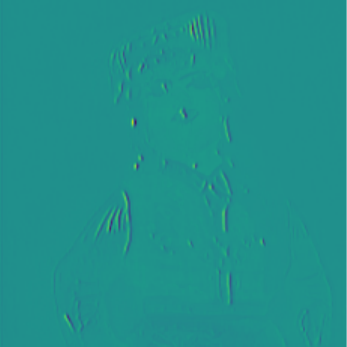}
\put(5,5){}
\end{overpic}&
\begin{overpic}[width=0.289\textwidth]{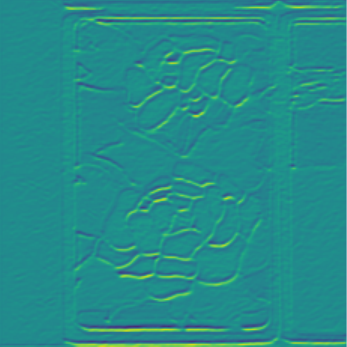}
\put(5,5){}
\end{overpic}&
\begin{overpic}[width=0.289\textwidth]{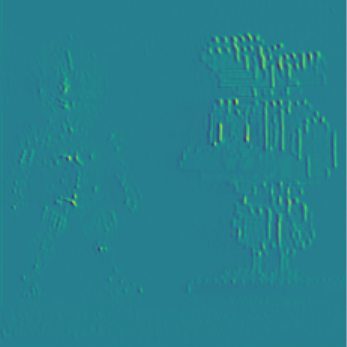}
\put(5,5){}
\end{overpic}&
\begin{overpic}[width=0.289\textwidth]{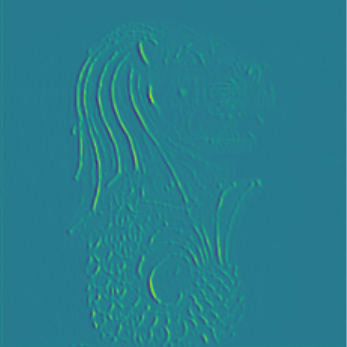}
\put(5,5){}
\end{overpic}
\end{tabular}
\end{adjustbox}
\\
\begin{adjustbox}{valign=t}
\begin{tabular}{ccccccccc}
\Huge\rotatebox{90}{\textbf{\texttt{Spa-MSA}}}&
\begin{overpic}[width=0.289\textwidth]{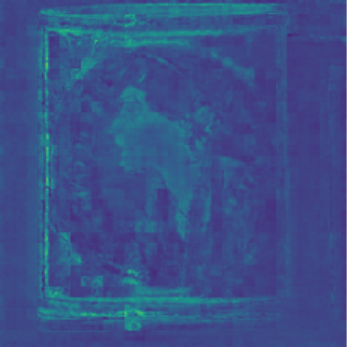}
\put(5,5){}
\end{overpic}&
\begin{overpic}[width=0.289\textwidth]{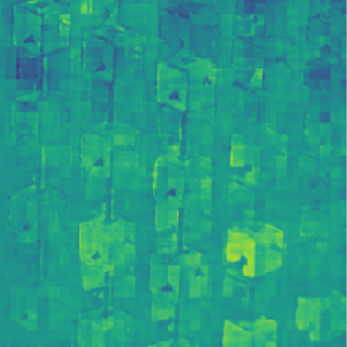}
\put(5,5){}
\end{overpic}&
\begin{overpic}[width=0.289\textwidth]{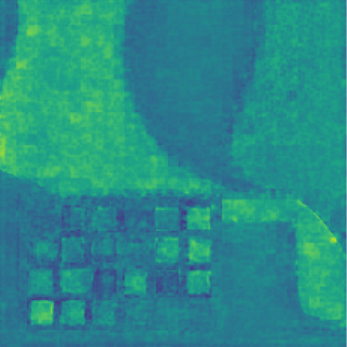}
\put(5,5){}
\end{overpic}&
\begin{overpic}[width=0.289\textwidth]{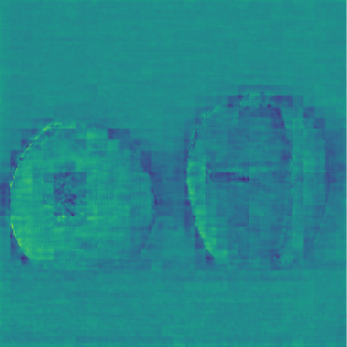}
\put(5,5){}
\end{overpic}&
\begin{overpic}[width=0.289\textwidth]{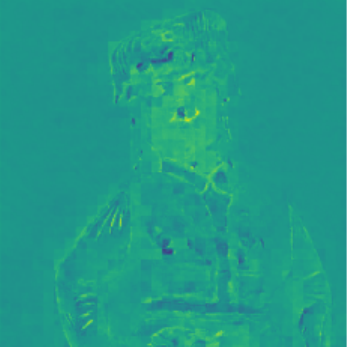}
\put(5,5){}
\end{overpic}&
\begin{overpic}[width=0.289\textwidth]{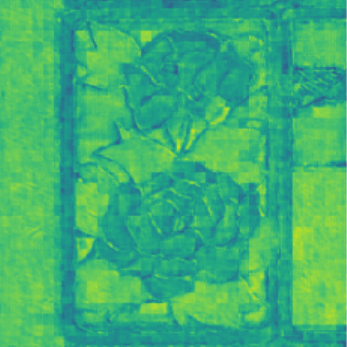}
\put(5,5){}
\end{overpic}&
\begin{overpic}[width=0.289\textwidth]{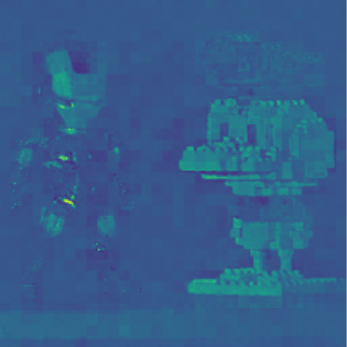}
\put(5,5){}
\end{overpic}&
\begin{overpic}[width=0.289\textwidth]{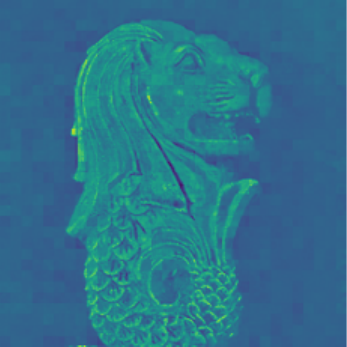}
\put(5,5){}
\end{overpic}
\end{tabular}
\end{adjustbox}
\\
\begin{adjustbox}{valign=t}
\begin{tabular}{ccccccccc}
\Huge\rotatebox{90}{\textbf{Fusion}}&
\begin{overpic}[width=0.289\textwidth]{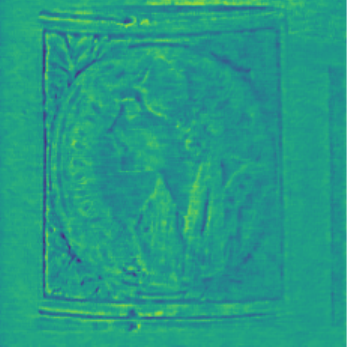}
\put(5,5){}
\end{overpic}&
\begin{overpic}[width=0.289\textwidth]{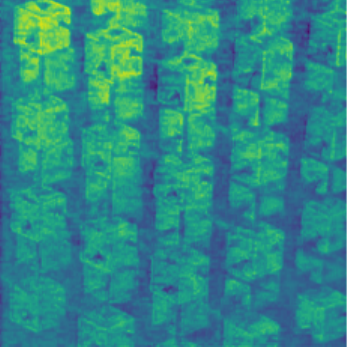}
\put(5,5){}
\end{overpic}&
\begin{overpic}[width=0.289\textwidth]{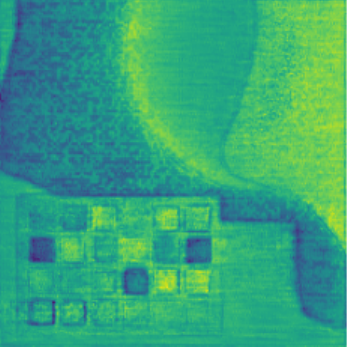}
\put(5,5){}
\end{overpic}&
\begin{overpic}[width=0.289\textwidth]{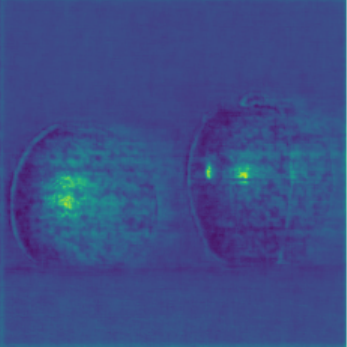}
\put(5,5){}
\end{overpic}&
\begin{overpic}[width=0.289\textwidth]{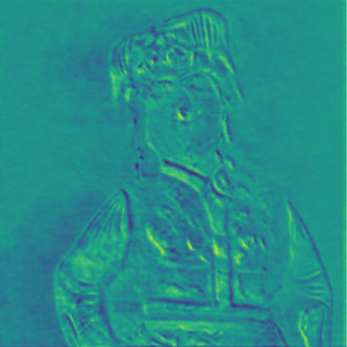}
\put(5,5){}
\end{overpic}&
\begin{overpic}[width=0.289\textwidth]{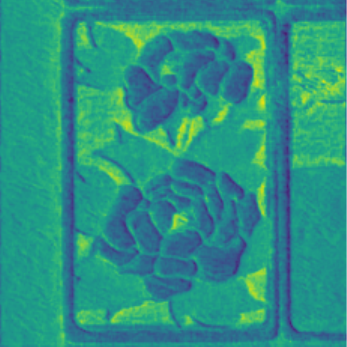}
\put(5,5){}
\end{overpic}&
\begin{overpic}[width=0.289\textwidth]{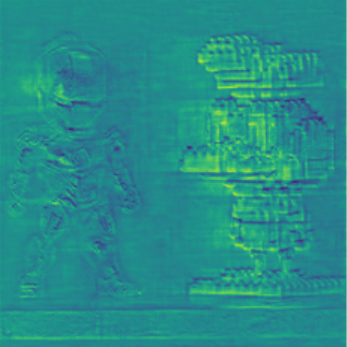}
\put(5,5){}
\end{overpic}&
\begin{overpic}[width=0.289\textwidth]{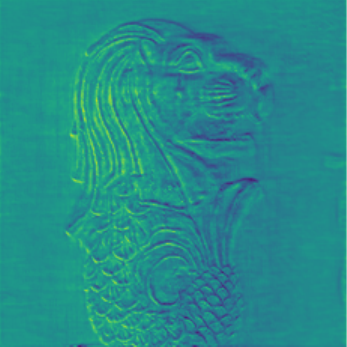}
\put(5,5){}
\end{overpic}
\end{tabular}
\end{adjustbox}
\\
\begin{adjustbox}{valign=t}
\begin{tabular}{ccccccccc}
\Huge\rotatebox{90}{\textbf{RGB}}&
\begin{overpic}[width=0.289\textwidth]{gt1}
\put(5,5){\color{white}\Huge Scene 1}
\end{overpic}&
\begin{overpic}[width=0.289\textwidth]{gt2}
\put(5,5){\color{white}\Huge Scene 2}
\end{overpic}&
\begin{overpic}[width=0.289\textwidth]{gt3}
\put(5,5){\color{white}\Huge Scene 3}
\end{overpic}&
\begin{overpic}[width=0.289\textwidth]{gt4}
\put(5,5){\color{white}\Huge Scene 5}
\end{overpic}
&
\begin{overpic}[width=0.289\textwidth]{gt6}
\put(5,5){\color{white}\Huge Scene 6}
\end{overpic}&
\begin{overpic}[width=0.289\textwidth]{gt7}
\put(5,5){\color{white}\Huge Scene 7}
\end{overpic}&
\begin{overpic}[width=0.289\textwidth]{gt8}
\put(5,5){\color{white}\Huge Scene 8}
\end{overpic}&
\begin{overpic}[width=0.289\textwidth]{gt10}
\put(5,5){\color{white}\Huge Scene 10}
\end{overpic}
\end{tabular}
\end{adjustbox}
\end{tabular}
\end{adjustbox}
\hspace{-4mm}
\begin{adjustbox}{valign=t}
\begin{tabular}{c}
\includegraphics[width=0.073\textwidth]{placeholder.png}
\\ 
\includegraphics[width=0.139\textwidth]{color_bar}
\end{tabular}
\end{adjustbox}
\end{tabular}
}
\vspace{-2mm}
\caption{
Visualization of the feature embedding within {\textbf{\texttt{Parall-SS}} blocks (\emph{i.e.}, \texttt{blk})}. Each column provides embedding representations from deep blocks (\emph{i.e.}, 22$\sim$24 out of 24) with selected feature channels (\emph{i.e.}, \texttt{chl} $i$, 1$\le$$i$$\le$$60$). \textit{Top} row: Embedding right after \texttt{Spe-MSA}, which works as an edge detector. \textit{Upper-medium} row: Embedding right after \texttt{Spa-MSA}, which captures visual details. \textit{Lower-medium} row: Fused embedding of \texttt{Spe-MSA} and \texttt{Spa-MSA}. \textit{Bottom} row: RGB content.}
\label{fig: edge_detector}
\vspace{-4mm}
\end{figure*}

\textbf{Implementation Details.}
The proposed $S^2$-Transformer contains $K$=4 stages, where each consists of $L$=6 $S^2$-attn blocks for a high-fidelity reconstruction performance. We let the number of embedding channels $C$ to be 60 and split them into $T$=6 heads in each $S^2$-attn block. We leverage window partitions to the feature embedding (\emph{i.e.}, window size $M$$=$$8$) and conduct the cyclic shifting following~\cite{liu2021swin}.
The model is trained for 300 epochs with Adam optimizer~\cite{kingma2014adam} ($\beta_1$=0.9, $\beta_2$=0.999). We set the batch size as 4. The initial learning rate is $4\times$$10^{-4}$ and halved every 50 epochs. Notably, the total amount of training epochs remains the same when employing the proposed mask-aware learning strategy. Specifically, the model is firstly pre-trained with $\mathcal{L}_{\texttt{ME}}$ for 150 epochs to get a promising approximation of the encoded signal $\widehat{\mathbf{F}}'$. Then we train the whole model with $\mathcal{L}_{\texttt{MA}}$ for the other 150 epochs. The mask-encoding (\texttt{ME}) weight $\alpha$ is  set as 1.5 in $\mathcal{L}_{\texttt{ME}}$ and then attenuated to 1.0 in $\mathcal{L}_{\texttt{MA}}$, and mask-aware term (\texttt{MA}) is weighted by $\beta$$=$$10$. We employ first half of the network to approximate the encoded signal, \emph{i.e.}, $k_{\texttt{ME}}$$=$$2$. Our experiments are conducted on NVIDIA RTX 3090 GPUs.

\textbf{Highlights on Mask-Aware Learning}. 
{Notably, we adopt the same dataset (detailed in ``Dataset'') and the same model structure (detailed in ``Implementation Details'') for both mask-encoding pre-training  stage and mask-aware learning stage. The only difference between two training stages lies in the loss function.  Specifically, (1) the model is firstly pre-trained with $\mathcal{L}_{\texttt{ME}}$ for $150$ epochs to get a promising approximation of the encoded signal $\widehat{\mathbf{F}}'$. 
Then we train the whole model with $\mathcal{L}_{\texttt{MA}}$ for the other $150$ epochs. (2) The mask-encoding (\texttt{ME}) weight $\alpha$ is  set as $1.5$ in $\mathcal{L}_{\texttt{ME}}$ and then attenuated to $1.0$ in $\mathcal{L}_{\texttt{MA}}$, and mask-aware term (\texttt{MA}) is weighted by $\beta$$=$$10$.
Besides, we keep the total amount of training epochs the same (\emph{i.e.}, $300$ epochs for $S^2$-Transformer) with or without the proposed mask-aware learning strategy.  
}

\begin{figure}[ht]
\footnotesize
\centering
\resizebox{\columnwidth}{!}{
\begin{tabular}{cc}
\begin{adjustbox}{valign=t}
\begin{tabular}{c}
\begin{adjustbox}{valign=t}
\begin{tabular}{cccccc}
\begin{overpic}[width=0.26\textwidth]{gt2}
\put(5,5){\color{white}\Huge Scene 2}
\end{overpic}&
\begin{overpic}[width=0.26\textwidth]{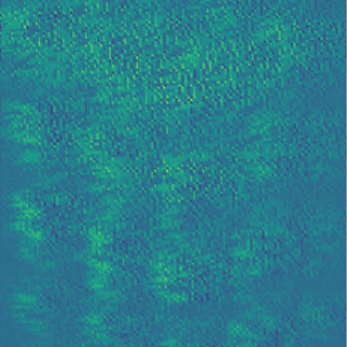}
\put(5,5){\color{white}\Huge blk 1}
\end{overpic}&
\begin{overpic}[width=0.26\textwidth]{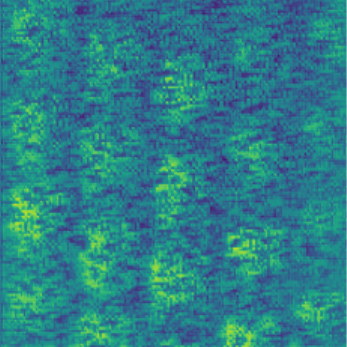}
\put(5,5){\color{white}\Huge blk 7}
\end{overpic}&
\begin{overpic}[width=0.26\textwidth]{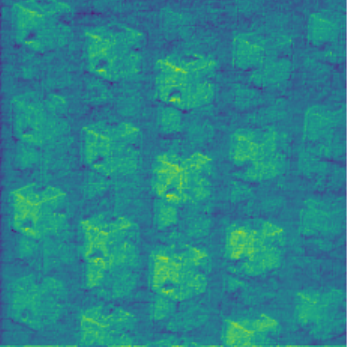}
\put(5,5){\color{white}\Huge blk 13}
\end{overpic}
&
\begin{overpic}[width=0.26\textwidth]{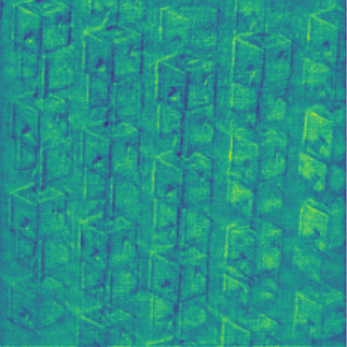}
\put(5,5){\color{white}\Huge blk 18}
\end{overpic}&
\begin{overpic}[width=0.26\textwidth]{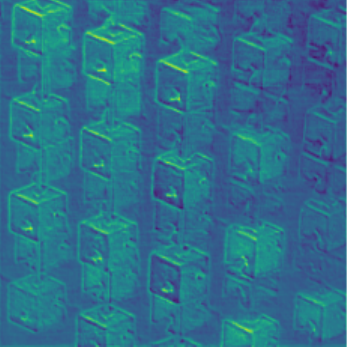}
\put(5,5){\color{white}\Huge blk 24}
\end{overpic}
\end{tabular}
\end{adjustbox}
\\
\begin{adjustbox}{valign=t}
\begin{tabular}{cccccc}
\begin{overpic}[width=0.26\textwidth]{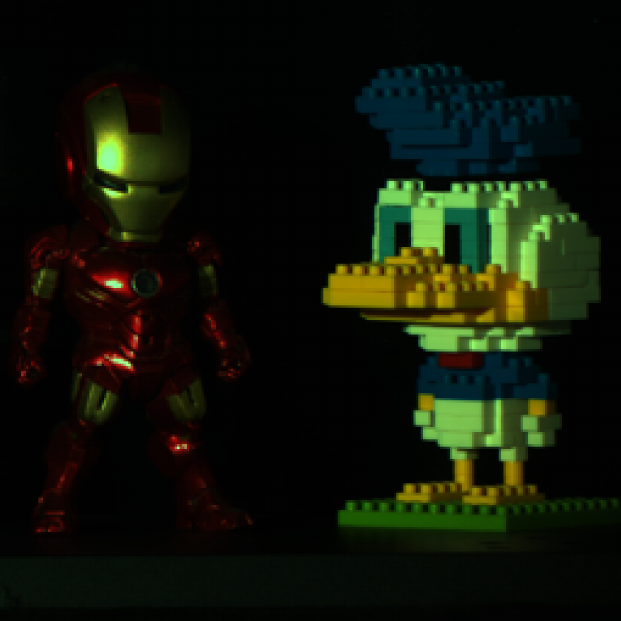}
\put(5,5){\color{white}\Huge Scene 8}
\end{overpic}&
\begin{overpic}[width=0.26\textwidth]{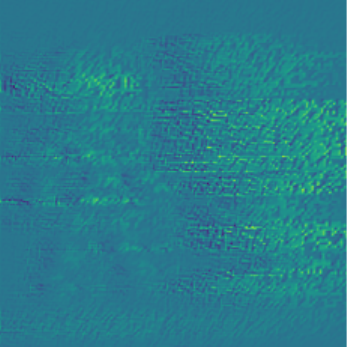}
\put(5,5){\color{white}\Huge blk 1}
\end{overpic}&
\begin{overpic}[width=0.26\textwidth]{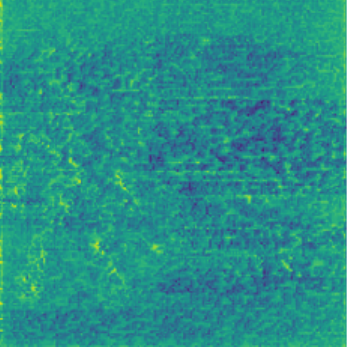}
\put(5,5){\color{white}\Huge blk 5}
\end{overpic}&
\begin{overpic}[width=0.26\textwidth]{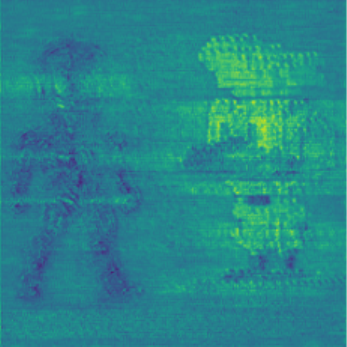}
\put(5,5){\color{white}\Huge blk 13}
\end{overpic}&
\begin{overpic}[width=0.26\textwidth]{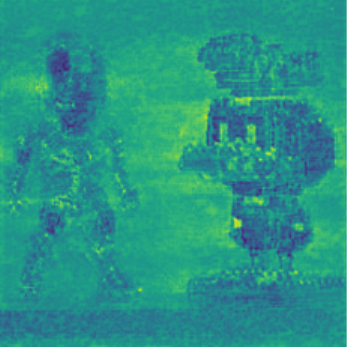}
\put(5,5){\color{white}\Huge blk 21}
\end{overpic}&
\begin{overpic}[width=0.26\textwidth]{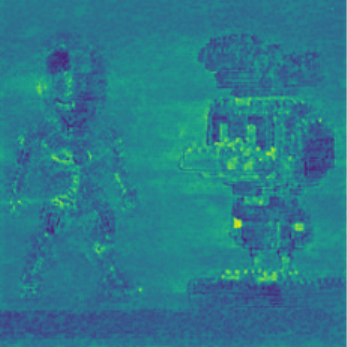}
\put(5,5){\color{white}\Huge blk 24}
\end{overpic}\\
\end{tabular}
\end{adjustbox}
\\
\begin{adjustbox}{valign=t}
\begin{tabular}{cccccc}
\begin{overpic}[width=0.26\textwidth]{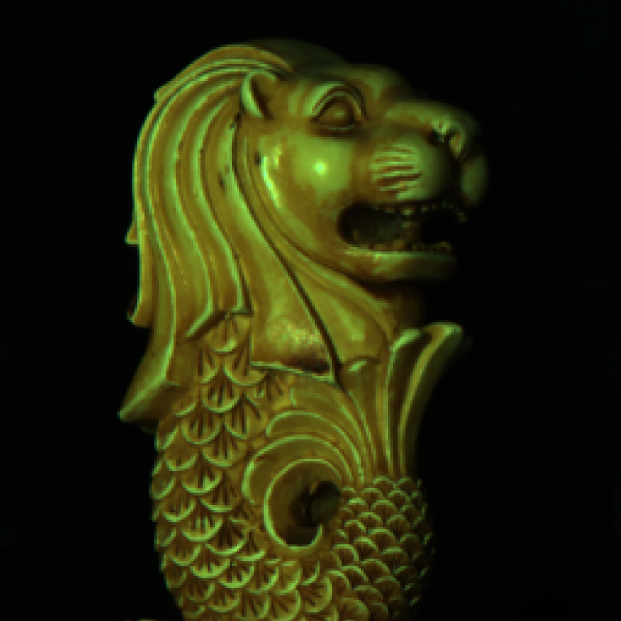}
\put(5,5){\color{white}\Huge Scene 10}
\end{overpic}&
\begin{overpic}[width=0.26\textwidth]{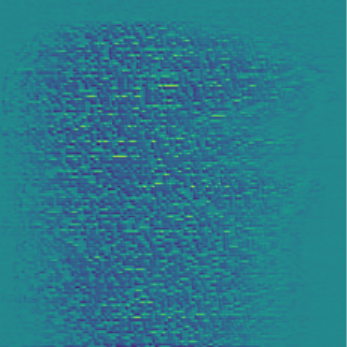}
\put(5,5){\color{white}\Huge blk 1}
\end{overpic}&
\begin{overpic}[width=0.26\textwidth]{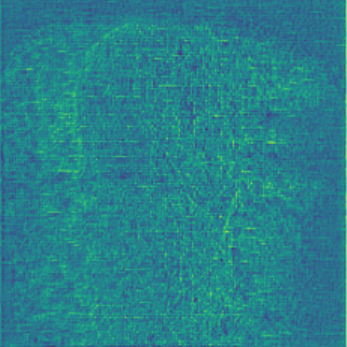}
\put(5,5){\color{white}\Huge blk 10}
\end{overpic}&
\begin{overpic}[width=0.26\textwidth]{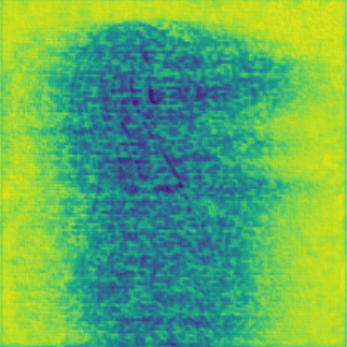}
\put(5,5){\color{white}\Huge blk 14}
\end{overpic}&
\begin{overpic}[width=0.26\textwidth]{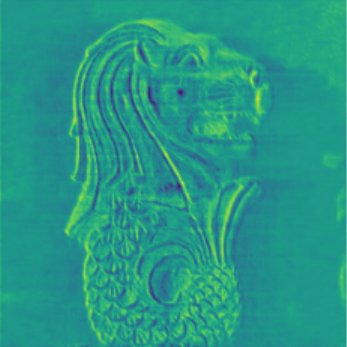}
\put(5,5){\color{white}\Huge blk 22}
\end{overpic}&
\begin{overpic}[width=0.26\textwidth]{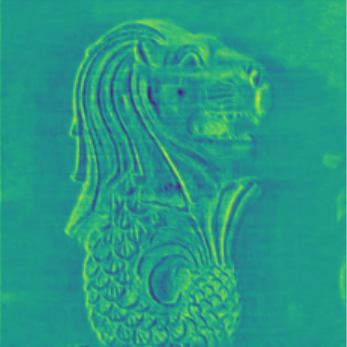}
\put(5,5){\color{white}\Huge blk 24}
\end{overpic}
\end{tabular}
\end{adjustbox}
\end{tabular}
\end{adjustbox}
\begin{adjustbox}{valign=t}
\begin{tabular}{c}
\includegraphics[width=0.125\textwidth]{color_bar}
\end{tabular}
\end{adjustbox}
\end{tabular}}
\caption{
Visualization of the fused embedding at different \texttt{Parall-SS} blocks (\emph{i.e.}, \texttt{B}). Embedding $Z$$\in$$\mathbb{R}^{H\times W \times C}$ gathers information from spatial and spectral attentions. Representations in deeper blocks (\emph{i.e.}, 24 in total) approaches the semantic meaning. We choose 37, 10, 59-th channels out of 60 for scene 2, 8, 10, respectively. }
\label{fig: fused_feature}
\vspace{-2mm}
\end{figure}

\begin{figure}[t]
\scriptsize
\centering
\resizebox{\columnwidth}{!}{
\begin{tabular}{cc}
\begin{adjustbox}{valign=t}
\begin{tabular}{c}
\begin{adjustbox}{valign=t}
\begin{tabular}{cccccc}
\includegraphics[width=0.26\textwidth]{gt1}  &
\includegraphics[width=0.26\textwidth]{gt2}  &
\includegraphics[width=0.26\textwidth]{gt3} &
\includegraphics[width=0.26\textwidth]{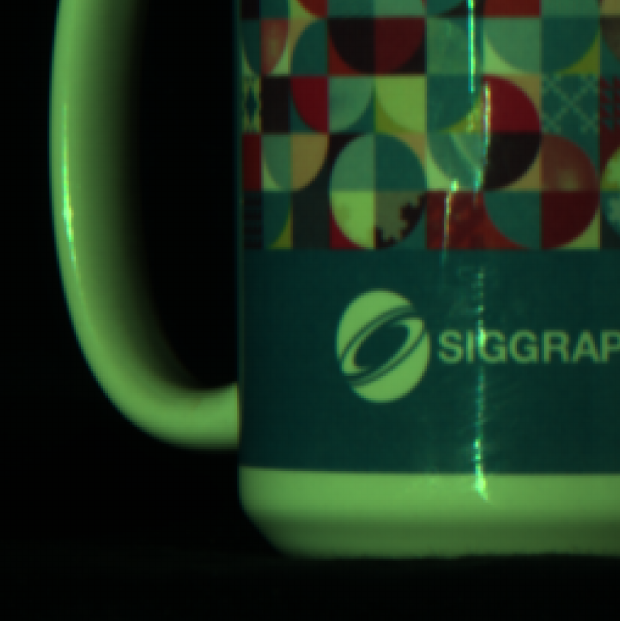}  &
\includegraphics[width=0.26\textwidth]{gt7}  &
\includegraphics[width=0.26\textwidth]{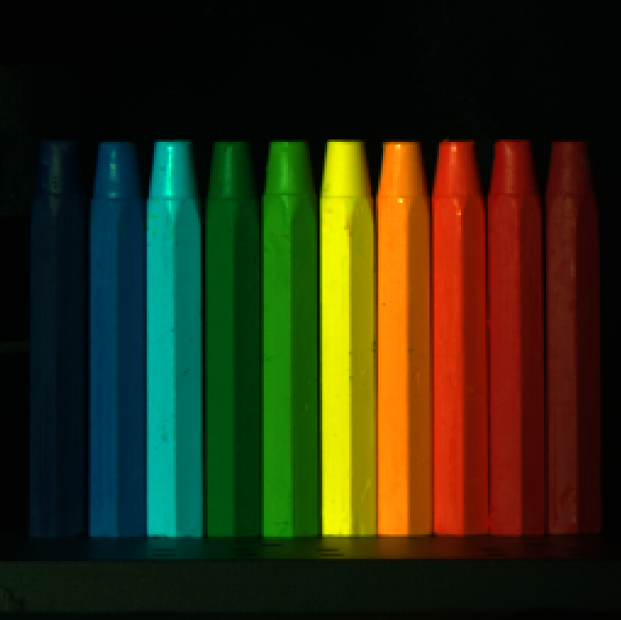}  
\end{tabular}
\end{adjustbox}
\\
\begin{adjustbox}{valign=t}
\begin{tabular}{cccccc}
\includegraphics[width=0.26\textwidth]{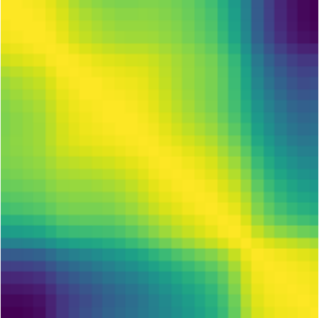}  &
\includegraphics[width=0.26\textwidth]{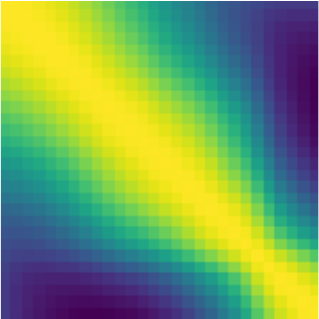}   &
\includegraphics[width=0.26\textwidth]{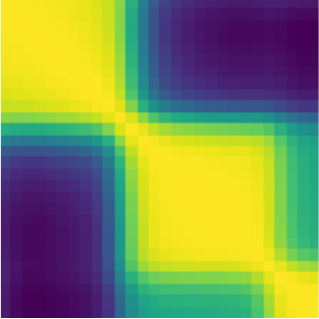}    &
\includegraphics[width=0.26\textwidth]{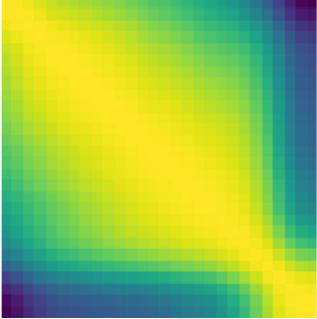}    &
\includegraphics[width=0.26\textwidth]{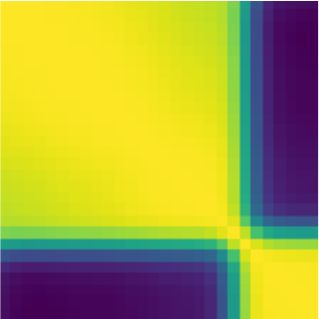} &
\includegraphics[width=0.26\textwidth]{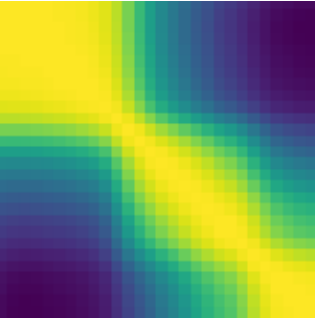} 
\end{tabular}
\end{adjustbox}
\\
\begin{adjustbox}{valign=t}
\begin{tabular}{cccccc}
\includegraphics[width=0.26\textwidth]{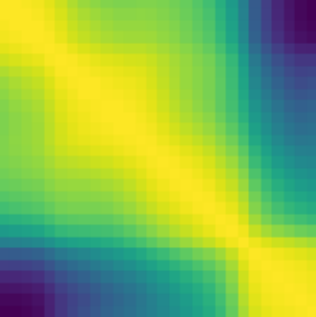}  &
\includegraphics[width=0.26\textwidth]{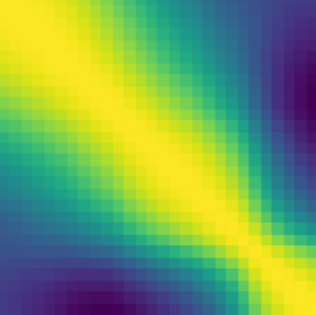}   &
\includegraphics[width=0.26\textwidth]{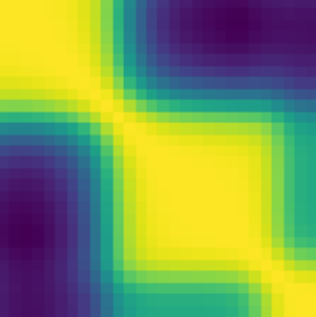}    &
\includegraphics[width=0.26\textwidth]{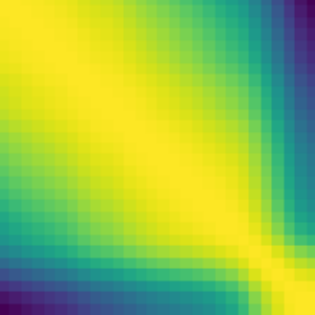}    &
\includegraphics[width=0.26\textwidth]{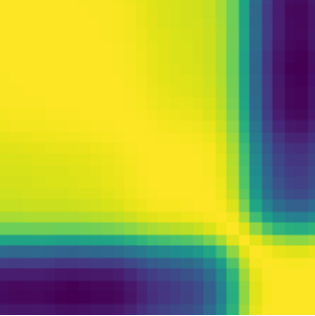} &
\includegraphics[width=0.26\textwidth]{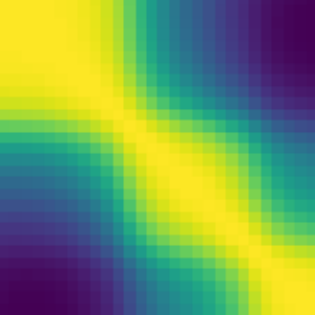}   
\end{tabular}
\end{adjustbox}
\end{tabular}
\end{adjustbox}
\hspace{-4mm}
\begin{adjustbox}{valign=t}
\begin{tabular}{c}
\includegraphics[width=0.125\textwidth]{color_bar}
\end{tabular}
\end{adjustbox}
\end{tabular}}
\vspace{-2mm}
\caption{
Visualization of the spectral correlation. \textit{Top}: RGB contents of the benchmark testing data. \textit{Middle}: spectral correlation coefficient matrices of the ground truth ($28$$\times$$28$). \textit{Bottom}: Corresponding matrices by the proposed method.}
\label{fig: spectral_corr}
\vspace{-4mm}
\end{figure}

\textbf{Compared Methods.}
We compare with eight state-of-the-art methods. 
Among them, 
DeSCI~\cite{liu2018rank}, and GAP-TV~\cite{yuan2016generalized} are model-based algorithms. 
The CNN-based methods include TSA-Net~\cite{Meng20ECCV_TSAnet}, SRN~\cite{wang2021new}, DGSMP~\cite{huang2021deep}, HDNet~\cite{hu2022hdnet}, MST~\cite{cai2022mask}, and CST~\cite{cai2022coarse}, in which MST and CST are recent Transformer-based network designs. Notably, we test and report the best performance of compared methods based on their open-sourced pre-trained models. The SRN is re-trained with the same training dataset as other methods for a fair comparison. 
Besides, we choose the best-performed variants of the compared methods, \emph{i.e.}, \texttt{SRN(v1)}, \texttt{MST-L}, and \texttt{CST-L-Plus}. For the simulation experiment, we conduct the quantitative comparison with the full-reference image quality assessment metrics, PSNR and SSIM~\cite{wang2004image}. For the real data evaluation, we adopt the no-reference image quality assessment metric, Naturalness Image Quality Evaluator (NIQE)~\cite{mittal2012making} due to the inaccessibility of the ground truth.

\subsection{HSI Reconstruction Performance}
\label{subsec: performance}

\textbf{Simulation}. 
We metrically evaluate the proposed $S^2$-Transformer with mask-aware learning strategy by comparing with other popular reconstruction methods. 
As shown in Table~\ref{Tab: psnr_single_trn} and Table~\ref{Tab: ssim_single_trn}, 
the proposed method outperforms CST~\cite{cai2022coarse} by 0.4dB/0.0018 in terms of PSNR/SSIM
and achieves 1.30dB/0.0108 improvement in contrast to MST~\cite{cai2022mask}.
We further analyze the computational complexity and the model size of different methods in Table~\ref{Tab: complexity}. Our $S^2$-Transformer requires the smallest computational overhead among Transformer-based methods in terms of the model size and the number of floating-point-operations (FLOPs), owning to the linear computational complexities of the window-based attention mechanisms. 
Moreover, we perceptually analyze the reconstruction performance of the proposed method.
In Fig.~\ref{Fig: simu_result}, we visualize the reconstruction results on three representative wavelengths. 
Our method enables more content retrieval on semantic areas and less distortions by the enlarged window. We also investigate the spectral fidelity of different methods by the density-wavelength curves on the informative area (\emph{e.g.}, \texttt{patch a} on the hyperspectral image). 
A higher curve correlation presents a better spectral modeling within the chosen narrowband.  Besides, we also noticed the proposed method may not improve over existing methods on challenging cases, \emph{e.g.}, $37.29$ ($S^2$-Transformer) \textit{v.s.} $38.20$ (CST) on Scene 3 in Table~\ref{Tab: psnr_single_trn}. There exists a complex spectra (various colors) in a small region where the proposed method improves partially. By comparison, CST selects and clusters the representative patches, allowing a better color detection in extreme cases. Note that the superiority of CST is at the cost of a much higher complexity (Table~\ref{Tab: complexity}). Overall, the proposed method achieves a better complexity-performance trade-off via
a well-founded attention design. Table~\ref{Tab: structure ablation} quantitatively demonstrates the better performance can be expected by exchanging
the attention type.

\begin{table}[h] 
\footnotesize
\caption{ Ablation study on mask-aware learning. Note that $\mathcal{L}_{\texttt{ME}}\rightarrow\mathcal{L}_{\texttt{MA}}$ is our full method that first pre-train the model with  $\mathcal{L}_{\texttt{ME}}$ and then employ $\mathcal{L}_{\texttt{MA}}$. We adopt the same dataset and the same model structure for the pre-training  ($\mathcal{L}_{\texttt{ME}}$) and mask-aware learning ($\mathcal{L}_{\texttt{MA}}$). All learning strategies takes the same amount of total training epochs for a fair comparison.}
\vspace{-1mm}
\label{tab: loss_ablation}
\centering
\resizebox{.48\textwidth}{!}{
\centering
\begin{tabular}{c|ccc|cc} 
    \hline
    \multirow{2}{*}{Learning strategies} & \multicolumn{3}{c|}{Loss terms} & \multirow{2}{*}{PSNR(dB)} & \multirow{2}{*}{SSIM} \\
    \cline{2-4}
    & \texttt{Recon} & \texttt{ME} & \texttt{MA}  & & \\
    \hline
    $\mathcal{L}_1$ (existing)  & \ding{51} & \ding{55}  & \ding{55} & 36.31 & 0.9569\\
    $\mathcal{L}_{\texttt{ME}}$ & \ding{51} & \ding{51} & \ding{55}  & 36.32 & 0.9569\\
    $\mathcal{L}_{\texttt{MA}}$ & \ding{51} & \ding{51} & \ding{51}  & 34.88 & 0.9496\\
    \hline
    $\mathcal{L}_{\texttt{ME}}$ $\rightarrow$ $\mathcal{L}_{\texttt{MA}}$ (full model) & \ding{51} & \ding{51} & \ding{51}  & \textbf{36.48} & \textbf{0.9584}\\
    \hline
\end{tabular}}
\vspace{-0.4cm}
\end{table}

\noindent\textbf{Real Measurement Retrieval}. We further evaluate the proposed method on real-captured measurements. In Fig.~\ref{Fig: real_result}, we visualize the reconstruction results of compared methods.
By observation, the Transformer-based methods provide the most visually pleasant results among all, but with different emphasizes. For example, the CST provides a better contrast while lose contents at dark area. The proposed method retrieves the most contents (\emph{e.g.}, the smallest hole at center of the yellow flower in top row). In Table~\ref{Tab: real metric}, we metrically compare the Transformer-based methods with the no-reference assessment, naturalness image quality evaluator (NIQE)~\cite{mittal2012making}. A smaller value indicates a better global retrieval performance of the proposed method.

\subsection{Spatial-spectral Attentions}
\label{subsec: spatial-spectral attn}

{\textbf{$S^2$-attn Analysis}.
We systematically investigate different types of spatial-spectral attentions and \textbf{propose the \texttt{Parall-SS} block as our method}.  For a comparable computational cost, we enhance the \texttt{Spa} and \texttt{Spe} blocks as shown by the dash boxes in Fig.~\ref{fig: ss_attn}. 
Table~5 reports the model size (\#params), computational complexity by floating point operations (FLOPs), and reconstruction performances. 
(1) \texttt{SpeSpe} block is the most limited among all architecture designs. This is because spectral attention only describes the embedding channel interplay, which heavily relies on the pre-determined \textit{wavelengths}. Spectral attention fails to abstract the fine-grained visual details in each spectral channel. Similar to the proposed \texttt{SpeSpe} attention, MST~\cite{cai2022mask} demonstrates a similar performance under a comparable complexity. (2)  The \texttt{SpaSpa} block outperforms the \texttt{SpeSpe} block. Since spatial attention computes the long-range dependencies in the spatial domain, which benefits the visually grounded semantics modeling and therefore allows a more granular reconstruction than the spectral attention. (3)  Next, \texttt{Sequn-SS} block observes a slight performance boost over the single spatial/spectral attention, which demonstrates the potential of joint modeling. However, since the spatial and spectral attentions are directly concatenated in each block, there lacks an interaction between these two modalities.  (4)  Finally, \texttt{Parall-SS} block yields the state-of-the-art result, \emph{i.e.}, $36.48$dB/$0.9584$ in terms of the PSNR/SSIM, specifically, the \texttt{Parall-SS} outperforms the \texttt{Sequn-SS} with a smaller computational burden, which better takes advantage of the two-fold \textit{orthogonal} data characteristics.  The learnable feature fusion in \texttt{Parall-SS} enables embedding interaction in each block, bringing more flexibility than the \texttt{Sequn-SS}. 
\textbf{(5)} Notably, both spatial-spectral attention structures outperforms  the solely spatial/spectral mechanisms, indicating a better data clue exploitation of joint modeling. 
}

{\textbf{Sequential Spatial-Spectral Attention Analysis}. 
Fig.~\ref{fig: equn_embedding visal} presents the feature embedding of the sequential spatial-spectral attention blocks.  By observation, the spatial and spectral self-attentions are alternately performed, leading to blended modeling across spatial and spectral dimensions. Either of the attention lacks independence, and thereby the behaviours of the spatial and spectral are quite similar in some cases (1st, 5th, 7th column) of Fig.~\ref{fig: equn_embedding visal}. { Notably, we also notice some differences between two attentions -- the spectral attention seems to be good at abstracting the most discriminated features of an image (\emph{e.g.}, strong edges in $3$rd, $5$th, and $7$th columns). By comparison, spatial attention hardly depicts edges but highlights more fine-grained visual details. This motivate us to evaluate spatial/spectral attentions in a more independent manner, \emph{i.e.}, parallel arrangement.}
}

\begin{figure}[t]
\centering
\resizebox{1.04\columnwidth}{!}{
\begin{tabular}{c}\hspace{-6mm}
\begin{adjustbox}{valign=t}
    \begin{tabular}{c}
    \includegraphics[width=\columnwidth]{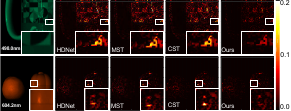}\\
    \vspace{2mm}\hspace{-1.01mm}
    \includegraphics[width=\columnwidth]{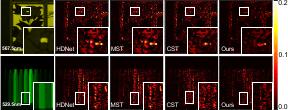}
    \end{tabular}
\end{adjustbox}
\end{tabular}}
\vspace{-4mm}
\caption{Reconstruction difficulty maps ($|\widehat{\mathbf{F}}$$-$$\mathbf{F}|_{\texttt{chl}}$) in \textbf{masked} areas. \textbf{Lower} intensity indicates a better reconstruction capacity -- the proposed method demonstrates the smallest difference as the ground truth on masked areas, regardless of the spatial textures and the spectral channels. See zoom-in patch for a better comparison.}
\label{fig: masked difficulty}
\vspace{-3mm}
\end{figure}

\begin{figure}[ht] 
\centering 
\includegraphics[width=\columnwidth]{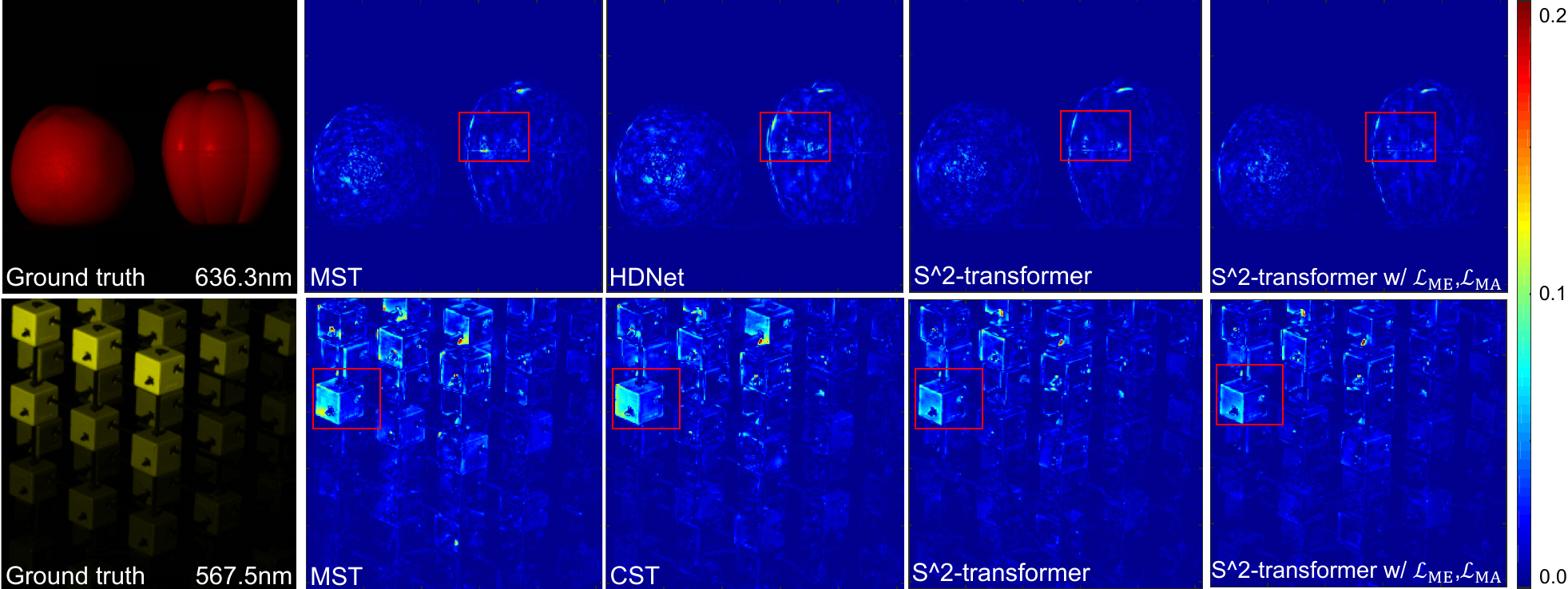}
\vspace{-4mm}
\caption{Reconstruction difficulty maps ($|\widehat{\mathbf{F}}-\mathbf{F}|_{\texttt{chl}}$) on simulation data. Lower intensity indicates a better retrieval.} 
\label{fig: global difficulty}
\vspace{-4mm}
\end{figure}

\begin{figure}[htp]
\centering
\includegraphics[width=0.9\columnwidth]{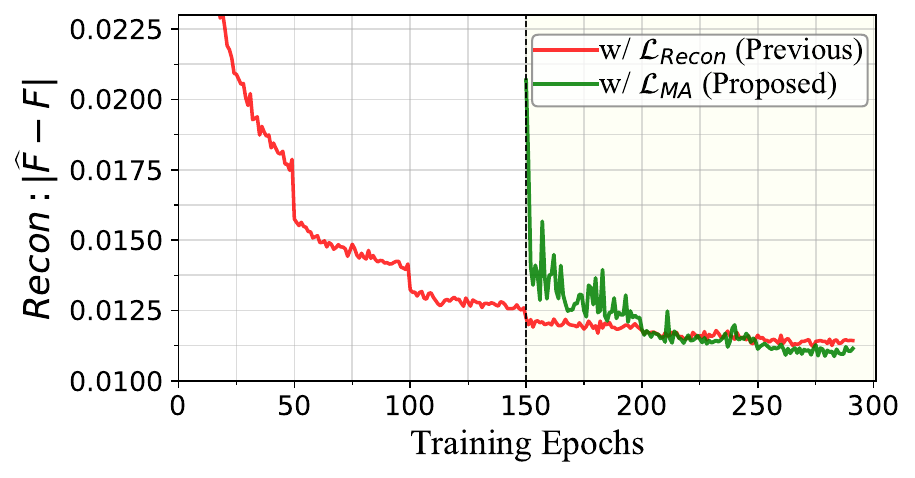} \vspace{-2mm}
\caption{Convergence analysis of the reconstruction. Previous strategy employs the \texttt{Recon} term, \emph{i.e.}, $\mathcal{L}_{\texttt{Recon}}$$=$$|\widehat{\mathbf{F}}-\mathbf{F}|$ (red curve). With the mask-aware learning, the reconstruction converges toward a lower scale (green curve). The $\mathcal{L}_\texttt{MA}$ follows 150-epoch pre-training.}
\label{fig: loss} 
\vspace{-4mm}
\end{figure}

\begin{figure}[t] 
\centering 
\includegraphics[width=0.3\textwidth]{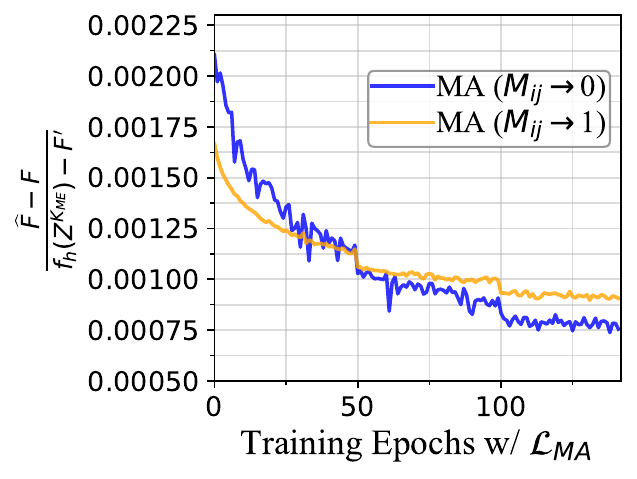}
\vspace{-3mm}
\caption{{
Loss curve of \texttt{MA} in the \textbf{second} training phase. \texttt{MA} gradually converges on both masked and unmasked regions. }
} \label{Fig: loss MA}
\vspace{-10mm}
\end{figure}

\textbf{Parallel Spatial-Spectral Attention Analysis}.
Following the above observations, we put emphasize on the \texttt{Parall-SS} block. First of all, we visualize the fused feature embedding within the attention blocks at different depth. There are total 24 blocks given $K$$=$$4$ and $L$$=$$6$. In Fig.~\ref{fig: fused_feature}, we visualize the feature maps of the embedding with random channel indices and depths. By observation, the representations in deeper blocks are closer to the visual semantics {and the deep blocks  present a similar abstracting ability (\emph{e.g.}, \texttt{blk 22} and \texttt{blk 24} of \texttt{Scene 10} in Fig.~\ref{fig: fused_feature}). Based on these blocks (\emph{i.e.}, 22$\sim$24 blocks), we visualize the exploit the behaviours of spatial and spectral attentions.} In Fig.~\ref{fig: edge_detector}, we present the feature embedding right after the $f_\texttt{Spa-MSA}(\cdot)$ and  $f_\texttt{Spe-MSA}(\cdot)$, respectively. { The evidences from different data indicate that  $f_\texttt{Spe-MSA}(\cdot)$ works like an \textit{edge detector}.
By comparison, the $f_\texttt{Spa-MSA}(\cdot)$ conversely collects the fine-grained visual details. These findings help explain in parallel spatial-spectral attention design, how both of the attention types contribute to the reconstruction, respectively.} This observation coincides with our motivations. (1) Spatial attention allows a fine-grained visual detail modeling by attending the total number of $H$$\times$$W$ tokens.
(2) The $f_\texttt{Spe-MSA}(\cdot)$ describes the 2D token relations by a scalar attention value regardless of the token spatial size, which indicates a heterogeneous spatial expressiveness across the 2D token plane, considering the finite modeling capacity. {It ultimately depends on the most discriminative regions, \emph{i.e.}, edges, to identify different feature channels. } 
Since the inconsistent representation space between the embedding ($C$$=$60) and the hyperspectral signal ($N{_\lambda}$$=$$28$), we directly observe the spectral distribution of the predictions upon \texttt{Parall-SS}. As shown in Fig.~\ref{fig: spectral_corr}, each image possesses a unique distribution. The proposed network retrieves the spectrum patterns properly.  
{Due to the better performance of the \texttt{Parall-SS} compared to \texttt{Sequn-SS}, we finalize the  network structure and conduct experiments in the following using this design.}

\subsection{Mask-aware Learning}
\label{subsec: mask-aware learning}

\textbf{Convergence Analysis.}
We firstly conduct an ablation study for the proposed mask-aware learning in Table~\ref{tab: loss_ablation}. 
We choose the $|\widehat{\mathbf{F}}-\mathbf{F}|$ as the \texttt{Previous} setting and all experiments follow the same optimization configurations (refer to Section~\ref{subsec: setting} for more details). By comparison, a complete mask-aware learning treatment enables 0.17dB/0.0015 improvement over the baseline. Besides, solely employing either $\mathcal{L}_{\texttt{ME} }$ or $\mathcal{L}_{\texttt{MA}}$ hardly benefits the reconstruction. (1) The mask-encoding (\texttt{ME}) term introduced in $\mathcal{L}_{\texttt{ME} }$ penalizes the shallow parts of the network. Its effect weakens as the network deepens. (2) The mask-aware (\texttt{MA}) term infers the masked regions upon the well-approximated $\widehat{\mathbf{F}}$, which is unattainable in early training steps. We further analyze the convergence of the reconstruction by directly observing the \texttt{Recon} term. As shown in Fig.~\ref{fig: loss}, the absolute error between the $\widehat{\mathbf{F}}$ and $\mathbf{F}$ is further minimized when applying \texttt{MA}. 
Since the $\mathcal{L}_{\texttt{MA}}$ prioritizes the masked region with higher penalty. {A more precise reconstruction accordingly contributes to a globally-averaged fidelity.} We perceptually show this below.

\begin{table*}[h] 
\footnotesize
\caption{Integrating mask-aware learning into other  methods. For SRN~\cite{wang2021new}, we adopt the first $k_{\texttt{ME}}$ residual blocks to compute the mask-encoded signal $\mathbf{F}'$, \emph{e.g.}, $||f_{h}(\mathbf{Z}^{k_\texttt{ME}}) - \mathbf{F}'||_1$, where  $f_h(\cdot)$ is implemented by an \texttt{LN-CONV} structure, being consistent with our method. Since MST~\cite{cai2022mask} is a U-shaped symmetric encoder-decoder network structure, we take the middle output embedding of the last bottleneck MSAB block to compute $\mathbf{Z}^{k_\texttt{ME}}$. }
\vspace{-1mm}
\label{tab: ma-SRN-MST}
\centering
\resizebox{.78\textwidth}{!}{
\centering
\begin{tabular}{c|cccc|cc} 
    \hline
    \multirow{2}{*}{Methods} & \multicolumn{4}{c|}{SRN~\cite{wang2021new}} & \multicolumn{2}{c}{MST~\cite{cai2022mask}}  \\
    \cline{2-7}
    & w/o mask aware & $k_{\texttt{ME}}=10$ & $k_{\texttt{ME}}=12$ & $k_{\texttt{ME}}=14$  & w/o mask aware & w/ mask aware \\
    \hline
    PSNR & 35.07  & 35.20 & 35.47  & 35.50 & 35.18 & 35.39 \\
    SSIM & 0.9430 & 0.9434 & 0.9479 & 0.9480 & 0.9476& 0.9524 \\
    \hline
\end{tabular}}
\end{table*}

\begin{table*}[htp] 
\caption{Performance under challenging scenarios. (1) \texttt{Non-smooth Spectra} attenuates the spectral correspondence by sampling far-between wavelengths. (2) \texttt{Poisson Noise} is imposed over all channels. $\lambda$ denotes the variance. }
\vspace{-1mm}
\label{Tab: challenging}
\centering
\resizebox{\textwidth}{!}{
\centering
\begin{tabular}{c|cc|cc|cc|cc|cc} 
\hline
\multirow{2}{*}{Methods} &  \multicolumn{2}{c|}{Non-Smooth Spectra}  & \multicolumn{2}{c|}{Poisson Noise ($\lambda=5$)} & \multicolumn{2}{c|}{Poisson Noise ($\lambda=10$)} & \multicolumn{2}{c|}{Poisson Noise ($\lambda=20$)} &  \multicolumn{2}{c}{Poisson Noise ($\lambda=50$)}  \\ \cline{2-11} & PSNR (dB) & SSIM & PSNR (dB) & SSIM & PSNR (dB) & SSIM & PSNR (dB) & SSIM  & PSNR (dB) & SSIM \\
\hline
MST~\cite{cai2022mask}    & 37.63  & 0.9772 & 35.32 & 0.9490 & 35.09  & 0.9457  & 33.17 & 0.9101  & 32.94 & 0.9089 \\
$S^2$-Transformer (\textbf{ours})    & \textbf{37.98} & \textbf{0.9796} & \textbf{36.50} & \textbf{0.9627} & \textbf{35.91} & \textbf{0.9504} & \textbf{35.34} & \textbf{0.9495} & \textbf{33.36} & \textbf{0.9145}  \\
\hline
\end{tabular}}
\vspace{-3mm}
\end{table*}

In Fig.~\ref{fig: loss}, we provide the loss tendency of the \texttt{Recon} term with or without the mask-aware learning, on the right, we also present the loss curve of the MA term in the second training phase. Specifically, we demonstrate the MA loss value only on the masked  region by the blue curve and the unmasked region by the orange curve, respectively. Both curves converge to a smaller value, which indicates an effective penalization toward the reconstruction of the masked/unmasked areas. Specifically, there is a more obvious degradation tendency of the MA term (blue curve) on the masked regions compared to the unmasked region, which means the MA term emphasizes more on the masked and uncertainty areas as expected (note that we choose the spatial patch where the masked and unmasked regions are of the similar spatial size). In unmasked regions, the MA term smoothly degrades, which empirically reveals that the proposed learning can benefit the reconstruction of unmasked (but noisy) regions, serving as more desirable evidence for than the theoretical induction.

\textbf{Data Uncertainty}.
We firstly compare the empirical difficulties among different methods. Spectral channels are selected at random and the masked regions are emphasized, \emph{i.e.}, $|\widehat{\mathbf{F}}-\mathbf{F}|$$\odot$$(\mathbf{1}-\mathbf{M})$, where $\odot$ denotes a pixel-wise multiplication. 
As shown by the Fig.~\ref{fig: masked difficulty}, { we empirically compare the reconstruction performance on different methods, especially on \textbf{textured, \emph{e.g.}, edges, and masked areas}. The lower intensity demonstrates a smaller pixel-wise absolute error and, therefore, a better reconstruction performance. The proposed method shows more accurate results.  Specifically, we select the same spatial size for the zoom-in and scale the patch with the same ratio. The only difference between two results is whether employing the mask-aware learning. } Identifying the masked areas is non-trivial, since the commonly used L1 or L2 norms hardly distinguish aleatoric uncertainty caused by the internal high-frequency textures and the external hardware encoding (\emph{i.e.}, masking). This leads to a common modeling behaviour. In Fig.~\ref{fig: global difficulty}, we further  visualize the reconstruction difficulties among compared methods and our ablated models. In line with the convergence curve in Fig.~\ref{fig: loss}, the $\mathcal{L}_{MA}$ enhances the reconstruction globally. Better approximation of the masked areas potentially benefits the neighbored pixel estimation owning to the spatial locality property of the model. We finally visualize both simulation and real predictions of the ablated model in Fig.~\ref{fig: mask aware visual result}. By the zoom-in windows, mask-aware learning contributes to a high-fidelity reconstruction.

\textbf{Backbone Discussion}. We further integrate the mask-aware learning into competitive methods of SRN~\cite{wang2021new} and MST~\cite{cai2022mask} in Table~\ref{tab: ma-SRN-MST}. For SRN, we adopt the first $k_{\texttt{ME}}$ (\emph{i.e.}, $1<k_{\texttt{ME}}<16$) residual blocks to compute the masked encoded signal $\textbf{F}'$, where where $\textbf{Z}^{k_\texttt{ME}}$ denotes the output embedding ($1<k_\texttt{ME}<16$) and $f_h(\cdot)$ is implemented by an \texttt{LN-CONV} structure, being consistent with our method. Since MST~\cite{cai2022mask} is a U-shaped symmetric encoder-decoder network structure, we take the output embedding of the last bottleneck MSAB block (in the middle of the whole U-shaped network) to compute $\mathbf{Z}^{k_\texttt{ME}}$. We firstly adopt the mask-encoding loss $\mathcal{L}_\texttt{ME}$ to train SRN for $100$ epochs and train MST for $80$ epochs. Notably, we keep the total training epochs of the mask-aware learning the same as the original settings of SRN and MST. Besides, we set the same $\alpha$ and $\beta$ schedule as the $S^2$-Transformer for both of the methods. As compared in Table~\ref{tab: ma-SRN-MST}, mask-aware learning brings $0.43$dB/$0.0050$ performance boost for SRN and $0.18$dB/$0.0048$ for MST in terms of PSNR/SSIM. Besides the transformer architecture (MST and the proposed $S^2$-Transformer), mask-aware learning also benefits the traditional convolution-based network (SRN).  Nevertheless, mask-aware learning brings clear performance boost regardless of the reconstruction backbone.

{
\subsection{Challenging Scenarios}
\label{subsec: challenge}
\textbf{Non-smooth Spectra}. 
We train and test the reconstruction models upon simulation data with non-smooth spectral profiles. Specifically, we zero out the two spectral channels every two channels per interval  
    \begin{equation}
        \begin{aligned}
            \mathbf{F}(:,:,n_{\lambda}) = 
            \begin{cases}
                \mathbf{0},&  \ n_{\lambda} \% 4 < 2, \\
                \mathbf{F}(:,:,n_{\lambda}),& \texttt{else},
            \end{cases} 
        \end{aligned}
    \end{equation}
    where $n_{\lambda}={0,...,27}$. Since a part of the spectral channels is masked, it is more challenging to model the long-range dependencies across spectra and infer the underlying correlations.     We compare the performance of the proposed method (\emph{i.e.}, $S^2$-Transformer with paralleled spatial-spectral attention) with the MST~\cite{cai2022mask} in Table~\ref{Tab: challenging}.     There is a larger performance gap between the two methods than the benchmark performance in Table~1 and 2. Masked spectra impose higher uncertainty for the reconstruction and break the intrinsic dependencies among adjacent wavelengths, which hinders the spectral long-range dependency modeling of MST~\cite{cai2022mask}. By comparison, the proposed method well disentangles the spatial-spectral data clues by paralleled attention design, which allows a better representation of learning in the spatial and spectral domains, respectively.
}

{
\textbf{Poisson Noise}.
We also consider the hyperspectral reconstruction under a more challenging scenario of extremely noisy conditions. 
We simulate the scenario by injecting the Poisson noise into both training and testing data.     Specifically, for each pixel at spatial location of $(i,j)$, the intensity is 
    \begin{equation}
        P(k_{(i,j)})=\frac{\lambda^{k_{(i,j)}}}{k_{(i,j)}!}e^{-\lambda},
    \end{equation}
    {where $k_{(i,j)}$ denotes the illumination of the pixel. We choose four representative variance  values, \emph{i.e.}, $\lambda=\{5, 10, 20, 50\}$.  The amount of noise added to the spectral channels increases as the lambda value grows. Table~\ref{Tab: challenging} shows that (1) the proposed $S^2$-Transformer consistently outperforms MST~\cite{cai2022mask} under different $\lambda$ values. (2) The performance gaps between  $S^2$-Transformer and MST can be significant under highly noisy scenario of $\lambda=20$, \emph{e.g.}, $2.17$dB/0.0394 in terms of PSNR/SSIM. (3) In the extremely noisy case of $\lambda=50$, both MST and the proposed method see a degraded performance, but the performance gap between them is still clear, \emph{i.e.}, $0.42$dB/$0.0056$ in terms of PSNR/SSIM. In summary, $S^2$-Transformer is less prone to the Poisson noise injection and non-smooth spectra. Different from the MST that only highlights the spectral dependency, $S^2$-Transformer jointly considers the information disentangling by spatial-spectral attention and emphasizes masked region by mask-aware learning.  Consistent with the above observation, $S^2$-Transformer perceptually demonstrates better results under the real-captured HSI (Fig.~\ref{Fig: real_result}), where the real measurement noise is inherently introduced. Overall, the performance of the proposed method can be better explained from the physical scope (optical encoding procedure), yielding more meaningful, trustworthy, and promising results. For example, the proposed $S^2$-Transformer can bring a remarkable performance boost on challenging scenarios. Besides, mask-aware learning shows promising performance on competitive methods of SRN~\cite{wang2021new} and MST~\cite{cai2022mask}.}
}

\begin{table}[tp] 
\caption{Experiment using remotely sensed geological data from AVIRIS and SCI simulation collection process. }
\vspace{-1mm}
\label{Tab: AVIRIS}
\centering
\resizebox{.45\textwidth}{!}{
\centering
\begin{tabular}{l|rrr} 
\hline
Methods & SRN~\cite{wang2021new} & MST~\cite{cai2022mask} & $S^2$-Transformer (Ours)  \\
\hline
PSNR (dB) & 34.95 & 36.66 & \textbf{37.07}\\
SSIM   & 0.9545 & 0.9589 & \textbf{0.9608}\\
\hline
\end{tabular}}
\vspace{-1mm}
\end{table}

\begin{table}[t] 
\footnotesize
\caption{Performance comparison between the proposed $S^2$-Transformer and recent state-of-the-art methods by PSNR (dB) and SSIM. BIRNAT~\cite{cheng2020birnat} employs an end-to-end bidirectional RNN architecture. DAUHST~\cite{cai2022degradation} and PADUT~\cite{li2023pixel} are deep-unfolding methods. }
\vspace{-1mm}
\label{tab: compare_sota}
\centering
\resizebox{0.49\textwidth}{!}{
\centering
\begin{tabular}{lc|cc} 
    \hline
    \multicolumn{2}{l|}{Methods}  & PSNR (dB) & SSIM \\
    \hline
    \multicolumn{2}{l|}{BIRNAT~\cite{cheng2020birnat}} 	& 37.58	& 0.960\\  
    \hline
    \multirow{2}{*}{PADUT~\cite{li2023pixel}} & 3-stage  & 36.95 & 0.962 \\ 
    & 12-stage   & 38.89 & 0.974 \\ 
    \hline
    \multirow{2}{*}{DAUHST~\cite{cai2022degradation}} 
    & 3-stage  & 37.21 & 0.959 \\ 
    & 9-stage  & 38.36 & 0.967 \\
    \hline 
    \cellcolor{cGrey} & \cellcolor{cGrey}3-stage & \cellcolor{cGrey}37.50 & \cellcolor{cGrey}0.962\\
     \multirow{-2}{*}{\cellcolor{cGrey}$S^2$-Trans-Unfolding (Ours) }  & \cellcolor{cGrey}9-stage  & \cellcolor{cGrey}38.74 & \cellcolor{cGrey}0.970\\
    \hline
\end{tabular}}
\vspace{-0.2cm}
\end{table}

\begin{table}[t] 
\footnotesize
\caption{Performance comparison between proposed method and recent methods on masked regions, \emph{i.e.}, scale both the predictions with the mask values $(\mathbf{1}-\mathbf{M})\widehat{\mathbf{F}}$ and ground truth $(\mathbf{1}-\mathbf{M}){\mathbf{F}}$. 
We compare the performance with averaged MSE. The proposed method demonstrates more accurate reconstruction performance on masked regions.}
\vspace{-1mm}
\label{tab: masked performance}
\centering
\resizebox{.43\textwidth}{!}{
\centering
\begin{tabular}{l|c} 
    \hline
    {Methods} &  Averaged MSE ($\times10^{-4}$) \\
    \hline
    {BIRNAT~\cite{cheng2020birnat}} 	& 2.65 \\  
    DAUHST~\cite{cai2022degradation} (3-stage)  
   &  2.60 \\ 
    PADUT~\cite{li2023pixel} (3-stage)   & 2.61\\ 
    $S^2$-Transformer (Ours)  &  2.44 \\ 
    \hline
\end{tabular}}
\end{table}

\subsection{Performance on AVIRIS data}\label{sec: AVIRIS}
{Besides the widely used HSI datasets of CAVE~\cite{yasuma2010generalized} and KAIST~\cite{choi2017high}. This work also collects a remotely sensed geological dataset captured by airborne visible/infrared imaging spectrometer (AVIRIS) instrument.  We simulate the SCI compression process and perform reconstruction using the proposed method. Specifically, we collect the Hyperspectral Infrared Imager data product provided by Jet Propulsion Laboratory\footnote{https://aviris.jpl.nasa.gov/index.html}. We collect $50$ data products captured with the spatial resolution of 60 meters (\texttt{hpc60}). For each data product, we random sample $256\times256$ patches during training for the data augmentation. The original data has $224$ bands ranging from $365$nm$\sim2496$nm. We select $28$ consecutive bands ranging from $404.6$nm$\sim667.6$nm, being consistent with the original spectrum range of simulated HSI data. We randomly select $5$ data products and crop $50$ non-overlapped patches as the testing dataset.  We simulate the compression process with the same coded aperture. The AVIRIS data serves as the ground truth for the quantitative computation. We train the proposed $S^2$-Transformer and MST~\cite{cai2022mask} for $150$ epochs and keep the other settings the same as Section~{\color{red}4.1}.  In Table~\ref{Tab: AVIRIS}, we compare reconstruction performance between different methods. The proposed $S^2$-Transformer outperforms MST by $0.41$dB/$0.0019$ and SRN by $2.12$dB/$0.0063$ in terms of the PSNR/SSIM.  
Notably, the remotely sensed geological data product generally takes up a large memory cost (\emph{e.g.}, requiring $3.4$Gb to store a $2372\times1084\times224$ data). The technology of SCI may be a potential solution for data compression. Our experiment shows that AVIRIS data with a high spatial resolution (\emph{e.g.}, $60$m ) can be reconstructed by SCI methods with high fidelity.
}

\begin{figure}[t] 
\centering 
\includegraphics[width=\columnwidth]{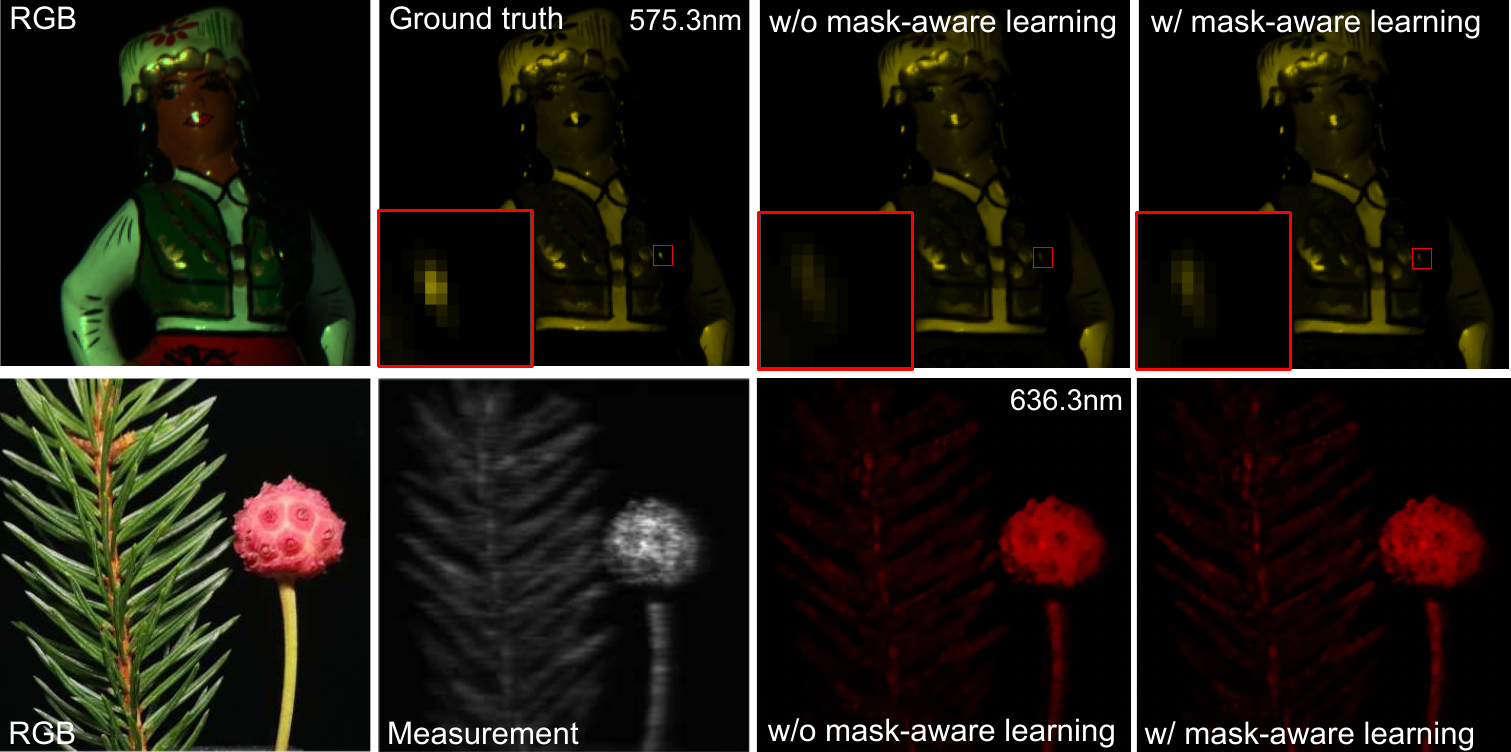}
\vspace{-5mm}
\caption{ Visual comparison of the proposed mask-aware learning. \textit{Top}: simulation data. \textit{Bottom}: real data. Note that the ground truth is unattainable in practice. Mask-aware learning enables more texture retrieval by comparison.}
\label{fig: mask aware visual result}
\vspace{-2mm}
\end{figure}

\begin{figure*}[t]
    \begin{minipage}{0.55\textwidth}
        \begin{table}[H] 
        \renewcommand\thetable{T\arabic{table}}
        \footnotesize
        \caption{Performance comparison between  $S^2$-Transformer and recent methods on specific samples, \emph{e.g.}, Scene 6 and Scene 8 (as shown on the right). The proposed method enables best PSNR and SSIM on these two cases, demonstrating that it can be particularly effective in specific spectral signatures and spatial patterns.}
        \vspace{-2mm}
        \label{tab: gt6 gt8}
        \centering
        \resizebox{\textwidth}{!}{
        \centering
        \begin{tabular}{lc|cc|cc} 
            \hline
            \multicolumn{2}{l|}{\multirow{2}{*}{Method}} & \multicolumn{2}{c|}{Scene 6}  & \multicolumn{2}{c}{Scene 8} \\
            \cline{3-6}
            & &  PSNR (dB) & SSIM & PSNR (dB) & SSIM \\ 
            \hline
            \multicolumn{2}{l|}{BIRNAT~\cite{cheng2020birnat}} & 35.30 & 0.959 & 33.96 & 0.956  \\  
            DAUHST~\cite{cai2022degradation} & (3-stage)  & 36.19 & 0.963 & 34.28 & 0.956  \\ 
            PADUT~\cite{li2023pixel} & (3-stage) & 35.58 & 0.965 & 33.76 & 0.960 \\ 
            \multicolumn{2}{l|}{$S^2$-Transformer (Ours)}  & \textbf{36.44} & \textbf{0.965} & \textbf{34.50} & \textbf{0.963}  \\ 
            \hline
        \end{tabular}}
        \vspace{-0.2cm}
        \end{table}
    \end{minipage}
    \hfill
    \begin{minipage}{0.2\textwidth}
    \vspace{6mm}
        \includegraphics[width=\textwidth]{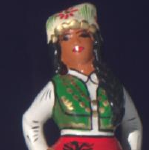}
        \centering
        {Scene 6 Ground Truth}
    \end{minipage}
    \hfill
    \begin{minipage}{0.2\textwidth}
        \vspace{5.4mm}
        \includegraphics[width=\textwidth]{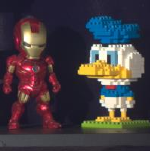}
        \centering\hspace{-4mm}
        {Scene 8 Ground Truth}
    \end{minipage}
\end{figure*}

\subsection{Comparing with Unfolding Methods}

{\textbf{Quantitative Comparison}. We further compare the proposed $S^2$-Transformer with recent state-of-the-art methods, including BIRNAT~\cite{cheng2020birnat}, DAUHST~\cite{cai2022degradation}, and PADUT~\cite{li2023pixel}. BIRNAT employs an end-to-end bidirectional recurrent network (RNN) architecture. Both DAUHST and PADUT are model-based deep unfolding methods that are scalable with different number of stages. 
As shown in Table~\ref{tab: compare_sota}.  The unfolding methods, \emph{e.g.}, DAUHST~\cite{cai2022degradation}, and PADUT~\cite{li2023pixel}, generally achieve state-of-the-art performances with a large number of stages, \emph{e.g.}, $9$-stage and $12$-stage, respectively.  }

{To better compare with above unfolding methods, we implement an $S^2$-Trans-Unfolding method by integrating the proposed transfomer backbone into the same unfolding framework of DAUHST. Table~\ref{tab: compare_sota} shows that the proposed $S^2$-Trans-Unfolding (3-stage) outperforms DAUHST (3-stage) by $0.29$dB/$0.003$ in PSNR/SSIM and PADUT (3-stage) by $0.55$dB in PSNR, respectively. The proposed $S^2$-Trans-Unfolding (9-stage)  improves  DAUHST (9-stage) by $0.38$dB/$0.003$ in terms of PSNR/SSIM and enables comparable performance with PADUT (12-stage).
Notably, both DAUHST~\cite{cai2022degradation} and PADUT~\cite{li2023pixel} have been published during the preparation and submission of this work. While the proposed $S^2$-Transformer has not yet been fully tailored into the unfolding framework, the current results show a great potential for incorporating the proposed $S^2$-Transformer in the unfolding framework in these methods. This is beyond the scope of our work here and we leave it for future work. }

{\textbf{Advantage Analysis}. The model-based methods integrates the deep network into the unfolding framework, which can be easily scaled by enlarging the number of stages. (as we shown in Table~\ref{tab: compare_sota}). 
Despite that above methods enables state-of-the-art performance, the proposed $S^2$-Transformer enjoys additional advantages. \textbf{First}, 
the proposed $S^2$-Transformer enables better reconstruction especially on masked regions compared with state-of-the-art methods, owning to the masked-aware learning. Table~\ref{tab: masked performance} shows that the proposed $S^2$-Transformer enables lower averaged MSE when solely focusing on the masked regions. This indicates that the proposed method can better restore the visual details that are missing caused by the CASSI mask encoding. 
\textbf{Second}, the proposed method can outperform state-of-the-art methods on specific spectral signatures and spatial patterns. In Table~\ref{tab: gt6 gt8}, we compare the performance of different methods on Scene 6 and Scene 8. The proposed $S^2$-Transformer enables highest PSNR and SSIM. Owning to the paralleled spatial-spectral attention design, $S^2$-Transformer can better preserve the sharp edges (also evidenced in figures in Table~\ref{tab: gt6 gt8}) and geometric structures (such as the complex boundaries in Lego pieces of duck in Scene 8). It can also be more sensitive to certain spectrum such as glowing colors (\emph{e.g.}, golden areas  shown in both  Scene 6 and 8). 
In summary, the proposed method not only studies the behaviour of spatial and spectral attentions, but also  
empirically and theocratically explores the impact of the mask-aware learning. Both of the designs can enhance the performance and the interpretability of the end-to-end network methods from the physics-driven perspective. Based on above insights, the proposed method provides a better chance to  bridge the performance gap between existing end-to-end methods and state-of-the-art model-based methods.}

\section{Conclusion}
\label{sec: conclusion}
In this work, 
we revealed two-fold data loss that impede a high-fidelity reconstruction of hyperspectral images by observing the physical encoding procedure of CASSI. For the entangled data loss, we resorted to exploit the hyperspectral image characteristics for a better information compensation upon the 2D measurement. We introduced the $S^2$-Transformer by systematically discussing different self-attention mechanisms. The empirical evidence demonstrated the different functionalities of both spatial and spectral attention under a paralleled arrangement.  For the masked data loss, we identify and take advantage of the optical-induced data uncertainty in a pixel-wise manner by prioritizing the masked pixels adaptively throughout the training. We theoretically and empirically presented the convergence tendencies of the proposed learning strategy upon both mask and unmasked regions. Our proposed method achieves superior performance compared with existing state-of-the-arts quantitatively and perceptually. Besides, this work serves as the first attempt to expedite the interpretability of designing deep reconstruction networks for HSI from the physical scope. We hope our discussions shed light on the future Transformer design for the HSI community.

\ifCLASSOPTIONcaptionsoff
  \newpage
\fi

% \appendices
% \section{PROOF OF COROLLARY 1}
% % \appendix[PROOF OF COROLLARY 1]
% \label{appendix: proof}
% \begin{proof}

% \end{proof}
% or


% \appendices
% \section{Appendix}

% \section{}

% % use section* for acknowledgment
% \ifCLASSOPTIONcompsoc
%   % The Computer Society usually uses the plural form
%   \section*{Acknowledgments}
% \else
%   % regular IEEE prefers the singular form
% \fi


\ifCLASSOPTIONcaptionsoff
  \newpage
\fi



\bibliographystyle{IEEEtran}
\bibliography{main}


\begin{IEEEbiography}
[{\includegraphics[width=1in,height=1.25in,clip,keepaspectratio]{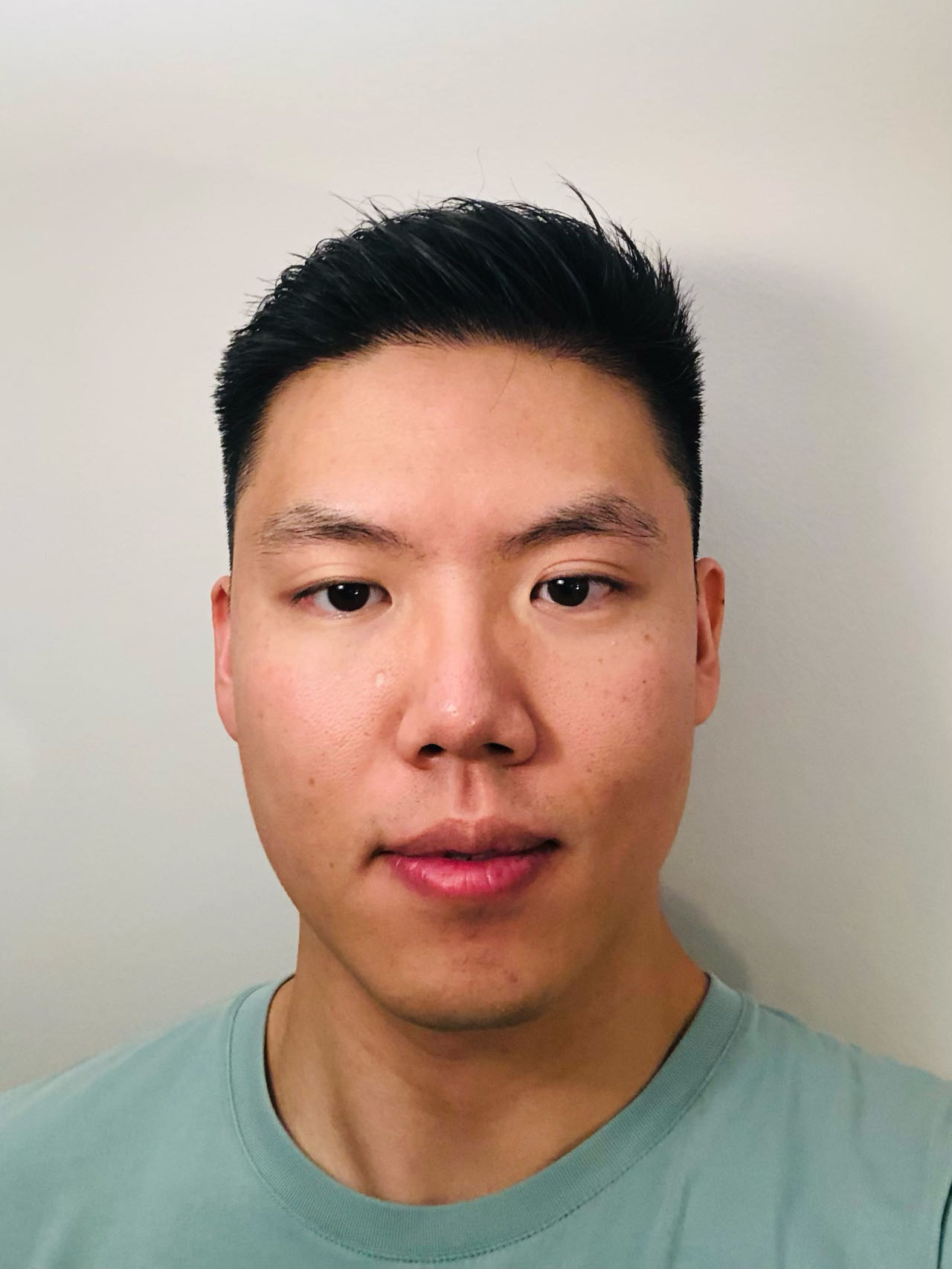}}]{Jiamian Wang}
received B.Eng. degree from electronic information engineering from Tianjin University, Tianjin, China, in 2018 and M.S. degree from University of Southern California, USA, in 2020. He is currently pursuing the Ph.D. degree with the Department of Computing and Information Sciences, Rochester Institute of Technology, USA. His research focuses on fine-grained visual understanding, with an emphasis on modeling visual uncertainty to enhance vision-language models and visual reconstruction tasks.
\end{IEEEbiography}

\vspace{-10mm}

\begin{IEEEbiography}[{\includegraphics[width=1in,height=1.25in,clip,keepaspectratio]{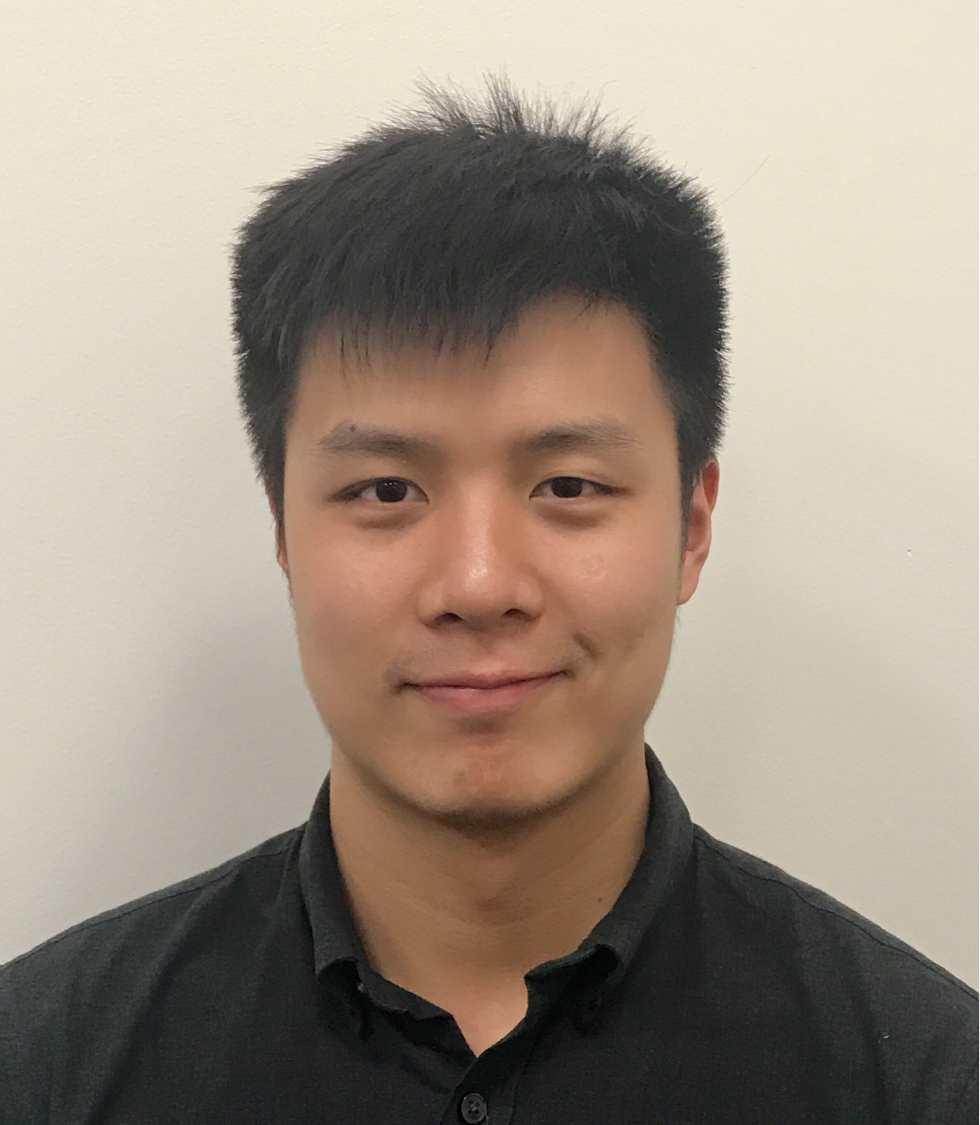}}]{Kunpeng Li} (S'18-M'21) received the B.Eng. degree in Information Engineering from South China University of Technology, China and Ph.D. degree in Computer Engineering, Northeastern University, Boston, MA. He is currently a Research Scientist at Meta's GenAI org, CA, USA. He has also spent time at Google Research, Adobe Research as a research intern. His research interests include learning with limited supervision, image and video generation foundation models, and multimodal learning.  
\end{IEEEbiography}

\vspace{-10mm}

\begin{IEEEbiography}[{\includegraphics[width=1in,height=1.25in,clip,keepaspectratio]{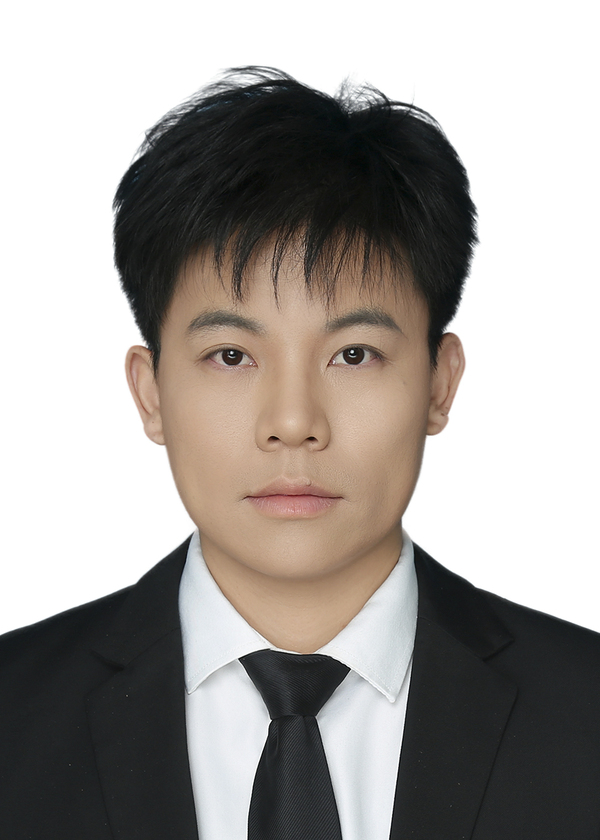}}]{Yulun Zhang} is a tenure-track associate professor at Shanghai Jiao Tong University. He was a postdoctoral researcher at Computer Vision Lab, ETH Zürich, Switzerland. He obtained the Ph.D. degree from the Department of ECE, Northeastern University, USA, in 2021. He also worked as a research fellow in Harvard University. Before that, he received the B.E. degree from the School of Electronic Engineering, Xidian University, China, in 2013 and the M.E. degree from the Department of Automation, Tsinghua University, China, in 2017. He has served as an area chair for CVPR, ICCV, ECCV, NeurIPS, ICML, ICLR, ACM MM, and IJCAI. His research interests include image/video restoration and synthesis, model compression, large language model, multimodal learning, and biomedical image analysis.
\end{IEEEbiography}

\vspace{-10mm}

\begin{IEEEbiography}[{\includegraphics[width=1in,height=1.25in,clip,keepaspectratio]{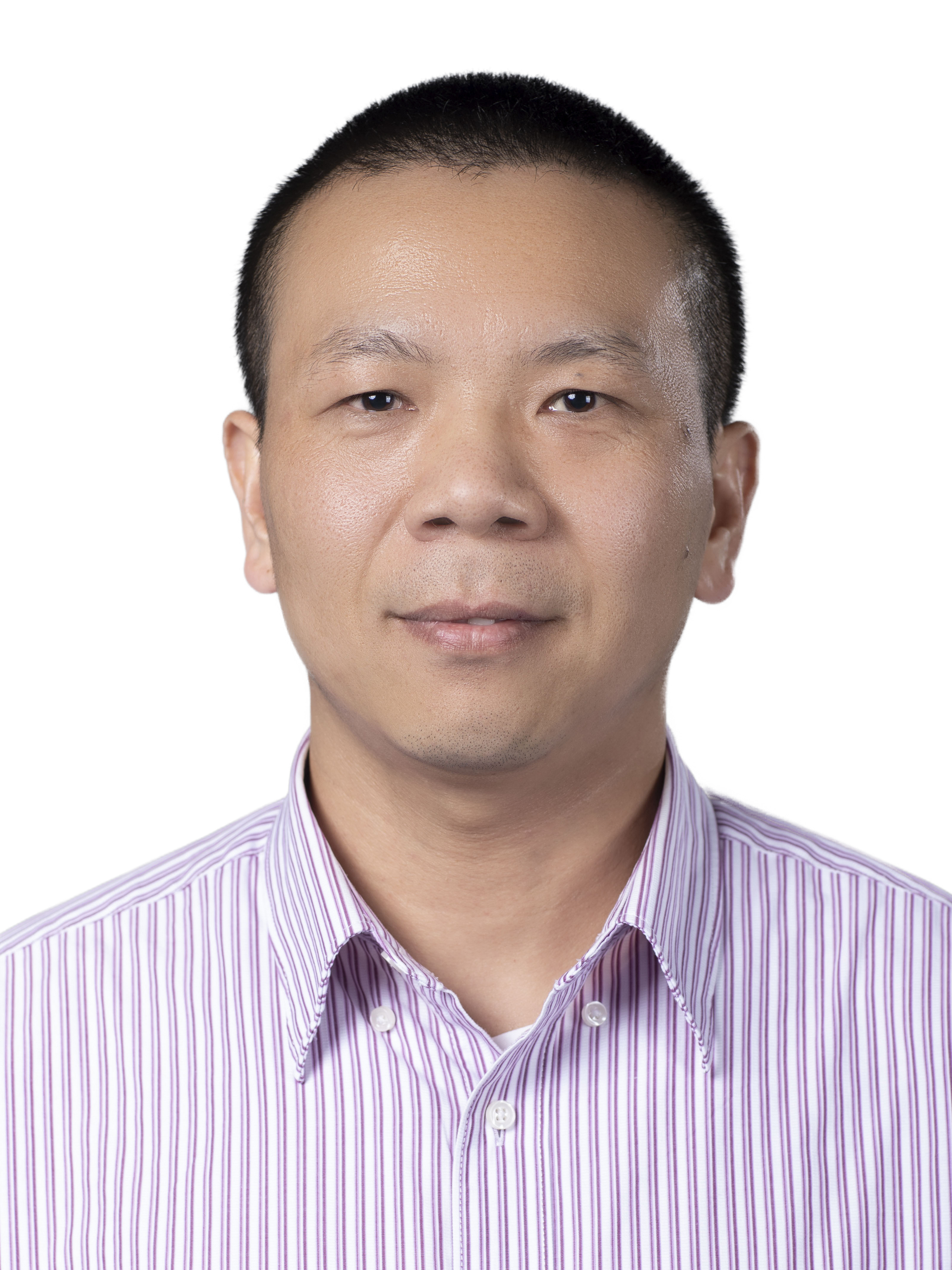}}]{Xin Yuan}
(SM'16) received the B.Eng. and M.Eng. degrees from Xidian University, in 2007 and 2009, respectively, and the Ph.D. from the Hong Kong Polytechnic University, in 2012. He is currently an Associate Professor at Westlake University. He was a video analysis and coding lead researcher at Bell Labs, Murray Hill, NJ, USA from 2015 to 2021. Prior to this, he was a Postdoctoral Associate in the Department of Electrical and Computer Engineering, Duke University from 2012 to 2015. His research interests are in signal processing, computational imaging and machine learning. He has been the receiving editor of Optics and Laser Technology since 2024, the Associate Editor of IEEE Transactions on Image Processing since 2025, Pattern Recognition since 2019, International Journal of Pattern Recognition and Artificial Intelligence since 2020, Chinese Optics Letters since 2021, and Advanced Imaging since 2024. He led the special issue of "Deep Learning for High Dimensional Sensing" in the IEEE Journal of Selected Topics in Signal Processing in 2022. He has been elected as the Optica Fellow in 2025 based on the seminal contributions to snapshot compressive computational imaging.
\end{IEEEbiography}

\vspace{-10mm}

\begin{IEEEbiography}[{\includegraphics[width=1in,height=1.25in,clip,keepaspectratio]{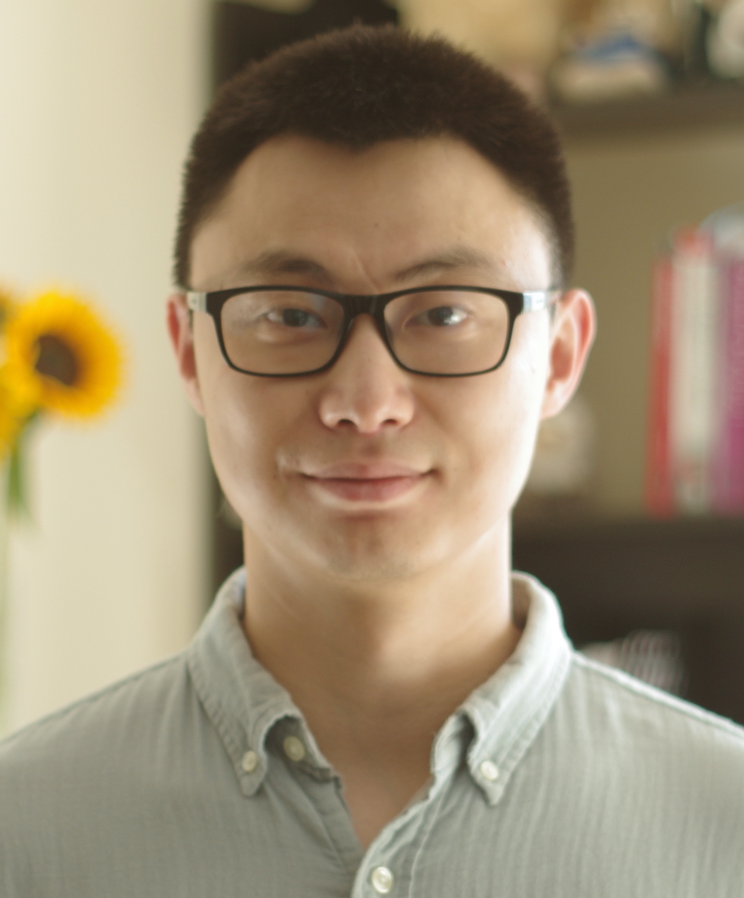}}]{Zhiqiang Tao} received the B.Eng. degree in software engineering from the School of Computer Software, and the M.Eng. degree in computer science from the School of Computer Science and Technology, Tianjin University, Tianjin, China, in 2012 and 2015, respectively. He received the Ph.D. degree from the Department of Electrical and Computer Engineering, Northeastern University, Boston MA, 2020. He has been a Tenure-Track Assistant Professor at the School of Information, Rochester Institute of Technology, since 2022. Before that, he was an Assistant Professor at the Department of Computer Science and Engineering, Santa Clara University, from 2020 to 2022. He has served as AEs of Neurocomputing (2022-) and IEEE TCSVT (2023-). His research interests include trustworthy machine learning, computational imaging, and large vision-language modeling. 
\end{IEEEbiography}


% % insert where needed to balance the two columns on the last page with
% %\newpage



%\vfill


\end{document}